\documentclass[11pt, letterpaper]{article}
\usepackage[utf8]{inputenc}
\usepackage{fullpage,times}

\usepackage{fullpage}

\usepackage{amsmath,amsfonts,amsthm,amssymb,multirow}
\usepackage{thm-restate}
\usepackage{times}
\usepackage{graphicx}
\usepackage{floatpag}
\usepackage{algorithm}
\usepackage[noend]{algpseudocode}
\usepackage{bbold}
\usepackage{enumitem}
\usepackage{subcaption} 
\usepackage{calc}
\usepackage{tikz}
\usepackage[countmax]{subfloat}
\usetikzlibrary{decorations.markings,decorations.pathreplacing,arrows, automata, positioning, quotes, shapes,backgrounds, arrows.meta}
\tikzstyle{vertex}=[circle, draw, inner sep=0pt, minimum size=4pt, fill = black]

\usepackage{graphicx}
\usepackage{tabularx}
\usepackage[hidelinks]{hyperref}
\makeatletter
\newcommand{\multiline}[1]{%
  \begin{tabularx}{\dimexpr\linewidth-\ALG@thistlm}[t]{@{}X@{}}
    #1
  \end{tabularx}
}
\makeatother
\usepackage{mathtools}

\makeatletter
\def\BState{\State\hskip-\ALG@thistlm}
\makeatother

\usepackage[compact]{titlesec}
\titlespacing{\section}{0pt}{3ex}{2ex}
\titlespacing{\subsection}{0pt}{2ex}{1ex}
\titlespacing{\subsubsection}{0pt}{0.5ex}{0ex}

\newtheorem{theorem}{Theorem}[section]

\newtheorem{definition}[theorem]{Definition}
\newtheorem{lemma}[theorem]{Lemma}
\newtheorem{claim}[theorem]{Claim}

\title{Algorithms and Lower Bounds for Replacement Paths under Multiple Edge Failures}

 \author{Virginia {Vassilevska Williams}\thanks{Supported by an NSF CAREER Award, NSF Grant CCF-2129139, a Google Research Fellowship and a Sloan Research Fellowship.}\\MIT\\virgi@mit.edu \and Eyob Woldeghebriel\\MIT\\eyobw@mit.edu \and Yinzhan Xu\thanks{Supported by an NSF CAREER Award and NSF Grant CCF-2129139.}\\MIT\\xyzhan@mit.edu}

\usepackage{mathtools}

\makeatletter
\let\c@fconjecture\c@conjecture
\makeatother

\makeatletter
\let\c@fconj\c@conj
\makeatother

\def \eps {\varepsilon}

\newcommand{\ignore}[1]{}

\def \polylog { \text{\rm polylog} }

\def\tO{\tilde{O}}

\newcommand{\mm}{n^{\omega}}
\newcommand{\sTRP}{sTRP}

\newcommand{\IGNORE}[1]{}

\begin{document}
\date{}

\maketitle
\pagenumbering{gobble} 
\begin{abstract}
This paper considers a natural fault-tolerant shortest paths problem: for some constant integer $f$, given a directed weighted graph with no negative cycles and two fixed vertices $s$ and $t$, compute (either explicitly or implicitly) for every tuple of $f$  edges, the distance from $s$ to $t$ if these edges fail. We call this problem $f$-Fault Replacement Paths ($f$FRP). 

We first present an $\tilde{O}(n^3)$ time algorithm for $2$FRP in $n$-vertex directed graphs with arbitrary edge weights and no negative cycles. As $2$FRP is a generalization of the well-studied Replacement Paths problem (RP) that asks for the distances between $s$ and $t$ for any single edge failure, $2$FRP is at least as hard as RP. Since RP in graphs with arbitrary weights is equivalent in a fine-grained sense to All-Pairs Shortest Paths (APSP) [Vassilevska Williams and Williams FOCS'10, J.~ACM'18], $2$FRP is at least as hard as APSP, and thus a substantially subcubic time algorithm in the number of vertices for $2$FRP would be a breakthrough. Therefore, our algorithm in $\tilde{O}(n^3)$ time is conditionally nearly optimal. Our algorithm immediately implies an $\tilde{O}(n^{f+1})$ time algorithm for the more general $f$FRP problem, giving the first improvement over the straightforward $O(n^{f+2})$ time algorithm. 

Then we focus on the restriction of $2$FRP to graphs with small integer weights bounded by $M$ in absolute values. We show that similar to RP, $2$FRP has a substantially subcubic time algorithm for small enough $M$. Using the current best algorithms for rectangular matrix multiplication, we obtain a randomized algorithm that runs in $\tilde{O}(M^{2/3}n^{2.9153})$ time. This immediately implies an improvement over our $\tilde{O}(n^{f+1})$ time arbitrary weight algorithm for all $f>1$. We also present a data structure variant of the algorithm that can trade off pre-processing and query time. In addition to the algebraic algorithms, we also give an $n^{8/3-o(1)}$ conditional lower bound for combinatorial $2$FRP algorithms in directed unweighted graphs, and more generally, combinatorial lower bounds for the data structure version of $f$FRP. 

\end{abstract}

\newpage
\pagenumbering{arabic}
\section{Introduction}

Shortest paths problems are among the most basic problems in graph algorithms, and computer science in general. An important practically motivated version considers shortest paths computation in failure-prone graphs. The simplest such problem is the {\em Replacement Path} (RP) problem (studied e.g. by \cite{malik1989k,nardelli2001faster,gotthilf2009improved,bernstein2010nearly,williams2018subcubic,emek2010near,klein2010shortest,bhosle2005improved,lee2014replacement,roditty2005replacement,weimann2010replacement,williams2011faster,GrandoniWilliamsSingleFailureDSO}) in which one is given a graph and two vertices $s$ and $t$ and one needs to return for every edge $e$, the shortest path from $s$ to $t$ in case edge $e$ fails.

RP has several motivations. The first is, preparing for roadblocks or bad traffic in road networks, and similar situations in which edge links are no longer available. The second motivation is in Vickrey pricing for shortest paths in mechanism design (see \cite{nisan2001algorithmic,VickreyPricingShortestPaths}).

It is natural to consider the generalization of RP to the case where up to $f$ edges can fail for $f>1$: given a graph $G$ and two vertices $s$ and $t$, for every set $F$ of up to $f$ failed edges, determine the distance between $s$ and $t$ in $G$ with $F$ removed. Let us call this the $f$FRP problem. Then similar to RP, $f$FRP is well-motivated: more than one roadblock can occur in road networks, and more than one link in a computer network can fail.

Since there are $\Theta(m^f)$ sets of $f$ edges in an $m$ edge graph, intuitively, the output of $f$FRP would have to be at least of size $\Theta(n^{2f})$ in a dense graph. However, just as in RP, one can show that the output only needs to be $\Theta(n^f)$. 

Consider for instance $f=2$. For every pair of edges $(e,e')$, at least one of them, say $e$ should be on the shortest path $P$ between $s$ and $t$, as otherwise the $s$-$t$ distance in $G\setminus F$ would be the same as that in $G$. Thus there are only $\leq n-1$ choices for $e$. 
Similarly, $e'$ should be on the shortest path between $s$ and $t$ in $G \setminus \{e\}$, and there are only $\leq n - 1$ choices for $e'$. 
Thus the output size of $2$FRP is only $\Theta(n^2)$, and by a similar argument, 
the output size of $f$FRP is $\Theta(n^f)$.

It is natural to ask how close to $n^f$ the best possible running time for $f$FRP can be. For graphs with nonnegative weights, $f$FRP can always be solved in $O(n^{f+2})$ time as follows. First, using Dijkstra's algorithm, compute in $O(n^2)$ time the shortest $s$-$t$ path $P$. Then for every edge $e$ on $P$, recursively solve $(f-1)$FRP in $G\setminus \{e\}$, where $0$FRP is just a shortest paths computation.

\begin{center}{\em
Can we do better than $O(n^{f+2})$ time for $f$FRP in directed weighted graphs?
}\end{center}

Let us consider RP (i.e. $1$FRP).
For graphs with $m$ edges, $n$ vertices, and arbitrary edge weights, RP can be solved in $\tO(m)$ time\footnote{The $\tO$ notation in this paper hides subpolynomial factors.} in undirected graphs~\cite{nardelli2001faster} and  can be solved in $O(mn + n^2\log\log n)$ time in directed graphs \cite{gotthilf2009improved}. In dense directed graphs RP with arbitrary edge weights is subcubically equivalent to All-Pairs Shortest Paths (APSP)~\cite{vw10, williams2018subcubic}, whose runtime has remained essentially cubic for 70 years, except for $n^{o(1)}$ factors.

Thus, (directed) $f$FRP for $f=1$ does not have an $O(n^{f+2-\eps})$ time algorithm for $\eps > 0$, under the APSP hypothesis. Furthermore, there is no known $O(n^{f+2-\eps})$ time algorithm for any $f>1$.

The special case of $f=2$, $2$FRP, is particularly interesting since its output has size only $O(n^2)$.
The APSP-based lower bound for RP implies that $2$FRP also requires $n^{3-o(1)}$ time to solve. However, there is a gap between the best known upper bound of $O(n^4)$ and the $n^{3-o(1)}$ lower bound.

This brings us to the following open question,  stated for instance in \cite{bhosle2004replacement}\footnote{Bhosle and Gonzalez~\cite{bhosle2004replacement}
 claimed an $O(n^3)$ time algorithm for the special case where both failed edges are on the original shortest path. 
 However, their approach doesn't quite work as written, and we include a discussion about the issue in Appendix~\ref{append:counter}.}.

\begin{center}
{\em Can $2$FRP be solved in essentially cubic time in directed weighted graphs?
}
\end{center}

Our first result is a resolution of the open question above for $2$FRP.

\begin{restatable}{theorem}{mainWeighted}
\label{thm:main_weighted}
In the Word-RAM Model with $O(\log n)$-bit words, the $2$FRP problem on $n$-vertex $O(\log n)$-bit integer weighted directed graphs with no negative cycles can be solved in $O(n^3 \log^2 n)$ time by a deterministic algorithm or in $O(n^3 \log n)$ time by a randomized algorithm that succeeds with high probability. 
\end{restatable}

Since even RP is known to require $n^{3-o(1)}$ time under the APSP hypothesis~\cite{vw10}, our algorithm for $2$FRP is conditionally tight, up to $n^{o(1)}$ factors.

Moreover, our theorem has an immediate corollary for $f$FRP for $f>2$ as well.

\begin{restatable}{corollary}{CorWeighted}
\label{cor:fFRP}
For all $f\geq 2$, $f$FRP in $n$-vertex directed weighted graphs with no negative cycles can be solved in $\tilde{O}(n^{f+1})$ time.
\end{restatable}

Thus, for all $f>1$, there is an algorithm for $f$FRP  that runs polynomially faster than the trivial $O(n^{f+2})$ time, even though for $f=1$ this was impossible under the APSP hypothesis!

For directed weighted graphs, replacement paths with vertex failures can be reduced to replacement paths with edge failures.\footnote{Given a graph $G=(V, E)$, we create a graph $G'$ as follows. For every $v \in V$, we create two vertices $v_{\text{in}}$ and $v_{\text{out}}$ in $G'$, and add an edge $(v_{\text{in}}, v_{\text{out}})$ with weight $0$ to $G'$. For every $(u, v) \in E$ with weight $w$, we add an edge $(u_{\text{out}}, v_{\text{in}})$ to $G'$ with weight $w$. Let $U \subseteq V$ be any set of $f$ failed vertices in $G$. We define $e_U$ to be $\{(u_{\text{in}}, u_{\text{out}}): u \in U\}$. For any $s, t \in V \setminus U$, it is not difficult to verify that $d_{G \setminus U}(s, t) = d_{G' \setminus e_U}(s_\text{in}, t_\text{out})$, which completes the reduction.}
Thus, Corollary~\ref{cor:fFRP} works even if $f$FRP is replaced with $f$ vertex-failure replacement paths. 

For unweighted graphs, or graphs with small integer edge weights, there are better known running times for RP in directed graphs. Here there is a distinction between \textit{algebraic} and \textit{combinatorial} algorithms.
 While the terms themselves are not well-defined, combinatorial algorithms usually refer to algorithms that do not use algebraic techniques such as  fast matrix multiplication while algebraic algorithms refer to algorithms that use such techniques. The study of combinatorial algorithms is motivated by the real-world inefficiency of fast matrix multiplication algorithms and by a desire to get algorithms that can perform better on sparser graphs, as it is usually difficult for algebraic algorithms to take full advantage of the sparsity of graphs.

In the case of unweighted graphs, directed RP has a combinatorial algorithm with an $\tO(m\sqrt{n})$ running time~\cite{roditty2005replacement}, which is essentially optimal as any further improvement would imply a breakthrough in Boolean Matrix Multiplication \cite{vw10, williams2018subcubic}. For graphs with small integer edge weights in the range $\{-M, \ldots, M\}$, there is an algebraic algorithm for directed RP with an $\tO(M\mm)$ running time \cite{williams2011faster, GrandoniWilliamsSingleFailureDSO}, where $\omega \in [2, 2.373)$ denotes the exponent for multiplying two $n \times n$ matrices \cite{AVW21, LeGall14, Vassilevska12}.

Can we solve $f$FRP in subcubic time in graphs with bounded integer weights when $f>1$? Due to the output size, $f$FRP for $f\geq 3$ cannot have a subcubic time algorithm. Thus the only generalization of RP that can still have a subcubic time algorithm in bounded weight graphs is $2$FRP. 

\begin{center}
{\em Does $2$FRP in graphs with small integer weights have a truly subcubic time algorithm?}
\end{center}

Similar to Corollary~\ref{cor:fFRP}, if $2$FRP in small weight graphs has an $O(n^{3-\eps})$ time algorithm, then for every $f\geq 2$, we would get an $O(n^{f+1-\eps})$ time algorithm for $f$FRP in small weight graphs, which is  close to the output size $\Theta(n^f)$ and always better than our new $\tO(n^{f+1})$ time algorithm for the arbitrary weight case.

A reason to believe that a subcubic time algorithm may be possible for $2$FRP is that the output size is only quadratic.
Another generalization of RP that has only quadratic output size, the Single Source Replacement Paths problem (SSRP), can be solved in $\tO(Mn^\omega)$ time for graphs with edge weights in $\{1,\ldots, M\}$~\cite{GrandoniWilliamsSingleFailureDSO} or in $O(M^{0.8043} n^{2.4957})$ time for graphs with edge weights in $\{-M, \ldots, M\}$ \cite{gu_et_al}. Hence it is possible that $2$FRP also has a subcubic time algorithm for small weight graphs.

Meanwhile, current techniques do not seem to yield subcubic time algorithms:
replacement paths problems have been studied (e.g. in \cite{WeimannYusterFDSO}) as a special case of $f$-failure distance sensitivity oracles (DSO), which are data structures that can support replacement path queries for any set of $f$ edge faults. To solve $2$FRP, one would need to pre-process a $2$-failure DSO and then perform $O(n^2)$ queries on it, assuming that the $O(n^2)$ queries are known. 

The best known DSO for graphs with small integer edge weights and more than one fault is by van den Brand and Saranurak~\cite{brandBatchUpdates}. For every $\alpha \in [0, 1]$, their oracle can achieve $\tO(Mn^{\omega + (3-\omega)\alpha})$ pre-processing time and $\tO(Mn^{2 - \alpha})$ query time for any constant number of failures on $n$-vertex graphs with edge weights in $\{-M, \ldots, M\}$. 
Balancing the pre-processing time and the $O(n^2)$ queries needed results in an $\tO(Mn^3)$ running time for $2$FRP -- a running time that is {\em never} subcubic, even if $M=O(1)$. 

We overcome the difficulties that come from using existing DSO techniques, and 
 are able to provide a new sensitivity oracle for $2$FRP with fast pre-processing and query times.

\begin{restatable}{theorem}{main}
\label{thm:main}
For any given positive integer parameter $g \leq O(n)$, there exists a data structure  that can pre-process a given directed graph $G$ with integer edge weights in $\{-M, \ldots, M\}$ and no negative cycles and fixed vertices $s$ and $t$, in
$\tO(M n^{\omega+1}/g + Mn^{2.8729})$ time, and can answer queries of the form $d_{G\setminus \{e_1,e_2\}}(s,t)$
in $\tO(g^2)$ time. This data structure has randomized pre-processing which succeeds with high probability. The size of the data structure is $\tO(n^{2.5})$.

If the edge weights of $G$ are positive, the pre-processing time and the size of the data structure can be improved to $\tO(M n^{\omega+1}/g +  M^{0.3544} n^{2.7778} + Mn^{2.5794})$ and $\tO(n^2)$ respectively. 
\end{restatable}

We note that the running time exponents shown above with $4$ digits after the decimal points all follow from fast rectangular matrix multiplication \cite{LU18}.  
As a corollary we obtain the first truly subcubic time algorithm for $2$FRP in graphs with bounded integer weights. 

\begin{restatable}{corollary}{mainBounded}
\label{cor: constant_query_time}
The $2$FRP problem on $n$-vertex directed graphs with integer edge weights in $\{-M,\ldots,M\}$ and with no negative cycles can be solved in $\tO(M^{2/3}n^{2.9153})$ time by a randomized algorithm that succeeds with high probability.

\end{restatable}

Our new algorithms for $2$FRP for bounded and arbitrary weight graphs are interesting as they show that $2$FRP is not much more difficult than RP-- both admit a cubic time algorithm for general graphs and subcubic time algorithms for bounded integer weight graphs. 

We also immediately obtain the following corollary for $f>2$, beating our algorithm for the arbitrary weight case for small enough $M$:

\begin{restatable}{corollary}{mainBoundedf}
\label{cor: constant_query_time_f}
For any $f\geq 2$, $f$FRP on $n$-vertex directed graphs with integer edge weights in $\{-M,\ldots,M\}$ and with no negative cycles can be solved in $\tO(M^{2/3}n^{f+0.9153})$ time by a randomized algorithm that succeeds with high probability. 
\end{restatable}

We remark that, if given the failed edges, all (except one) of our algorithms are able to report an optimal replacement path $P$ in $\tO(|P|)$ time. The exception is the positive weight case in Theorem~\ref{thm:main}, as it uses Gu and Ren's DSO \cite{gurenDSO} as a subroutine, which does not support path reporting. 

So far our algorithms have used fast matrix multiplication. Often one desires more practical combinatorial algorithms. How fast can combinatorial algorithms for $f$FRP in unweighted graphs be?

Since the output size for $f$FRP with $f>2$ is supercubic, only distance sensitivity oracles can give subcubic bounds. We show (conditionally) that any combinatorial $f$-failure sensitivity oracle for a fixed pair of vertices that can answer queries faster than running Dijkstra's algorithm at each query, must have high pre-processing time.

Our conditional lower bound is based on the Boolean Matrix Multiplication (BMM) hypothesis which says that any ``combinatorial'' algorithm for
 $n\times n$ Boolean matrix multiplication requires $n^{3-o(1)}$ time. By~\cite{vw10,williams2018subcubic}, the BMM hypothesis is equivalent to the hypothesis which states that any combinatorial algorithm for Triangle Detection, which asks whether a given graph contains a triangle, in $n$-vertex graphs requires $n^{3-o(1)}$ time.
\begin{restatable}{theorem}{lowerbound}
\label{thm:intro_lower_bound}
Let $k \ge 1$ be any constant integer. Suppose that there is a combinatorial data structure that can pre-process any
directed unweighted $n$-vertex graph $G$ and fixed vertices $s,t$ in $\tO(n^{2 + k/(k+1) - \epsilon})$ time, and can then answer $k$-fault distance sensitivity queries between $s$ and $t$ in $\tO(n^{2 - \epsilon})$ time, for $\epsilon > 0$. Then there is a combinatorial algorithm for Triangle Detection running in $\tO(n^{3 - \epsilon})$ time. 
\end{restatable}

 For $k=2$,  the above theorem implies that (combinatorial) $2$FRP requires $n^{8/3-o(1)}$ time under the BMM hypothesis.  This means that RP is slightly easier than $2$FRP in the combinatorial setting since it has an $\tilde{O}(n^{2.5})$ time algorithm~\cite{roditty2005replacement}.

We leave it as an open problem to obtain an $\tilde{O}(n^{8/3})$ time combinatorial algorithm for unweighted $2$FRP. One reason to suspect the existence of such an algorithm  is the existence of $2$-fault-tolerant BFS trees of size $O(n^{5/3})$~\cite{GuptaK17}. Note that for RP there are fault-tolerant BFS trees of size $O(n^{3/2})$ and a combinatorial algorithm of runtime $\tO(n^{5/2})$ \cite{ParterP16,roditty2005replacement}, though there is no direct reduction between algorithm running times and the sparsity of fault-tolerant subgraphs.

\subparagraph*{Related work.}

The running time for APSP in a graph with arbitrary $\polylog(n)$ bit integers has remained essentially cubic in the number of vertices for almost 70 years, with the current best running time being $n^3/ 2^{\Omega\left(\sqrt{\log n}\right)}$ by Williams~\cite{Williams14a, Williams18}. 
A truly subcubic time algorithm for APSP with arbitrary weights would be a significant breakthrough.

The restriction of APSP to graphs with small integer edge weights does have truly subcubic algorithms. Seidel gave an $\tO(n^\omega)$ time algorithm for APSP in an undirected unweighted graph \cite{seidel1995all}. This algorithm was later generalized by Shoshan and Zwick to yield an $\tO(M n^\omega)$ time algorithm for APSP in an undirected graph with integer edge weights in $\{0, 1, \ldots, M\}$ \cite{shoshan1999all}. For directed graphs with edge weights in $\{-M, \ldots, M\}$, the current best algorithm is by Zwick~\cite{Zwick02} that runs in $\tO(M^{1/(4-\omega)}n^{2+1/(4-\omega)})$ time, or $\tO(M^{0.7519}n^{2.5286})$ time using the current best algorithm for rectangular matrix multiplication~\cite{LeGall14}.

For graphs with $m$ edges, $n$ vertices and arbitrary integer weights, RP can be solved in $\tO(m)$ time in undirected graphs~\cite{nardelli2001faster}, and in $O(mn + n^2 \log \log n)$ time in directed graphs~\cite{gotthilf2009improved}. For graphs with small integer weights in $\{-M, \ldots, M\}$, Vassilevska Williams \cite{williams2011faster} showed an $\tO(Mn^\omega)$ time randomized algorithm for RP. For unweighted directed graphs, there is an $\tO(m\sqrt{n})$ time deterministic combinatorial algorithm for RP~\cite{alon2019deterministic}. 
RP has also been studied in the approximate setting~\cite{bernstein2010nearly}, in planar graphs~\cite{emek2010near, klein2010shortest} and in DAGs~\cite{bhosle2005improved, lee2014replacement}. 

In the more general Single-Source Replacement Paths (SSRP) problem, we are asked to compute all the replacement path distances for a single source $s$ but for all possible targets $t$ and all possible edges $e$ on one $s$ to $t$ shortest path. Grandoni and Vassilevska Williams~\cite{GrandoniWilliamsSingleFailureDSO,GrandoniW12} generalized the RP algorithm of Vassilevska Williams \cite{williams2011faster} to compute SSRP. Their randomized algorithm runs in $\tO(Mn^\omega)$ time for graphs with weights in $\{1, \ldots, M\}$ and in $\tO(M^{0.7519} n^{2.5286})$ time for graphs with weights in $\{-M, \ldots, M\}$. The latter case was recently improved by Gu, Polak, Vassilevska Williams and Xu to $O(M^{0.8043} n^{2.4957})$ time~\cite{gu_et_al}. For unweighted directed graphs with $n$ vertices and $m$ edges, Chechik and Magen~\cite{chechikSSRP} showed an $\tO(m\sqrt{n} + n^2)$ time combinatorial algorithm and an $mn^{1/2-o(1)}$ conditional lower bound for combinatorial algorithms.

There is a significant body of work on single-fault distance sensitivity oracle. The first nontrivial DSO for weighted graphs was given by  Demetrescu, Thorup, Chowdhury and Ramachandran \cite{demetrescu2008oracles}, who gave a deterministic oracle with constant query time and $O(n^{3.5})$ construction time. They also had an alternative DSO that needs $O(n^4)$ construction time, but only uses $\tO(n^2)$ space and keeps constant query time. 
Bernstein and Karger~\cite{bernstein2009nearly} gave a deterministic DSO for weighted graphs with $\tO(n^3)$ pre-processing time and $\tO(1)$ query time. This is essentially optimal barring improvements in APSP. On the other hand, single-fault DSOs for graphs with small integer edge weights do not have a (conditionally) optimal algorithm.   
The first DSO for small integer weighted graphs with subcubic pre-processing time and sublinear query time is given by Grandoni and Vassilevska Williams~\cite{GrandoniWilliamsSingleFailureDSO, GrandoniW12}. Their DSO for directed graphs with integer edge weights in $\{-M, \ldots, M\}$ has an $\tO(Mn^{\omega+1/2}+M n^{\omega+\alpha (4-\omega)})$ pre-processing time and an $\tO(n^{1-\alpha})$ query time for any parameter $\alpha \in [0, 1]$. Chechik and Cohen improved the DSO to  $\tO(Mn^{2.873})$ pre-processing time and $\tO(1)$ query time \cite{chechik2020distance}. An algorithm by Ren improves the pre-processing time to $\tO(M n^{2.733})$ and the query time to $O(1)$, but it only works for graphs with positive integer weights in $\{1, \ldots, M\}$~\cite{RenImprovedDSO}.
Gu and Ren \cite{gurenDSO} recently improved the pre-processing time to $\tO(M n^{2.5794})$ for constructing DSO for such graphs. 

The first major step in multiple-fault DSOs was a DSO by Weimann and Yuster \cite{WeimannYusterFDSO} which can efficiently handle up to $f = O(\log n/\log \log n)$ edge failures (for larger number of failures it will not be faster than brute-force). Their DSO is randomized, and has an $\tO(Mn^{\omega + 1 - \alpha})$ pre-processing time and an $\tO(n^{2 - (1-\alpha)/f})$ query time for graphs with weights in $\{-M, \ldots, M\}$, for any chosen parameter $\alpha \in (0, 1)$. 
They also have an alternative DSO which has  $\tO(M^{0.68}n^{3.529-\alpha})$ pre-processing time and $\tO(n^{2-2(1-\alpha) / f})$ query time for any chosen parameter $\alpha \in (0, 1)$ using the current best algorithm for rectangular matrix multiplication~\cite{LU18}. For graphs with arbitrary edge weights, their DSO  has an $\tO(n^{4 - \alpha})$ pre-processing time and an $\tO(n^{2 - 2(1-\alpha)/f})$ query time~\cite{WeimannYusterFDSO}. Their DSO for graphs with arbitrary edge weights was later derandomized by Alon, Chechik, and Cohen in \cite{alon2019deterministic}. The current best multiple-fault DSO for small integer weighted graphs is a randomized DSO by van den Brand and Saranurak, which has an $\tO(Mn^{\omega + (3-\omega)\alpha})$ pre-processing time and an $\tO(Mn^{2-\alpha}f^2 + Mnf^{\omega})$ query time, for any parameter $\alpha \in [0, 1]$ \cite{brandBatchUpdates}.

Duan and Pettie designed an $\tO(n^2)$ space  two-fault DSO with $\tO(1)$ query time~\cite{DuanPettieDualFailureDSO}. Since their focus is space complexity instead of pre-processing time, their result is not directly comparable to ours. Recently, Duan and Ren  \cite{duan2021maintaining} generalized \cite{DuanPettieDualFailureDSO} to $f$ failures: they designed an $\tO(fn^4)$ space $f$-fault DSO with $f^{O(f)}$ query time, though it only works for undirected weighted graphs.

\section{Preliminaries}
\label{sec:prelim}
Throughout this paper, use $\pi_G(u, v)$ to denote a shortest path from $u$ to $v$ in $G$ and use $d_G(u, v)$ to denote its distance. We also use $\pi_G(u, v, e)$ to denote a shortest path from $u$ to $v$ in the graph $G$ with edge $e$ removed, and use $d_G(u, v, e)$ to denote its distance. We sometimes drop the subscript $G$ if it is clear from the context. All graphs considered in this paper don't have negative cycles. 

Let $s, t$ be the source and target vertices in a graph $G$ and let $\pi_G(s, t)$ be a shortest path from $s$ to $t$ in $G$. Suppose we remove a set of edges $S$ from $G$. We say a path $P$ is \textit{canonical} (with respect to $\pi_G(s, t)$ and $S$) if for any vertices $u, v$ that appear both on $P$ and on $\pi_G(s, t)$ such that $u$ appears before $v$ in both $P$ and $\pi_G(s, t)$ and the subpath from $u$ to $v$ on $\pi_G(s, t)$ is not disconnected by $S$, then the subpath from $u$ to $v$ in $P$ is the same as the subpath from $u$ to $v$ in $\pi_G(s, t)$. Clearly, at least one of the replacement path $\pi_{G \setminus S}(s, t)$ is canonical. This is a light-weighted tie-breaking scheme. 

\begin{subfigures}
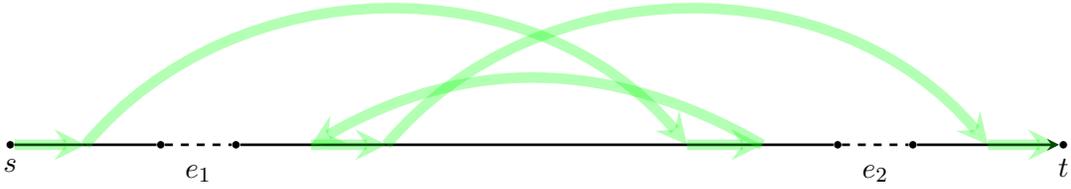
\begin{figure}[ht]
    \centering
    \begin{tikzpicture}
    	\node at (0,0) [circle,fill, inner sep = 1pt, label=below:$s$] (s){};
    	\node at (14,0) [circle,fill, inner sep = 1pt, label=below:$t$] (t){};
    	\node at (2,0) [circle,fill, inner sep = 1pt] (e1a){};
    	\node at (3,0) [circle,fill, inner sep = 1pt] (e1b){};
    	\node at (2.5,0) [label = below:$e_1$] (e1){};
    	
    	\node at (11,0) [circle,fill, inner sep = 1pt] (e2a){};
    	\node at (12,0) [circle,fill, inner sep = 1pt] (e2b){};
    	\node at (11.5,0) [label = below:$e_2$] (e2){};

	 \draw [-, line width = 1pt] (s) to[] (e1a);
    \draw [-, dashed, line width = 1pt] (e1a) to[] (e1b);
    \draw [-, line width = 1pt] (e1b) to[] (e2a);
    \draw [-, dashed, line width = 1pt] (e2a) to[] (e2b);
    \draw [-stealth, line width = 1pt] (e2b) to[] (t);

	\begin{scope}[transparency group, opacity=0.3, text opacity=1]
		\draw[-stealth,line width=4pt, green] (s) to[] (1, 0);
	    \draw[-stealth,line width=4pt, green, bend left = 50] (1,0) to[] (9, 0);
	    \draw[-stealth,line width=4pt, green] (9,0) to[] (10, 0);
		\draw[-stealth,line width=4pt, green, bend right = 30] (10,0) to[] (4, 0);
		\draw[-stealth,line width=4pt, green] (4, 0) to[] (5, 0);
		\draw[-stealth,line width=4pt, green, bend left = 50] (5, 0) to[] (13, 0);
		\draw[-stealth,line width=4pt, green] (13, 0) to[] (t);
	\end{scope}
	
    \end{tikzpicture}
    \caption{An example of a canonical path. The horizontal line (including $e_1$ and $e_2$) is the shortest path $\pi_G(s, t)$ from $s$ to $t$ in $G$, and the set of edges we remove is $S = \{e_1, e_2\}$. }
\end{figure}
\hfill
\begin{figure}[ht]
    \centering
    \begin{tikzpicture}
    	\node at (0,0) [circle,fill, inner sep = 1pt, label=below:$s$] (s){};
    	\node at (14,0) [circle,fill, inner sep = 1pt, label=below:$t$] (t){};
    	\node at (2,0) [circle,fill, inner sep = 1pt] (e1a){};
    	\node at (3,0) [circle,fill, inner sep = 1pt] (e1b){};
    	\node at (2.5,0) [label = below:$e_1$] (e1){};
    	
    	\node at (11,0) [circle,fill, inner sep = 1pt] (e2a){};
    	\node at (12,0) [circle,fill, inner sep = 1pt] (e2b){};
    	\node at (11.5,0) [label = below:$e_2$] (e2){};
    	\node at (5,0) [circle,fill, inner sep = 1pt, label=below:$u$] (u){};
    	\node at (9,0) [circle,fill, inner sep = 1pt, label=below:$v$] (v){};

	 \draw [-, line width = 1pt] (s) to[] (e1a);
    \draw [-, dashed, line width = 1pt] (e1a) to[] (e1b);
    \draw [-, line width = 1pt] (e1b) to[] (e2a);
    \draw [-, dashed, line width = 1pt] (e2a) to[] (e2b);
    \draw [-stealth, line width = 1pt] (e2b) to[] (t);

	\begin{scope}[transparency group, opacity=0.3, text opacity=1]
		\draw[-stealth,line width=4pt, red] (s) to[] (1, 0);
	    \draw[-stealth,line width=4pt, red, bend left = 60] (1,0) to[] (4, 0);
	    \draw[-stealth,line width=4pt, red] (4,0) to[] (5, 0);
		\draw[-stealth,line width=4pt, red, bend left = 50] (5,0) to[] (9, 0);
		\draw[-stealth,line width=4pt, red] (9, 0) to[] (10, 0);
		\draw[-stealth,line width=4pt, red, bend left = 60] (10, 0) to[] (13, 0);
		\draw[-stealth,line width=4pt, red] (13, 0) to[] (t);
	\end{scope}
	
    \end{tikzpicture}
    \caption{An example of a non-canonical path. The horizontal line (including $e_1$ and $e_2$) is the shortest path $\pi_G(s, t)$ from $s$ to $t$ in $G$, and the set of edges we remove is $S = \{e_1, e_2\}$. The path is not canonical because the subpath from $u$ to $v$ on $\pi_G(s, t)$ is not disconnected by $\{e_1, e_2\}$ while the path is not using that subpath.}
\end{figure}
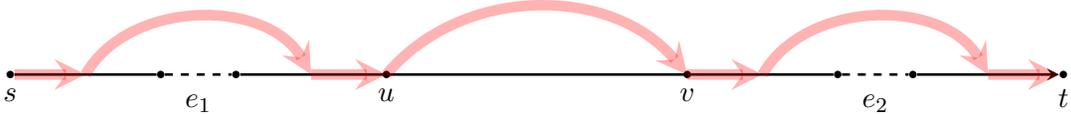
\end{subfigures}

For a graph $G$, we use $\widehat{G}$ to denote a copy of $G$ with all directions of the edges reversed. We use $\widehat{G}$ for notational succinctness. For instance, instead of saying computing single-target replacement paths to $t$ in $G$, we can say computing SSRP from $t$ in $\widehat{G}$. 

We use $\omega < 2.3729$ \cite{AVW21, LeGall14, Vassilevska12} to denote the square matrix multiplication exponent. For any $k > 0$, we also use $\omega(1, k, 1)$ to denote the exponent for multiplying an $n \times n^k$ matrix and an $n^k \times n$ matrix. Currently, the fastest algorithm for rectangular matrix multiplication is by Le Gall and Urrutia \cite{LU18}. It is well-known that the function $\omega(1, k, 1)$ with respect to $k$ is convex when $k > 0$ (see e.g. \cite{le2012faster}). 

\section{Technical Overview}
\label{sec:overview}

In this section, we describe the high-level ideas and key components in our algorithms for $2$FRP in arbitrary weighted graphs and in small integer weighted graphs. 

Let $G = (V, E)$ be a weighted graph with no negative cycles, and let $s, t \in V$ be the fixed source and target of the $2$FRP instance. Let $\pi_G(s, t)$ be a shortest path from $s$ to $t$. 
Both algorithms handle the following cases of two-fault replacement paths queries separately: the case where only one of the two failed edges $e_1, e_2$ is on the original shortest path $\pi_G(s, t)$, and the case where both failed edges are on $\pi_G(s, t)$. 

\begin{subfigures}
\begin{figure}[ht]
    \centering
    \begin{tikzpicture}
    	\node at (0,0) [circle,fill, inner sep = 1pt, label=below:$s$] (s){};
    	\node at (14,0) [circle,fill, inner sep = 1pt, label=below:$t$] (t){};
    	\node at (6.7,0) [circle,fill, inner sep = 1pt] (e1a){};
    	\node at (7.3,0) [circle,fill, inner sep = 1pt] (e1b){};
    	\node at (7,0) [label = below:$e_1$] (e1){};
	\node at (4, 0) [circle, fill, inner sep = 1pt, label = below:$u$](u){};
	\node at (10, 0) [circle, fill, inner sep = 1pt, label = below:$v$](v){};

	\draw [-, line width = 1pt] (s) to[] (e1a);
    \draw [-, dashed, line width = 1pt] (e1a) to[] (e1b);
    \draw [-stealth, line width = 1pt] (e1b) to[] (t);

	\begin{scope}[transparency group, opacity=0.3, text opacity=1]
		\draw[-stealth,line width=4pt, blue] (s) to[]  node(a)[]{} (u);
		\draw[-stealth,line width=4pt, blue, bend left=50] (u) to[] node(b)[]{} (v);		
		\draw[-stealth,line width=4pt, blue] (v) to[]  node(c)[]{} (t);

	\end{scope}
	
	\node at (a) [blue, yshift=10pt] {$\pi_G(s, u)$};
	\node at (b) [blue, yshift=10pt] {$\pi_{G \setminus \pi_G(s, t) \setminus \{e_2\}}(u, v)$};
	\node at (c) [blue, yshift=10pt] {$\pi_G(v, t)$};
    \end{tikzpicture}
    \caption{A typical $2$-fault replacement path where $e_1$ is on the original shortest path while $e_2$ is not. }
    \label{fig:overviewa}
\end{figure}
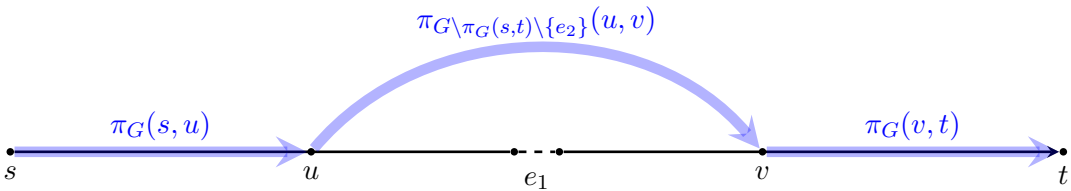
\hfill 
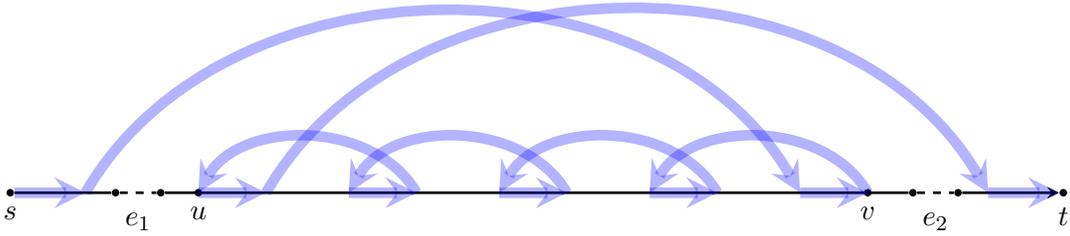
\begin{figure}[ht]
    \centering
    \begin{tikzpicture}
    	\node at (0,0) [circle,fill, inner sep = 1pt, label=below:$s$] (s){};
    	\node at (14,0) [circle,fill, inner sep = 1pt, label=below:$t$] (t){};
    	\node at (1.4,0) [circle,fill, inner sep = 1pt] (e1a){};
    	\node at (2,0) [circle,fill, inner sep = 1pt] (e1b){};
    	\node at (1.7,0) [label = below:$e_1$] (e1){};
    	
    	\node at (12,0) [circle,fill, inner sep = 1pt] (e2a){};
    	\node at (12.6,0) [circle,fill, inner sep = 1pt] (e2b){};
    	\node at (12.3,0) [label = below:$e_2$] (e2){};

	\node at (2.5, 0) [circle, fill, inner sep = 1pt, label = below:$u$](u){};
	\node at (11.4, 0) [circle, fill, inner sep = 1pt, label = below:$v$](v){};

	 \draw [-, line width = 1pt] (s) to[] (e1a);
    \draw [-, dashed, line width = 1pt] (e1a) to[] (e1b);
    \draw [-, line width = 1pt] (e1b) to[] (e2a);
    \draw [-, dashed, line width = 1pt] (e2a) to[] (e2b);
    \draw [-stealth, line width = 1pt] (e2b) to[] (t);

	\begin{scope}[transparency group, opacity=0.3, text opacity=1]
		\draw[-stealth,line width=4pt, blue] (s) to[] (1, 0);
	    \draw[-stealth,line width=4pt, blue, bend left = 60] (1,0) to[] (10.5, 0);
	    \draw[-stealth,line width=4pt, blue] (10.5,0) to[] (11.4, 0);
		\draw[-stealth,line width=4pt, blue, bend right = 60] (11.4, 0) to[] (8.5, 0);
		\draw[-stealth,line width=4pt, blue] (8.5, 0) to[] (9.4, 0);
		\draw[-stealth,line width=4pt, blue, bend right = 60] (9.4,0) to[] (6.5, 0);
		\draw[-stealth,line width=4pt, blue] (6.5, 0) to[] (7.4, 0);
		\draw[-stealth,line width=4pt, blue, bend right = 60] (7.4,0) to[] (4.5, 0);
		\draw[-stealth,line width=4pt, blue] (4.5, 0) to[] (5.4, 0);
		\draw[-stealth,line width=4pt, blue, bend right = 60] (5.4,0) to[] (2.5, 0);
		\draw[-stealth,line width=4pt, blue] (2.5, 0) to[] (3.4, 0);
		\draw[-stealth,line width=4pt, blue, bend left = 60] (3.4, 0) to[] (13, 0);
		\draw[-stealth,line width=4pt, blue] (13, 0) to[] (t);
	\end{scope}
	
    \end{tikzpicture}
    \caption{A typical $2$-fault replacement path where both $e_1$ and $e_2$ are on the original shortest path. All paths shown that do not lie on the original shortest path do not use any edge on it.}
    \label{fig:overviewb}
\end{figure}

\end{subfigures}

First, we consider the case where only one of the two failed edges $e_1, e_2$, say $e_1$, is on $\pi_G(s, t)$. Let $H = G \setminus \{e_2\}$. We aim to compute $d_{G \setminus \{e_1, e_2\}}(s, t) = d_{H \setminus \{e_1\}}(s, t)$, i.e., a one-fault replacement path query in $H$. 
Since $e_2$ is not on $\pi_G(s, t)$, $\pi_G(s, t)$ is still a shortest path from $s$ to $t$ in $H$. Given a shortest path $\pi_H(s,t) = \pi_G(s, t)$, the structure of one-fault replacement paths is well-understood. It is known (see e.g.~\cite{WeimannYusterFDSO}) that one of the optimal one-fault replacement paths shares a prefix and a suffix with the shortest path, and contains a detour part that connects the prefix and the suffix. Importantly, the detour part does not use any edge on the original shortest path $\pi_H(s, t) = \pi_G(s, t)$. Therefore, in order to understand the distances of the detours, we need to compute the distance in the graph $H \setminus \pi_G(s, t)$,\footnote{Throughout this paper, a path $P$ formally denotes the set of its edges. Thus, $H \setminus \pi_G(s, t)$ is a subgraph of $H$ that removes all edges on the $s$ to $t$ shortest path, but keeps all the vertices on it. } which is exactly the graph $G \setminus \pi_G(s, t) \setminus \{e_2\}$. Thus, the detour distances can be efficiently computed by a one-fault DSO on the graph $G \setminus \pi_G(s, t)$. Based on this intuition, a key component in both of our algorithms is a one-fault DSO on the graph $G \setminus \pi_G(s, t)$. Depending on the range of edge weights of the input graph, we will use different DSOs accordingly. 

One-fault DSO does not seem to help the case where both failed edges $e_1, e_2$ are on $\pi_G(s, t)$. The structure of $d_{G \setminus \{e_1, e_2\}}(s, t)$ is more complicated than the structure of one-fault replacement paths. One can show that one optimal $2$-fault replacement path still shares a prefix and a suffix with $\pi_G(s, t)$, but the middle part between the prefix and the suffix is not simply a detour that does not use any edge on $\pi_G(s, t)$. In fact, it is possible that the middle part enters and leaves the subpath of $\pi_G(s, t)$ between $e_1$ and $e_2$ an arbitrary number of times, as shown in Figure~\ref{fig:overviewb}. To understand the middle part better, we study the following problem as a key subroutine in our algorithms: for every two vertices $u, v$ on $\pi_G(s, t)$ where $u$ appears earlier than $v$, we aim to compute $f(v, u)$ which is defined as $d_{G \setminus \pi_G(s, u) \setminus \pi_G(v, t)}(v, u)$, i.e., the distance of a shortest path from $v$ to $u$ that is not allowed to use edges before $u$ or after $v$ on $\pi_G(s, t)$. Intuitively, $f(v, u)$ captures the structure of the middle part of $d_{G \setminus \{e_1, e_2\}}(s, t)$, as the optimal path for $f(v, u)$ can also enter and leave the shortest path $\pi_G(s, t)$ multiple times. In Section~\ref{sec:backwards}, we will give efficient algorithms for computing these distances $f(v, u)$ in both arbitrary weighted graphs and small integer weighted graphs. The running times of these two algorithms are summarized below.
 
\begin{restatable}{lemma}{computingFWeighted}
\label{lem:computing_f_weighted}
There exists a deterministic algorithm that can compute $f$ in $n$-vertex weighted graphs with no negative cycles in $O(n^3)$ time. 
\end{restatable}

\begin{restatable}{lemma}{computingF}
\label{lem:computing_f}
There exists a randomized algorithm that can compute $f$ in $n$-vertex graphs with integer edge weights in $\{-M, \ldots, M\}$ in  $\tO(M^{1/3} n^{2+\omega / 3})$ time. Using rectangular matrix multiplication, the running time improves to $ \tO(M^{0.3544}n^{2.7778})$.
\end{restatable}

Our algorithm for $2$FRP on small integer weighted graphs needs to run SSRP multiple times on different subgraphs of $G$ (and related graphs) with different sources (see Section~\ref{sec:unweighted:overview} for an overview of the algorithm). However, the best running time for SSRP on graphs with integer edge weights in $\{-M, \ldots, M\}$ is $O(M^{0.8043} n^{2.4957})$ by Gu, Polak, Vassilevska Williams and Xu~\cite{gu_et_al}. Simply running their algorithm the required amount of times easily exceeds the running time we aim for. Fortunately, in most of our SSRP computations, we only need the replacement path distances $d_{G' \setminus \{e\}}(s, t)$ for $t \in T$ and $e \in \pi_{G'}(s, t)$, where the size of $T$ is much smaller than $n$. In Section~\ref{sec:sTRP}, we will adapt Grandoni and Vassilevska Williams's algorithm for SSRP on graphs with integer edge weights in $\{-M, \ldots, M\}$ \cite{GrandoniWilliamsSingleFailureDSO} to achieve a more efficient algorithm when $T$ is small: 

\begin{restatable}{lemma}{sTRP}
\label{lem:sTRP}
Given an $n$-vertex graph $G$ whose edge weights are in $\{-M, \ldots, M\}$, a source vertex $s \in V(G)$, and a subset $T \subseteq V(G)$, there is a randomized algorithm that computes $d_G(s, t, e)$ for every $t \in T$ and $e \in \mathcal{T}_s$ where $\mathcal{T}_s$ is a shortest path tree rooted at $s$ in $\tO(Mn^\omega + M^{\frac{1}{4-\omega}} n^{1+\frac{1}{4-\omega}} \cdot |T|)$ time with high probability. 
\end{restatable}

Note that we can potentially use ideas from \cite{gu_et_al} to make Lemma~\ref{lem:sTRP} faster for large enough $T$, but this lemma won't be a bottleneck of our whole algorithm.
In fact, for the value of $|T|$ we end up using for our $2$FRP algorithm, the $Mn^\omega$ term in the above lemma dominates the other term, and \cite{gu_et_al}'s techniques cannot avoid the $Mn^\omega$ term either. 
Therefore, we choose to adapt the simpler algorithm by Grandoni and Vassilevska Williams~\cite{GrandoniWilliamsSingleFailureDSO}.

All (except the positive weight case of Theorem~\ref{thm:main}) of our algorithms can be used to report paths efficiently, by using known techniques for finding solutions of dynamic programming and finding witnesses for matrix multiplication problems \cite{AlonWitness}.

\section{Nearly Cubic Time Algorithm for Weighted Graphs}
In this section, we show our $\tO(n^3)$ time algorithm for $2$FRP. We use two drastically different algorithms for the case where only one failed edge is on the original $s$ to $t$ shortest path and the case where both failed edges are on the original  shortest path. 

When only one failed edge is on the original $s$ to $t$ shortest path, our algorithm is essentially a simple reduction to the (one-fault) distance sensitivity oracle problem. For the other case where both failed edges are on the shortest path, we carefully  design algorithms that can capture   the patterns of optimal replacement paths. 

\subsection{Only One Failed Edge on Original Shortest Path}

Let $G=(V, E)$ be our input graph, and let $\pi_G(s, t)$ be a shortest path from $s$ to $t$ in $G$. We will compute all replacement path distances $d_{G \setminus \{e_1, e_2\}}(s, t)$ for $e_1 \in \pi_G(s, t)$ and $e_2 \not \in \pi_G(s, t)$. Our algorithm relies on the following efficient data structure for one-failure distance sensitivity oracle by Bernstein and Karger~\cite{bernstein2009nearly}.

\begin{theorem}[\cite{bernstein2009nearly}]
\label{thm:bernsteinDSO}
Given a weighted graph $H = (V, E)$ with $n$ vertices and no negative cycles, there exists a deterministic data structure that can pre-process $H$ in $O(n^3 \log^2 n)$ time and then answer queries in the form $d_{H \setminus \{e\}}(u, v)$ for any $u, v \in V$ and $e \in E$ in $O(1)$ time. Allowing randomized data structure that succeeds with high probability, the pre-processing time can be improved to $O(n^3 \log n)$.
\end{theorem}

Even though they only stated their DSO for graphs with nonnegative edge weights, their DSO also works for graphs with possibly negative edge weights but no negative cycles, after an $O(n^3)$ pre-processing step that replaces all edges with nonnegative edges \cite{johnson1977efficient, fredman1987fibonacci}. 

Given our graph $G$, we create another graph $G'$ with $O(n)$ vertices. First, we copy $G$ to $G'$ and remove all edges on $\pi_G(s, t)$. Let $h=O(n)$ be the number of vertices on $\pi_G(s, t)$ and let $p_1, \ldots, p_h$ be vertices on the path $\pi_G(s, t)$, in order they appear on $\pi_G(s, t)$. In particular, $s=p_1$ and $t = p_h$. We add $h$ vertices $a_1, \ldots, a_h$ to $G'$. For any $1 \le i \le j \le h$, we add an edge from $a_j$ to $p_i$ with weight $d_{G}(s, p_i)$. Finally, we add another $h$ vertices $b_1, \ldots, b_h$ and for any $1 \le i \le j \le h$, we add an edge from $p_j$ to $b_i$ with weight $d_G(p_j, t)$. 

The following lemma relates $d_{G \setminus \{e_1, e_2\}}(s, t)$ with replacement path distances in $G'$. 

\begin{lemma}
\label{lem:weighted_one_on}
For any $e_1 = (p_i, p_{i+1}) \in \pi_G(s, t)$ and any $e_2 \not \in \pi_G(s, t)$, $$ d_{G \setminus \{e_1, e_2\}}(s, t) = d_{G' \setminus \{e_2\}}(a_i, b_{i+1}).$$
\end{lemma}
\begin{proof}
First, we notice that $\pi_G(s, t)$ is still a shortest path from $s$ to $t$ in $G \setminus \{e_2\}$. Also, $\pi_{G \setminus \{e_1, e_2\}}(s, t)$ is a one-fault replacement path on the graph $G \setminus \{e_2\}$. It is well-known that (at least one) one-fault replacement path consists of the following three parts: a prefix that is a prefix of the original shortest path, a detour that does not use any edge on the original shortest path, and a suffix that is also a suffix of the original shortest path (see e.g. \cite{GrandoniWilliamsSingleFailureDSO}). Therefore, $d_{G \setminus \{e_1, e_2\}}(s, t)$, when viewed as a one-fault replacement path in $G \setminus \{e_2\}$ can be expressed as 
\begin{equation*}
\begin{split}
   d_{G \setminus \{e_1, e_2\}}(s, t) &=  \min_{1 \le x \le i <  y \le h} \left\{ d_{G \setminus \{e_2\}}(s, p_x) + d_{G \setminus \{e_2\} \setminus \pi_G(s, t)}(p_x, p_y) + d_{G \setminus \{e_2\}}(p_y, t) \right\}\\
   &=  \min_{1 \le x \le i <  y \le h} \left\{ d_{G}(s, p_x) + d_{G \setminus \{e_2\} \setminus \pi_G(s, t)}(p_x, p_y) + d_{G}(p_y, t) \right\}.
\end{split}
\end{equation*}
On the other hand, we consider $d_{G' \setminus \{e_2\}}(a_i, b_{i+1})$. Clearly, $a_i$ must first go to some neighbor $p_x$ for some $x \le i$ where the edge weight of $(a_i, p_x)$ is $d_G(s, p_x)$. Similarly, the last edge on any $a_i$ to $b_{i+1}$ path must travel from a neighbor of $b_{i+1}$ to $b_{i+1}$. Thus, the last edge must be from $p_y$ for some $y \ge i+1$ with weight $d_G(p_y, t)$. Also, the subpath from $p_x$ to $p_y$ lies entirely in $G' \setminus \{e_2\}$; this subpath cannot use any vertex $a_j$ or $b_j$ for $1 \le j \le h$ either because these vertices either have $0$ out-degree or $0$ in-degree. Thus, the subpath from $p_x$ to $p_y$ actually lies entirely in $G \setminus \{e_2\} \setminus \pi_G(s, t)$. We can therefore express $d_{G' \setminus \{e_2\}}(a_i, b_{i+1})$ as 
\begin{equation*}
    d_{G' \setminus \{e_2\}}(a_i, b_{i+1}) = \min_{1 \le x \le i <  y \le h} \left\{ d_{G}(s, p_x) + d_{G \setminus \{e_2\} \setminus \pi_G(s, t)}(p_x, p_y) + d_{G}(p_y, t) \right\},
\end{equation*}
which matches exactly with the formula for $d_{G \setminus \{e_1, e_2\}}(s, t)$.
\end{proof}

Using Lemma~\ref{lem:weighted_one_on} and Theorem~\ref{thm:bernsteinDSO}, we can easily solve the case where only one failed edge is on $\pi_G(s, t)$ in $\tO(n^3)$ time. We first construct $G'$ and use Theorem~\ref{thm:bernsteinDSO} to pre-process $G'$. Then for any two-fault replacement path query $d_{G \setminus \{e_1, e_2\}}(s, t)$ where $e_1 \in \pi_G(s, t)$ and $e_2 \not \in \pi_G(s, t)$, we query $d_{G' \setminus \{e_2\}}(a_i, b_{i+1})$ from the DSO in $O(1)$ time. By Lemma~\ref{lem:weighted_one_on}, this distance equals $ d_{G \setminus \{e_1, e_2\}}(s, t)$. Since there are only $O(n^2)$ queries, the pre-processing is the bottleneck and thus this case takes $O(n^3 \log^2 n)$ deterministic time or $O(n^3 \log n)$ randomized time with high probability. 

\subsection{Both Failed Edges on Original Shortest Path}

In this section, we will describe an algorithm that computes all replacement path distances $d_{G \setminus \{e_1, e_2\}}(s, t)$ for $e_1, e_2\in \pi_G(s, t)$. 
Again, we let $p_1, \ldots, p_h$ be vertices on the path $\pi_G(s, t)$, in the order they appear on  $\pi_G(s, t)$. Without loss of generality, $e_1 = (p_i, p_{i+1})$ and $e_2 = (p_j, p_{j+1})$ for some $1 \le i < j < h$. 

Let $P$ be a shortest path from $s$ to $t$ in $G\setminus \{e_1, e_2\}$ that is canonical (recall the definition of canonical in Section~\ref{sec:prelim}). There are essentially two cases in this section: the case where $P$ does not use any vertex on $\pi_G(s, t)$ between $e_1$ and $e_2$ and the case where $P$ uses at least one such vertex.

We first consider the case where $P$ does not use any vertex on $\pi_G(s, t)$ between $e_1$ and $e_2$. This can be thought as a generalization of the RP algorithm in \cite{gotthilf2009improved}.

\begin{lemma}
In $O(n^3)$ time, we can compute the replacement path distances $d_{G \setminus \{e_1, e_2\}}(s, t)$ for every pair of $e_1, e_2 \in \pi_G(s, t)$ where $e_1 = (p_i, p_{i+1})$ and $e_2 = (p_j, p_{j+1})$ for some $1 \le i < j < h$ and the replacement path does not use any vertex $p_k$ for $ i< k \le j$. 
\end{lemma}
\begin{proof}
First we fix some $e_1, e_2$ and consider their corresponding replacement path $P$. Without loss of generality, we assume $P$ is canonical. 
Let $p_x$ be the rightmost vertex (furthest from $s$) $P$ uses  on $\pi_G(s, t)$ before $e_1$. Similarly, let $p_y$ be the leftmost vertex (furthest from $t$) $P$ uses on $\pi_G(s, t)$ after $e_2$. 

Since $P$ is canonical, its subpath from $s$ to $p_x$ must use the portion from $s$ to $p_x$ on $\pi_G(s, t)$ and thus has length $d_G(s, p_x)$. Also, it implies that $p_y$ must appear after $p_x$ on $P$. 
Similarly, the subpath of $P$ from $p_y$ to $t$ must use the portion from $p_y$ to $t$ on $\pi_G(s, t)$ and thus has length  $d_G(p_y, t)$. 

Now we consider the subpath of $P$ from $p_x$ to $p_y$. It cannot use any edge between $s$ and $p_x$ on $\pi_G(s, t)$, since otherwise, the subpath from $s$ to this edge does not match the portion from $s$ to this edge on $\pi_G(s, t)$, making $P$ not canonical.  Similarly, it cannot use any edge between $p_y$ and $t$ on $\pi_G(s, t)$. The subpath of $P$ from $p_x$ to $p_y$ cannot use any edge between $p_x$ and $e_1$ on $\pi_G(s, t)$ either, due to the definition of $p_x$. Similarly, it cannot use any edge between $e_2$ and $p_y$. We also assumed that $P$ does not use any vertex between $e_1$ and $e_2$ on $\pi_G(s, t)$, and thus it does not use any edge between $e_1$ and $e_2$ either. Therefore, the subpath from $p_x$ to $p_y$ completely avoids $\pi_G(s, t)$ and thus its length is $d_{G \setminus \pi_G(s, t)}(p_x, p_y)$. 

Thus, we have shown that $d_{G \setminus \{e_1, e_1\}}(s, t) = d_G(s, p_x) + d_{G \setminus \pi_G(s, t)}(p_x, p_y) + d_G(p_y, t)$. In general, for any $e_1 = (p_i, p_{i+1})$ and $e_2 = (p_j, p_{j+1})$ for some $1 \le i < j < h$, 
$$d_{G \setminus \{e_1, e_2\}}(s, t) = \min_{1 \le x \le i, j + 1 \le y \le h} \left\{d_G(s, p_x) + d_{G \setminus \pi_G(s, t)}(p_x, p_y) + d_G(p_y, t) \right\},$$
as long as the replacement path does not use any vertex $p_k, $ for $ i< k \le j$.

Let $T(x, y)$ be $d_G(s, p_x) + d_{G \setminus \pi_G(s, t)}(p_x, p_y) + d_G(p_y, t)$. After running APSP in $G \setminus \pi_G(s, t)$ in $O(n^3)$ time, we can compute all values $T(x, y)$ and store points $(x, y)$ in a 2D range tree and associate a value $T(x, y)$ with point $(x, y)$ in $\tO(n^2)$ time, so that the 2D range tree can support orthogonal range minimum queries. Then for any  $e_1 = (p_i, p_{i+1})$ and $e_2 = (p_j, p_{j+1})$ for some $1 \le i < j < h$,  we can query the 2D range tree to get the minimum value of $T(x, y)$ such that $1 \le x \le i$ and $j + 1 \le y \le h$. Each query takes $\tO(1)$ time. Overall, the running time for this case is $O(n^3)$. 
\end{proof}

It remains to consider the case where $P$ uses some vertex on $\pi_G(s, t)$ between $e_1$ and $e_2$. We again show that there is an $O(n^3)$ time algorithm for it.

\begin{lemma}
\label{lem:cubic_two_case2}
In $O(n^3)$ time, we can compute the replacement path distances $d_{G \setminus \{e_1, e_2\}}(s, t)$ for every pair of $e_1, e_2 \in \pi_G(s, t)$ where $e_1 = (p_i, p_{i+1})$ and $e_2 = (p_j, p_{j+1})$ for some $1 \le i < j < h$ and the replacement path uses some vertex $p_k$ for $ i< k \le j$. 
\end{lemma}
\begin{proof}
First we fix some $e_1, e_2$ and consider it's corresponding replacement path $P$. Without loss of generality, we assume $P$ is canonical. 
Let $k$ be the largest integer where $ i< k \le j$ and $P$ contains $p_k$. Also, let $k'$ be the smallest integer where $i < k' \le k$ and the subpath of $P$ from $p_k$ to $t$ (including $p_k$ and $t$) uses $p_{k'}$. 

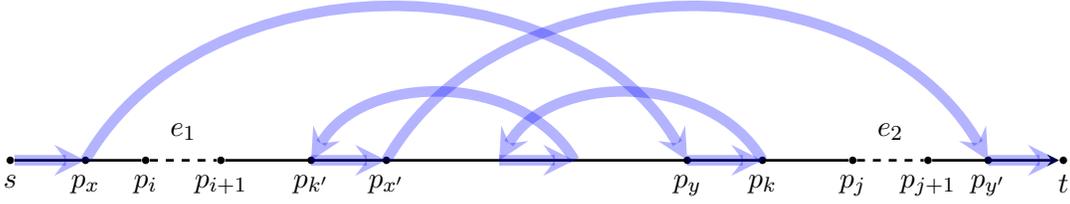
\begin{figure}[ht]
    \centering
    \begin{tikzpicture}
    	\node at (0,0) [circle,fill, inner sep = 1pt, label=below:$s$] (s){};
    	\node at (1,0) [circle,fill, inner sep = 1pt, label=below:$p_x$] (px){};
    	\node at (14,0) [circle,fill, inner sep = 1pt, label=below:$t$] (t){};
    	\node at (1.8,0) [circle,fill, inner sep = 1pt, label=below:$p_i$] (e1a){};
    	\node at (2.8,0) [circle,fill, inner sep = 1pt, label=below:$p_{i+1}$] (e1b){};
    	\node at (2.3,0) [label = above:$e_1$] (e1){};
    	
    	\node at (11.2,0) [circle,fill, inner sep = 1pt, label=below:$p_j$] (e2a){};
    	\node at (12.2,0) [circle,fill, inner sep = 1pt, label=below:$p_{j+1}$] (e2b){};
    	\node at (11.7,0) [label = above:$e_2$] (e2){};

	\node at (4, 0) [circle, fill, inner sep = 1pt, label = below:$p_{k'}$](u){};
	\node at (5, 0) [circle, fill, inner sep = 1pt, label = below:$p_{x'}$](px2){};
	\node at (13, 0) [circle, fill, inner sep = 1pt, label = below:$p_{y'}$](py2){};

	\node at (10, 0) [circle, fill, inner sep = 1pt, label = below:$p_{k}$](v){};
	\node at (9, 0) [circle, fill, inner sep = 1pt, label = below:$p_{y}$](px){};

	 \draw [-, line width = 1pt] (s) to[] (e1a);
    \draw [-, dashed, line width = 1pt] (e1a) to[] (e1b);
    \draw [-, line width = 1pt] (e1b) to[] (e2a);
    \draw [-, dashed, line width = 1pt] (e2a) to[] (e2b);
    \draw [-stealth, line width = 1pt] (e2b) to[] (t);

	\begin{scope}[transparency group, opacity=0.3, text opacity=1]
		\draw[-stealth,line width=4pt, blue] (s) to[] (1, 0);
	    \draw[-stealth,line width=4pt, blue, bend left = 60] (1,0) to[] (9, 0);
	    \draw[-stealth,line width=4pt, blue] (9,0) to[] (10, 0);
		\draw[-stealth,line width=4pt, blue, bend right = 60] (10, 0) to[] (6.5, 0);
		\draw[-stealth,line width=4pt, blue] (6.5, 0) to[] (7.5, 0);
		\draw[-stealth,line width=4pt, blue, bend right = 60] (7.5,0) to[] (4, 0);
		\draw[-stealth,line width=4pt, blue] (4, 0) to[] (5, 0);
		\draw[-stealth,line width=4pt, blue, bend left = 60] (5, 0) to[] (13, 0);
		\draw[-stealth,line width=4pt, blue] (13, 0) to[] (t);
	\end{scope}
	
    \end{tikzpicture}
    \caption{This figure depicts the vertex labels for Lemma~\ref{lem:cubic_two_case2}.}
    \label{fig:cubic_lemma}
\end{figure}

Now we consider three subpaths of $P$ separately: the subpath from $s$ to $p_k$, the subpath from $p_k$ to $p_{k'}$ and the subpath from $p_{k'}$ to $t$. 

On the $s$ to $p_k$ subpath, let $p_x$ be the rightmost vertex before $e_1$ and let $p_y$ be the leftmost vertex after $e_1$. Since $P$ is canonical and all edges between $s$ and $p_x$ do not include $e_1$ or $e_2$, the portion from $s$ to $p_x$ is $\pi_G(s, p_x)$, so $p_y$ appears after $p_x$ on the subpath. Thus, we can further decompose the $s$ to $p_k$ subpath to three parts: from $s$ to $p_x$, from $p_x$ to $p_y$ and from $p_y$ to $p_k$. Since $P$ is canonical, the subpath from $s$ to $p_x$ and from $p_y$ to $p_k$ use edges entirely from $\pi_G(s, t)$, and we know these edges don't include $e_1$ or $e_2$. Thus, the lengths of these two subpaths are $d_G(s, p_x)$ and $d_G(p_y, p_k)$ respectively. We then argue that the $p_x$ to $p_y$ subpath cannot use any edge on $\pi_G(s, t)$. It does not use any edge between $s$ and $p_x$ or between $p_y$ and $p_k$ since that would imply $P$ is not canonical. It does not use any edge between $p_x$ and $p_y$ by definitions of $p_x$ and $p_y$. It does not use any edge between $p_k$ and $p_j$ by definition of $k$. Finally, it does not use any edge between $p_{j+1}$ and $t$ because if it does, a canonical path $P$ should go directly to $t$ from that edge instead of going back to $p_y$. Thus, the subpath from $p_x$ to $p_y$ has length $d_{G \setminus \pi_G(s, t)}(p_x, p_y)$. By the above discussion, the length of the subpath from $s$ to $p_k$ can be expressed as 
$$\min_{1 \le x \le i, i + 1 \le y \le k} \left\{ d_G(s, p_x) + d_{G \setminus \pi_G(s, t)}(p_x, p_y) + d_G(p_y, p_k)\right\}. $$
We can denote this value by $g_s(e_1, p_k)$, and by using 2D range tree, we can compute $g_s(e_1, p_k)$ for all values of $e_1$ and $p_k$ in $\tO(n^2)$ time after computing APSP of $G \setminus \pi_G(s, t)$ in $O(n^3)$ time. More specifically, since 
\begin{equation*}
    \begin{split}
        g_s(e_1, p_k) &= \min_{1 \le x \le i, i + 1 \le y \le k} \left\{ d_G(s, p_x) + d_{G \setminus \pi_G(s, t)}(p_x, p_y) + (d_G(s, p_k) - d_G(s, p_y))\right\}\\
        &= d_G(s, p_k) + \min_{1 \le x \le i, i + 1 \le y \le k} \left\{ d_G(s, p_x) + d_{G \setminus \pi_G(s, t)}(p_x, p_y) - d_G(s, p_y)\right\},
    \end{split}
\end{equation*}
we can create a table $T$ where $T(x, y) = d_G(s, p_x) + d_{G \setminus \pi_G(s, t)}(p_x, p_y) - d_G(s, p_y)$ and store it in a 2D range tree, and then computing each $g_s(e_1, p_k)$ essentially costs a 2D range minimum query. 

The subpath from $p_{k'}$ to $t$ is similar. On the $p_{k'}$ to $t$ subpath, let $p_{x'}$ be the rightmost vertex before $e_2$ and let $p_{y'}$ be the leftmost vertex after $e_2$. Since $P$ is canonical, we can decompose the subpath from $p_{k'}$ to $t$ to three subpaths: from $p_{k'}$ to $p_{x'}$, from $p_{x'}$ to $p_{y'}$ and from $p_{y'}$ to $t$. 
Since $P$ is canonical, the subpath from $p_{k'}$ to $p_{x'}$ and the subpath from $p_{y'}$ to $t$ have lengths $d_G(p_{k'}, p_{x'})$ and $d_G(p_{y'}, t)$ respectively. Now we consider the portion from $p_{x'}$ to $p_{y'}$. It cannot use any edge before $p_i$ or after $p_{y'}$ since $P$ is canonical. It cannot use any edge between $p_{i+1}$ and $p_{k'}$ by definition of $p_{k'}$. It cannot use any edge between $p_{k'}$ and $p_{x'}$ since $P$ is canonical. It cannot use any edge between $p_{x'}$ and $p_{y'}$ by definitions of $p_{x'}$ and $p_{y'}$. Thus, it entirely avoids $\pi_G(s, t)$ and its length should be $d_{G \setminus \pi_G(s, t)}(p_{x'}, p_{y'})$. Therefore, the length of the subpath from $p_{k'}$ to $t$ can be expressed as 
$$\min_{k' \le x' \le j, j + 1 \le y' \le h} \left\{ d_G(p_{k'}, p_{x'}) + d_{G \setminus \pi_G(s, t)}(p_{x'}, p_{y'}) + d_G(p_{y'}, t)\right\}. $$
We denote this value by $g_t(e_2, p_{k'})$. By using 2D range tree, we can compute $g_t(e_2, p_{k'})$ for all values of $e_2$ and $p_{k'}$ in $\tO(n^2)$ time after computing APSP of $G \setminus \pi_G(s, t)$ in $O(n^3)$ time. We omit the details for the $2$D range tree in this case since it is almost identical to  the $s$ to $p_k$ subpath case. 

Finally, we consider the $p_k$ to $p_{k'}$ subpath. It does not use any edge on $\pi_G(s, t)$ before $e_1$ or after $e_2$ because $P$ is canonical. It does not use any edge after $e_1$ and before $p_{k'}$ or any edge after $p_k$ and before $e_2$ by definitions of $k$ and $k'$. Therefore, this subpath lies entirely in $G \setminus \pi_G(s, p_{k'}) \setminus \pi_G(p_k, t)$. Thus, the length of this subpath is exactly $d_{G \setminus \pi_G(s, p_{k'}) \setminus \pi_G(p_k, t)}(p_k, p_{k'})$, which was denoted by $f(p_k, p_{k'})$ in Section~\ref{sec:overview}. All values of $f(p_k, p_{k'})$ for any $p_k$ and $p_{k'}$ can be computed deterministically in $O(n^3)$ time by Lemma~\ref{lem:computing_f_weighted}.

Therefore, for any $e_1 = (p_i, p_{i+1})$ and $e_2 = (p_j, p_{j+1})$ for some $1 \le i < j < h$, 
$$d_{G \setminus \{e_1, e_2\}}(s, t) = \min_{i+1 \le k' \le k \le j} \left\{ g_s(e_1, p_k) + f(p_k, p_{k'}) + g_t(e_2, p_{k'}) \right\},$$
as long as the replacement path uses some vertex on $\pi_G(s, t)$ between $e_1$ and $e_2$. To compute the right hand side of the above equation efficiently, we first create an array $A_{k, e_2}(k') = f(p_k, p_{k'}) + g_t(e_2, p_{k'})$ for every $k$ and $e_2$ and build a data structure  that supports range minimum queries for each array. Then for every $e_1 = (p_i, p_{i+1}), e_2 = (p_j, p_{j+1})$, we enumerate $k \in [i+1, j]$. We can write 
$$\min_{i + 1 \le k' \le k } \left\{ g_s(e_1, p_k) + f(p_k, p_{k'}) + g_t(e_2, p_{k'}) \right\}$$
as
$$g_s(e_1, p_k) + \min_{i + 1 \le k' \le k} A_{k, e_2}(k').$$
Thus, it essentially costs one range minimum query for every triple of $e_1, e_2, k$. If we use range minimum query data structures that supports linear time pre-processing and $O(1)$ range minimum queries (see e.g.~\cite{harel1984fast}), this step takes $O(n^3)$ time. 

Therefore, the overall running time is $O(n^3)$. 
\end{proof}

\subsection{Putting It All Together}

Recall our Theorem~\ref{thm:main_weighted} is the following:

\mainWeighted*
\begin{proof}

All components in our algorithm run in $O(n^3)$ time deterministically except the pre-processing phase of the distance sensitivity oracle from Theorem~\ref{thm:bernsteinDSO}. Since the DSO has an $O(n^3 \log n)$ randomized pre-processing time or an $O(n^3 \log^2 n)$ deterministic pre-processing time, our algorithm for $2$FRP has $O(n^3 \log n)$ randomized time or $O(n^3 \log^2 n)$ deterministic time. 
\end{proof}

Using Theorem~\ref{thm:main_weighted}, we immediately obtain Corollary~\ref{cor:fFRP}.

\CorWeighted*
\begin{proof}
Since the graph has no negative cycles, we can first use an $O(n^3)$ pre-processing step that replaces all edges with nonnegative edges \cite{johnson1977efficient, fredman1987fibonacci}. Then we compute a shortest path $P_1$ from $s$ to $t$ in $\tilde{O}(n^2)$ time using Dijkstra's algorithm. For every edge $e_1$ on $P_1$, we compute a shortest $s$-$t$ path $P_2$ from $s$ to $t$ in $G\setminus \{e_1\}$. More generally, for each $i \le f - 2$, and each choice of $(e_1,\ldots,e_{i})$ and computed paths $P_1,\ldots,P_{i}$ where each $P_j$ is a shortest $s$-$t$ path in $G\setminus \{e_1,\ldots,e_{j-1}\}$ and $e_j\in P_j$, we compute a shortest $s$-$t$ path $P_{i+1}$ in $G\setminus \{e_1,\ldots,e_{i}\}$. This computation takes $\tilde{O}(n^f)$ time. Then for each of the $O(n^{f-2})$ choices of $(e_1,\ldots,e_{f-2})$, we compute $2$FRP using Theorem~\ref{thm:main_weighted} in $G\setminus \{e_1,\ldots,e_{i}\}$ in overall time $\tilde{O}(n^{f+1})$.
\end{proof}

\section{Subcubic Time Algorithm for Graphs with Bounded Integer Weights}
\label{sec:bounded_weight}
In this section, we show how to improve the running time of $2$FRP when we restrict the graphs to graphs with small integer edge weights and no negative cycles, providing proofs for Theorem~\ref{thm:main} and Corollary~\ref{cor: constant_query_time}.

\subsection{General Approach and Intuitions}
\label{sec:unweighted:overview}
We first give some intuitions and high-level ideas of our algorithm. 

Let $G = (V, E)$ be an $n$-vertex directed graph with integer edge weights in $\{-M \ldots M\}$ and no negative cycles. Let $s$ be the source and $t$ be the target for our $2$FRP instance. 
The general approach to our algorithm is to divide $\pi_G(s, t)$ into intervals of $g$ vertices each for a positive integer parameter $g = O(n)$.\footnote{For instance, we will set $g=n^{(\omega-1)/3}=O(n)$ when $M=O(1)$.} Let $I$ be one of the intervals, then we use $V(I)$ to  denote the vertices inside the interval, and $E(I)$ to denote the edges inside the interval. The intervals are created in such a way that the last vertex in the previous interval is the first vertex in the next interval. 
Once we have created these intervals we can classify all the two-fault replacement paths queries to the following three cases: (1) only one failed edge is on $\pi_G(s, t)$, (2) both failed edges are on $\pi_G(s, t)$ in the same interval, and (3) both failed edges are on $\pi_G(s, t)$ in different intervals. Note that we don't need to consider cases where neither of the failed edges is on $\pi_G(s, t)$ as the original shortest path will exist in $G \setminus \{e_1, e_2\}$. Now, we can create three separate sub-algorithms that handle each of these cases, and combine them to get the overall $2$FRP algorithm. 

We will have a general precomputation step and some sub-algorithms will also have their own precomputation steps to compute any needed information that was not computed in the general precomputation step. Our approach to querying the length of the replacement path in all of the sub-algorithms is to construct a weighted auxiliary graph to aid with the query. To build one such auxiliary graph, we first determine a set of critical vertices that break down the replacement path into a series of subpaths between them. These vertices will form the vertex set of the auxiliary graph, and the edges in the auxiliary graph will represent subpaths between these vertices. 

We say that the auxiliary graph \textit{encodes} a subpath from $u$ to $v$ in $G \setminus \{e_1, e_2\}$ if there is a path from $u$ to $v$ in the auxiliary graph with the same length as the subpath. We also say that an edge $(u, v)$ in the auxiliary graph \textit{encodes} a subpath from $u$ to $v$ in $G \setminus \{e_1, e_2\}$ if the weight of that edge equals the length of the subpath.
We will show many of those subpaths are encoded in the auxiliary graph, and eventually, show that the $s$-$t$ shortest path in $G \setminus \{e_1, e_2\}$ is encoded. Thus, we can run a Single-Source Shortest Paths (SSSP) algorithm on the auxiliary graph to get the length of the shortest replacement path.

\subsection{Precomputed Distances}

In the general precomputation step and the precomputation steps specific to each sub-algorithm, we will use the
SSRP algorithm from \cite{GrandoniWilliamsSingleFailureDSO} that has the same running time as Zwick's APSP algorithm for graphs with edge weights in $\{-M, \ldots, M\}$ and our algorithm for SSRP with a small set of targets from Lemma~\ref{lem:sTRP}.
In the precomputation steps and the query step, we will also use the near-linear time SSSP algorithm by Bernstein, Nanongkai and Wulff-Nilsen \cite{bernstein2022negative}. On $n$-vertex dense graphs, their algorithm runs in $\tO(n^2)$ time. 

In this general precomputation step we compute sets of distances that will be needed for all three sub-algorithms:

\begin{enumerate}
    \item Run SSSP and SSRP from $s$ in $G$, and store the results.
    \item Run SSSP and SSRP from $t$ in $\widehat{G}$, where $\widehat{G}$ represents $G$ with the directions of all of its edges reversed, and store the results.
    \item For each interval $I$, create the graphs $G\setminus E(I)$ and $\widehat{G \setminus E(I)}$, then:
    \begin{enumerate}
         \item Run SSSP and SSRP with target set $V(I) \cup \{t\}$ from $s$ in $G\setminus E(I)$, and store the results.
         \item Run SSSP and SSRP with target set $V(I) \cup \{s\}$ from $t$ in $\widehat{G \setminus E(I)}$ and store the results.
    \end{enumerate}
    \item Run Zwick's All-Pairs Shortest Paths algorithm  \cite{Zwick02} on the graph with all the edges of $\pi_G(s, t)$ removed and store the results.
\end{enumerate}

Steps 1, 2 and 4 take $\tO(M^{1/(4 - \omega)}n^{2 + 1/(4-\omega)})$ time, and each iteration of step 3 takes $\tO(Mn^\omega + M^{1/(4 - \omega)}n^{1 + 1/(4-\omega)} g)$ time, so overall these pre-processing steps take $\tO(Mn^{\omega+1}/g + M^{1/(4 - \omega)}n^{2 + 1/(4 - \omega)})$ time. The space complexity of the stored results in this step is $\tO(n^2)$. 

\subsection{Only One Failed Edge on Original Shortest Path}

First, we consider the case where only one of the failed edges is on $\pi_G(s, t)$. Let $e_1$ be the failed edge on $\pi_G(s, t)$, $I_1$ be the interval containing $e_1$, and $e_2$ be the failed edge that is not on $\pi_G(s, t)$. Our auxiliary graph requires distances not computed in the general precomputation step of the algorithm, so we will have a precomputation step for this algorithm. During this precomputation step, we compute a single-fault DSO for $G\setminus \pi_G(s, t)$. Using Chechik and Cohen's DSO \cite{chechik2020distance}, the pre-processing time is $\tO(Mn^{2.8729})$ and the space is $\tO(n^{2.5})$. If all edge weights are positive integers in $\{1, \ldots, M\}$, we can instead use Gu and Ren's DSO~\cite{gurenDSO}, which has $\tO(M n^{2.5794})$ pre-processing time and $\tO(n^2)$ size. Note that although Gu and Ren's DSO has $\tO(n^2)$ size, it could use  $\tO(n^{2.4207})$ space during pre-processing \cite{gurenDSO}.

Let $G'$ be the auxiliary graph, and let its vertex set be $\{s, t\} \cup V(I_1)$. We will add edges in the following steps:

\begin{enumerate}
    \item Add an edge from $s$ to $t$ with weight $d_{G\setminus E(I_1)}(s, t, e_2)$.
    \item For every $v \in V(I)$, add an edge $(s, v)$ with weight $d_{G\setminus E(I_1)}(s, v, e_2)$.
    \item For every $v \in V(I)$, add an edge $(v, t)$ with weight $d_{G\setminus E(I_1)}(v, t, e_2)$.
    \item Add all of the edges in $I_1$ that are not one of the two failed edges. Then, for every $u, v \in V(I_1)$, add the edge $(u, v)$ with weight $d_{G\setminus \pi_G(s, t)}(u, v, e_2)$.
\end{enumerate}

Now we can run SSSP from $s$ in $G'$. Since there are $O(g)$ vertices in each interval, there are $O(g)$ vertices in $G'$, so building $G'$ and running the query takes $\tO(g^2)$ times.

\begin{theorem}
\label{theorem:OneEdgeOnOneEdgeOffCorrectness}
$d_{G'}(s, t)$ is equal to $d_{G\setminus \{e_1, e_2\}}(s, t)$.
\end{theorem}
\begin{proof}
Here are two main cases for the shortest $s$-$t$ path that avoids $e_1$ and $e_2$: (1) the path does not use any edge in $E(I_1)$ or (2) the path does use edges in $E(I_1)$. The path for the first case is encoded in $G'$ via the edge added in step 1.

For the second case, let $P$ be a canonical replacement path. We know that the shortest path $P$ must use a vertex in $V(I_1)$. Let $u$ be the first vertex on the replacement path that is in $V(I_1)$, and $v$ be the last vertex on the replacement path that is in $V(I_1)$. Then, the replacement path can be broken down into three subpaths: (1) a subpath from $s$ to $u$, (2) a subpath from $u$ to $v$, and (3) a subpath from $v$ to $t$. If $G'$ encodes each of these subpaths for every possible value of $u$ and $v$, and the replacement path does use edges in $E(I_1)$, then $\pi_{G'}(s, t)$ will use the optimal choices for $u$ and $v$, which will give us the length of the replacement path.

First, we will focus on the subpaths from $s$ to $u$ for all choices of $u$. Since $u$ will be the first vertex in $V(I_1)$ on the replacement path, this subpath will not use any edge in $E(I_1)$. Similarly, the subpaths from $v$ to $t$ will not use any edge in $E(I_1)$ either, since $v$ is the last vertex in $V(I_1)$ on the replacement path. All of the subpaths must also avoid $e_2$, since it is a failed edge. Therefore, the edges added in steps 2 and 3 are sufficient to encode the $s$-$u$ subpaths and $v$-$t$ subpaths into $G'$.

Next, we will focus on the subpaths between $u$ and $v$ for all choices of $u$ and $v$. If a canonical replacement path does travel between two vertices in $V(I_1)$, then it will not use any edge on $\pi_G(s, t)$ outside of $I_1$ to do so, as that would prevent it from being canonical. For example, if the path between $u$ and $v$ reached a vertex $w$ in $\pi_G(s, t)$ before $I_1$, then the replacement path should go directly from $s$ to $w$, instead of going to $u$ first, because it is canonical. The mirror situation occurs if it touches a vertex after $I_1$ and cannot happen for similar reasons. Therefore, when traveling between two vertices in $V(I_1)$, the replacement path will only use edges in $E(I_1)$ and edges not on $\pi_G(s, t)$. As a result, every $u$-$v$ subpath can be broken into a series of smaller paths consisting of edges in $E(I_1)$ and paths between vertices in $V(I_1)$ that do not use any edge on $\pi_G(s, t)$. The edges in step 4 encode all of these smaller paths into $G'$, and as a result every $u$-$v$ subpath is encoded in $G'$.

In total, $G'$ encodes every possible subpath which the replacement path could be constructed from, so the shortest $s$-$t$ path in $G'$ can not be longer than the replacement path in $G\setminus \{e_1, e_2\}$. It is impossible for $d_{G'}(s, t)$ to be smaller than $d_{G\setminus \{e_1, e_2\}}(s, t)$ because all of the edges in $G'$ have weights that correspond to the lengths of some paths that are present in $G\setminus \{e_1, e_2\}$. Therefore, $d_{G'}(s, t)$ must be equal to $d_{G\setminus \{e_1, e_2\}}(s, t)$.
\end{proof}

\subsection{Both Failed Edges on Original Shortest Path: Same Interval}

Next, we will consider the case where both failed edges $e_1, e_2$ are on $\pi_G(s, t)$ in the same interval $I$. We start by constructing the auxiliary graph for this query. Let $G'$ be the auxiliary graph, and let its vertex set be $\{s, t\} \cup V(I)$. We will add edges in the following steps:

\begin{enumerate}
    \item Add an edge from $s$ to $t$ with weight $d_{G\setminus E(I)}(s, t)$.
    \item For every $v \in V(I)$, add an edge $(s, v)$ with weight $d_{G\setminus E(I)}(s, v)$.
    \item For every $v \in V(I)$, add an edge $(v, t)$ with weight $d_{G\setminus E(I)}(v, t)$.
    \item Add all of the edges in $I$ that are not one of the two failed edges. Then, for every $u, v \in V(I)$, add an edge $(u, v)$ with weight $d_{G\setminus \pi_G(s, t)}(u, v)$.
\end{enumerate}

Now we can run SSSP from $s$ in $G'$. Since there are $O(g)$ vertices in each interval, there are $O(g)$ vertices in $G'$, so building $G'$ and running the query takes $\tO(g^2)$ time. All of the edge weights of $G'$ were already calculated in the general precomputation step, so there is no additional precomputation step for this sub-algorithm.

\begin{theorem}
\label{theorem:SameIntervalCorrectness}
$d_{G'}(s, t)$ is equal to $d_{G\setminus \{e_1, e_2\}}(s, t)$.
\end{theorem}
\begin{proof}
The proof for this case is similar to the proof of Theorem~\ref{theorem:OneEdgeOnOneEdgeOffCorrectness} since both only involve one interval $I$ on $\pi_G(s, t)$. As before, there are two main cases for the shortest $s$-$t$ path that avoids $e_1$ and $e_2$: (1) the path does not use any edge in $E(I)$ or (2) the path does use edges in $E(I)$. The path for the first case is encoded in $G'$ via the edge added in step 1.

For the second case, let $P$ be a canonical replacement path, let $u$ be the first vertex on $P$ that is in $V(I)$, and let $v$ be the last vertex on $P$ that is in $V(I)$. As before, $P$ can be broken down into a subpath from $s$ to $u$, a subpath from $u$ to $v$, and a subpath from $v$ to $t$. If $G'$ encodes each of these subpaths for every possible value of $u$ and $v$, and $P$ does use edges in $E(I)$, then $\pi_{G'}(s, t)$ will use the optimal choices for $u$ and $v$, which will give us the length of the replacement path.

First, we will focus on the subpaths from $s$ to $u$ for all choices of $u$. Since $u$ will be the first vertex in $V(I)$ on $P$, this subpath will not use any edge in $E(I)$. Therefore, the edges added in step 2 encode these subpaths into $G'$. Similarly, the subpaths from $v$ to $t$ will also not use any edge in $E(I)$, since $v$ is the last vertex in $V(I)$ on $P$. Therefore, the edges added in step 3 encode these subpaths into $G'$.

Next, we will focus on the subpaths between $u$ and $v$ for all choices of $u$ and $v$. Using similar reasoning to that presented in the previous proof, we find that when traveling between two vertices in $V(I)$, $P$ will only use edges in $E(I)$ and edges not on $\pi_G(s, t)$. As a result, every $u$-$v$ subpath can be broken into a series of smaller paths consisting of edges in $E(I)$ and paths between vertices in $V(I)$ that do not use any edge of $\pi_G(s, t)$. The edges added in step 4 encode all of these smaller paths into $G'$, and as a result every $u$-$v$ subpath is encoded in $G'$.

In total, $G'$ encodes every possible subpath which the replacement path could be constructed from, so the shortest $s$-$t$ path in $G'$ can not be longer than the replacement path in $G\setminus \{e_1, e_2\}$. It is impossible for $d_{G'}(s, t)$ to be smaller than $d_{G\setminus \{e_1, e_2\}}(s, t)$ either because all of the edges in $G'$ have weights that correspond to the lengths of some paths that are present in $G\setminus \{e_1, e_2\}$. Therefore, $d_{G'}(s, t)$ must be equal to $d_{G\setminus \{e_1, e_2\}}(s, t)$.
\end{proof}

\subsection{Both Failed Edges on Original Shortest Path: Different Intervals}
\label{sec:bounded_weight_both_different}

\newcommand{\Ione}{I_1}
\newcommand{\Itwo}{I_2}
\newcommand{\leftone}{l_1}
\newcommand{\lefttwo}{l_2}
\newcommand{\rightone}{r_1}
\newcommand{\righttwo}{r_2}
\newcommand{\Beforeone}{L_1}
\newcommand{\Beforetwo}{L_2}
\newcommand{\Afterone}{R_1}
\newcommand{\Aftertwo}{R_2}

Finally, we will consider the case where both failed edges are on $\pi_G(s, t)$ but in two different intervals $\Ione$ and $\Itwo$, where $\Ione$ comes first on the shortest path. Let the failed edge in $\Ione$ be $e_1 = (a_1, b_1)$ and the failed edge in $\Itwo$ be $e_2 = (a_2, b_2)$.

We use $\leftone$ (resp. $\lefttwo$) to denote the first vertex of  $\Ione$ (resp. $\Itwo$), and $\rightone$ (resp. $\righttwo$) to denote the last vertex of  $\Ione$ (resp. $\Itwo$).
For the purposes of this query algorithm we will divide the vertices in $\Ione$ and $\Itwo$ into four groups as follows: $\Beforeone$ contains every vertex in $\Ione$ from $\leftone$ to $a_1$ inclusively, $\Afterone$ contains every vertex in $\Ione$ from $b_1$ to $\rightone$ inclusively, $\Beforetwo$ contains every vertex in $\Itwo$ from $\lefttwo$ to $a_2$ inclusively, and $\Aftertwo$ contains every vertex in $\Itwo$ from $b_2$ to $\righttwo$ inclusively. 

\begin{figure}[ht]
    \centering
    \begin{tikzpicture}
    	\node at (0,0) [circle,fill, inner sep = 1pt, label=below:$s\vphantom{\leftone}$] (s){};
    	\node at (14,0) [circle,fill, inner sep = 1pt, label=below:$t\vphantom{\leftone}$] (t){};
    	\node at (3.4,0) [circle,fill, inner sep = 1pt] (e1a){};
    	\node at (4,0) [circle,fill, inner sep = 1pt] (e1b){};
    	\node at (3.7,0) [label = below:$e_1\vphantom{\leftone}$] (e1){};
    	
    	\node at (10,0) [circle,fill, inner sep = 1pt] (e2a){};
    	\node at (10.6,0) [circle,fill, inner sep = 1pt] (e2b){};
    	\node at (10.3,0) [label = below:$e_2\vphantom{\leftone}$] (e2){};

	\node at (2,0) [circle,fill, inner sep = 1pt, label=below:$\leftone\vphantom{\leftone}$] (leftone){};
	\node at (5,0) [circle,fill, inner sep = 1pt, label=below:$\rightone\vphantom{\leftone}$] (rightone){};

	\node at (9,0) [circle,fill, inner sep = 1pt, label=below:$\lefttwo\vphantom{\leftone}$] (lefttwo){};
	\node at (12,0) [circle,fill, inner sep = 1pt, label=below:$\righttwo\vphantom{\leftone}$] (righttwo){};
	
	\draw[Bar-Bar, line width=0.5pt] (2, 0.3)  -- node[fill=white]{$\Beforeone$} (3.4, 0.3) ;
	\draw[Bar-Bar, line width=0.5pt] (4.0, 0.3)  -- node[fill=white]{$\Afterone$} (5.0, 0.3) ;

	\draw[Bar-Bar, line width=0.5pt] (9, 0.3)  -- node[fill=white]{$\Beforetwo$} (10, 0.3) ;
	\draw[Bar-Bar, line width=0.5pt] (10.6, 0.3)  -- node[fill=white]{$\Aftertwo$} (12, 0.3) ;

	\draw[Bar-Bar, line width=0.5pt] (2.0, 0.8)  -- node[fill=white]{$\Ione$} (5.0,0.8) ;
	\draw[Bar-Bar, line width=0.5pt] (9.0, 0.8)  -- node[fill=white]{$\Itwo$} (12.0,0.8) ;

    \draw [-, line width = 1pt] (s) to[] (e1a);
    \draw [-, dashed, line width = 1pt] (e1a) to[] (e1b);
    \draw [-, line width = 1pt] (e1b) to[] (e2a);
    \draw [-, dashed, line width = 1pt] (e2a) to[] (e2b);
    \draw [-stealth, line width = 1pt] (e2b) to[] (t);

    \end{tikzpicture}
    \caption{Some notations used in Section~\ref{sec:bounded_weight_both_different}.}
    \label{fig:bounded_weight_both_different}
\end{figure}
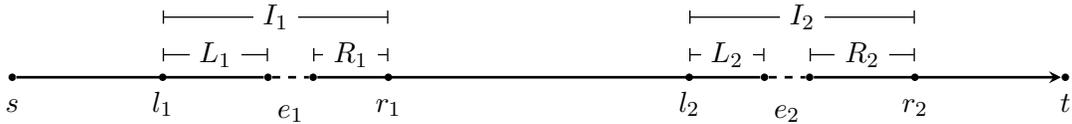

For technical reasons, we will first show how to compute the distance of replacement paths that are \textbf{good}, which is defined as follows.
\begin{definition}
A replacement path $\pi_{G \setminus \{e_1, e_2\}}(s, t)$ is called \textbf{bad} if it contains a subpath from $u$ to $v$ where 
\begin{enumerate}
    \item $u \in \Beforeone$ and $v \in \Afterone$;
    \item the subpath contains no other vertices in $V(\Ione) \cup V(\Itwo) \cup \{s, t\}$;
    \item and the subpath uses some edge on $\pi_G(s, t)$. 
\end{enumerate}
All other replacement paths are called \textbf{good}. 
\end{definition}

Clearly, if we reverse the directions of all the edges, switch the roles of $s$ and $t$, and keep the same intervals in an algorithm that works for good replacement paths, the algorithm can work for replacement paths that are \textbf{reversely good}, which is more formally defined as follows.

\begin{definition}
A replacement path $\pi_{G \setminus \{e_1, e_2\}}(s, t)$ is called \textbf{reversely bad} if it contains a subpath from $u$ to $v$ where 
\begin{enumerate}
    \item $u \in \Beforetwo$ and $v \in \Aftertwo$\;
    \item the subpath contains no other vertices in $V(\Ione) \cup V(\Itwo) \cup \{s, t\}$\;
    \item and the subpath uses some edge on $\pi_G(s, t)$. 
\end{enumerate}
All other replacement paths are called \textbf{reversely good}. 
\end{definition}

Our next lemma shows that a canonical replacement path is either good or reversely good, thus allowing us to focus on good and canonical replacement path since we can take care of the reversely good replacement paths by running the same algorithm for good replacement paths on $\widehat{G}$ with the roles of $s$ and $t$ switched. 

\begin{lemma}
Let $P$ be a canonical replacement path from $s$ to $t$ avoiding $e_1$ and $e_2$, then $P$ is either good or reversely good. 
\end{lemma}
\begin{proof}
Assume for the sake of contradiction that $P$ is both bad and reversely bad. By the definitions, we can find $u_1 \in \Beforeone$ and $v_1 \in \Afterone$ on $P$ such that the $u_1$-$v_1$ subpath does not use any other vertices in $V(\Ione) \cup V(\Itwo) \cup \{s, t\}$ while it uses some edge on $\pi_G(s, t)$. Also, we can find $u_2 \in \Beforetwo$ and $v_2 \in \Aftertwo$ such that the $u_2$-$v_2$ subpath does not use any other vertices in $V(\Ione) \cup V(\Itwo) \cup \{s, t\}$ while it uses some edge on $\pi_G(s, t)$.

Since $P$ is canonical, the $u_1$-$v_1$ subpath must not use any edge before $e_1$ or after $e_2$ on $\pi_G(s, t)$. Also, it does not use any other vertices in $V(\Ione) \cup V(\Itwo) \cup \{s, t\}$, so it must not use any edge on $\Ione$ or $\Itwo$ either. Therefore, the edge $(x, y)$ on $\pi_G(s, t)$ that it uses must be after $\Ione$ and before $\Itwo$. Also, $u_2$ must appear later than $u_1$ and $v_1$ on $P$, since the $s$-$u_1$ subpath should exactly be $\pi_G(s, u_1)$ by the assumption that $P$ is canonical and the $u_1$-$v_1$ subpath does not use any other vertices in  $V(\Ione) \cup V(\Itwo) \cup \{s, t\}$. However, since $P$ is canonical, once it reaches $x$, it should head directly to $u_2$ using edges on $\pi_G(s, t)$ instead of going to $v_1$, a contradiction. 
\end{proof}

From now on, we can focus on computing the lengths of good replacement paths. 

Our auxiliary graph will require distances that were not computed in the general precomputation step, so we will have a precomputation step for this sub-algorithm. The precomputation step is as follows:

\begin{enumerate}
    \item Run Lemma~\ref{lem:computing_f} to compute $f(v, u) = d_{G\setminus \pi_G(s, u)\setminus \pi_G(v, t)}(v, u)$ for all $u, v \in V(\pi_G(s, t))$ such that $u$ comes before $v$ on $\pi_G(s, t)$ and store the results. Here, we use $\pi_G(s, u)$ to denote the subpath from $s$ to $u$ on the path $\pi_G(s, t)$; similarly, we use $\pi_G(v, t)$ to denote the subpath from $v$ to $t$ on the path $\pi_G(s, t)$.
    \item For each interval $I$, create the graphs $G\setminus E(I)$ and $\widehat{G \setminus E(I)}$, then:
    \begin{enumerate}
        \item Run SSRP with target set $V(I) \cup \{s, t\}$ from the last vertex in $V(I)$ in $G\setminus E(I)$ and store the results for all intervals $I$. 
        \item Run SSRP with target set $V(I) \cup \{s, t\}$ from the first vertex in $V(I)$  in $\widehat{G \setminus E(I)}$ and store the results for all intervals $I$.
    \end{enumerate}
\end{enumerate}

Step 1 takes $\tO(M^{1/3}n^{2 + \omega/3})$ time and each iteration of step 2 takes $\tO(Mn^\omega + M^{1/(4 - \omega)}n^{1 + 1/(4-\omega)} g)$ time, so overall the precomputation step of this sub-algorithm takes $\tO(M n^{\omega+1}/g + M^{1/(4 - \omega)}n^{2 + 1/(4-\omega)} + M^{1/3}n^{2 + \omega/3})$ time. By Lemma~\ref{lem:computing_f}, we can improve the third term to $\tO(M^{0.3544}n^{2.7778})$ using rectangular matrix multiplication.  The space complexity is again $\tO(n^2)$.

Let the auxiliary graph be called $G'$ and let its vertex set be $\{s, t\} \cup V(\Ione) \cup V(\Itwo)$. We will add edges to $G'$ in the following steps:

\begin{enumerate}
    \item Add an edge from $s$ to $t$ with weight $\min\left\{ d_{G\setminus E(\Ione)}(s, t, e_2), d_{G\setminus E(\Itwo)}(s, t, e_1)) \right\}$.
    \item Add every edge in $E(\Ione)$ other than $e_1$, and every edge in $E(\Itwo)$ other than $e_2$.
    \item Add an edge from $s$ to $\leftone$ with weight $d_G(s, \leftone)$, an edge from $\righttwo$ to $t$ with weight $d_G(\righttwo, t)$, and an edge from $\rightone$ to $\lefttwo$ with weight $d_G(\rightone, \lefttwo)$.
    \item For all $v \in \Beforeone \cup \Afterone $, add an edge from $s$ to $v$ with weight $d_{G\setminus E(\Ione)}(s, v, e_2)$ and add an edge from $v$ to $t$ with weight $d_{G\setminus E(\Ione)}(v, t, e_2)$.
    \item For all $v \in  \Beforetwo \cup \Aftertwo$, add an edge from $s$ to $v$ with weight $d_{G\setminus E(\Itwo)}(s, v, e_1)$ and add an edge from $v$ to $t$ with weight $d_{G\setminus E(\Itwo)}(v, t, e_1)$.
    \item For every $(u, v) \in  \Afterone \times \Beforetwo$, add an edge from $v$ to $u$ with weight $d_{G\setminus \pi_G(s,u)\setminus \pi_G(v,t)}(v, u) = f(v, u)$. 
    \item For every $v \in \Afterone$, add an edge from $\rightone$ to $v$ with weight $d_{G\setminus E(\Ione)}(\rightone, v, e_2)$. 
    \item For every $v \in \Beforetwo$, add an edge from $v$ to $\lefttwo$ with weight $d_{G\setminus E(\Itwo)}(v, \lefttwo, e_1)$. 
    \item For every $u, v \in V(\Ione) \cup V(\Itwo) \cup \{s, t\}$, add an edge from $u$ to $v$ with weight $d_{G\setminus \pi_G(s,t)}(u, v)$. 
\end{enumerate}

Now we can run SSSP from $s$ in $G'$. Since there are $O(g)$ vertices in each interval, there are $O(g)$ vertices in $G'$, so building $G'$ and running the query takes $\tO(g^2)$ time. 
Each edge weight of $G'$ can be computed in $O(1)$ time due to the general precomputation step and this sub-algorithm's own precomputation step. 
We will show the following theorem. 

\begin{theorem}
\label{theorem:DifferentIntervalCorrectness}
$d_{G'}(s, t)$ is equal to $d_{G\setminus \{e_1, e_2\}}(s, t)$ as long as one of the canonical replacement paths is \textit{good}. 
\end{theorem}

In the following, let $P$ be a canonical and good shortest path from $s$ to $t$ in $G \setminus \{e_1, e_2\}$.
We will first show the following two simple lemmas that will be used multiple times later. 

\begin{lemma}
\label{lem:sub_path_encoding}
Let $u$, $v$, and $w$ be vertices both on $P$ and in $G'$ such that $u$ comes before $v$ on $P$ and $v$ comes before $w$ on $P$. If the subpath from $u$ to $v$ of $P$ is encoded in $G'$, and the subpath from $v$ to $w$  of $P$ is encoded in $G'$, then the subpath from $u$ to $w$  of $P$ is encoded in $G'$.
\end{lemma}
\begin{proof}
Due to the order of the vertices on $P$, the length of the subpath from $u$ to $w$ can be written as $d_{G \setminus \{e_1, e_2\}}(u, v) + d_{G \setminus \{e_1, e_2\}}(v, w)$. Since the $u$-$v$ subpath is encoded in $G'$, there is a path in $G'$ from $u$ to $v$ with weight $d_{G \setminus \{e_1, e_2\}}(u, v)$, and since the $v$-$w$ subpath is encoded in $G'$, there is a path in $G'$ from $v$ to $w$ with weight $d_{G \setminus \{e_1, e_2\}}(v, w)$. Therefore, there is a path in $G'$ from $u$ to $w$ with weight $d_{G \setminus \{e_1, e_2\}}(u, v) + d_{G \setminus \{e_1, e_2\}}(v, w)$, so the subpath from $u$ to $w$ is encoded in $G'$.
\end{proof}

We say that a subpath between two vertices $u$ and $v$ on $\pi_G(s, t)$ moves \textit{forward} if $u$ comes before $v$ on $\pi_G(s, t)$, and \textit{backward} if $u$ comes after $v$ on $\pi_G(s, t)$.

\begin{lemma}
\label{lem:forwards_and_backwards_movement}
There is no subpath of $P$ which moves backward to a vertex before $e_1$, and there is no subpath of $P$ which moves backward from a vertex after $e_2$.
\end{lemma}
\begin{proof}
For any vertex $u$ on $\pi_G(s, t)$ before $e_1$, the subpath of $P$ from $s$ to $u$ is the same as the subpath of $\pi_G(s, t)$ from $s$ to $u$, since $P$ is canonical. Therefore, it is impossible for there to be a subpath of $P$ that moves backward to $u$ as this would imply that the $s$-$u$ subpath of $P$ is different from the $s$-$u$ subpath of $\pi_G(s, t)$.

Similarly, for any vertex $v$ on $\pi_G(s, t)$ after $e_2$, the subpath of $P$ from $v$ to $t$ is the same as the subpath of $\pi_G(s, t)$ from $v$ to $t$. Therefore, it is impossible for there to be a subpath of $P$ that moves backward from $v$ as this would imply that the $v$-$t$ subpath of $P$ is different from the $v$-$t$ subpath of $\pi_G(s, t)$.
\end{proof}

Using Lemma~\ref{lem:sub_path_encoding} and Lemma~\ref{lem:forwards_and_backwards_movement} as basic tools, we can proceed to the main proofs. The general strategy in this proof is an inductive approach. Suppose we have already showed that $G'$ encodes the $s$-$u$ subpath of $P$ for some $u$, and $v$ is the first vertex after $u$ in $P$ that belongs to some subset of vertices $U \subseteq V(G')$. We will show that $G'$ also encodes the $s$-$v$ subpath of $P$. We break down all cases to several lemmas depending on where $u$ is in the graph $G'$. The subset $U$ could be different in different lemmas, but it is essential that the set $U$ contains the vertex $t$. 

\begin{lemma}
\label{lem:from_s}
Let $u$ be the first vertex in $V(\Ione) \cup V(\Itwo) \cup \{t\}$ on $P$, then $G'$ encodes the subpath from $s$ to $u$ on $P$. 
\end{lemma}
\begin{proof}
We break down the cases depending on where $u$ is. 

Suppose $u \in V(\Ione)$. The subpath is encoded since we added the edge $(s, u)$ with weight $d_{G \setminus E(\Ione)}(s, u, e_2)$ in step 4. It is sufficient since $u$ is the first vertex on $P$ in $V(\Ione) $, so the $s$ to $u$ subpath will not use any edge in $E(\Ione)$. Thus the subpath will exist in $G \setminus (E(\Ione) \cup \{e_2\})$.

Similarly, suppose $u \in  V(\Itwo)$, then the subpath is encoded since we added the edge $(s, u)$ with weight $d_{G \setminus E(\Itwo)}(s, u, e_1)$ in step 5. It is sufficient since $u$ is the first vertex on $P$ in $ V(\Itwo)$, so the $s$ to $u$ subpath will not use any edge in $E(\Itwo)$. Thus the subpath will exist in $G \setminus (E(\Itwo) \cup \{e_1\})$.

Suppose $u=t$, then the edge from $s$ to $t$ with weight $d_{G \setminus E(\Ione)}(s, t, e_2)$ added in step 1 encodes the subpath since in this case $P$ does not contain any vertex in $V(\Ione) \cup V(\Itwo)$. 
\end{proof}

\begin{lemma}
\label{lem:from_\Beforeone}
Let $u \in \Beforeone$ be a vertex on $P$. Let $v$ be the first vertex after $u$ in $V(\Ione) \cup \Beforetwo \cup \{t\}$ on $P$. Suppose $G'$ encodes the subpath from $s$ to $u$ on $P$, then $G'$ encodes the subpath from $s$ to $v$ as well.
\end{lemma}
\begin{proof}
We break down the cases depending on where $v$ is. 

Suppose $v \in \Beforeone$. When moving within $\Beforeone$, the replacement path can move either forward or backward, as defined previously. The best way to move forward within $\Beforeone$ is to follow the edges on $\pi_G(s, t)$, so by adding the edges in $E(\Ione) \setminus \{e_1\}$ in step 2 we are able to encode the $u$ to $v$ subpath into $G'$. By Lemma~\ref{lem:sub_path_encoding}, $G'$ also encodes the $s$-$v$ subpath. $P$ will never move backward to a vertex in $\Beforeone$, as shown in Lemma~\ref{lem:forwards_and_backwards_movement}.

Suppose $v \in \Afterone$.\footnote{This is the case where we need to assume $P$ is good. If the input graph doesn't have negative edges, then we can afford to compute a full SSRP from $s$ in $G \setminus E(\Itwo)$ instead of an SSRP with a small target set, and thus this case can be handled similarly to the next case by using $d_{G \setminus E(\Itwo)}(s, v, e_1)$. For this reason, our proof can be simplified for graphs with edge weights in $\{1, \ldots, M\}$.} The $u$-$v$ subpath cannot contain any vertex in $\{s\} \cup \Aftertwo$ by Lemma~\ref{lem:forwards_and_backwards_movement}, and cannot contain any vertex (other than $u$ and $v$) in $V(\Ione) \cup \Beforetwo \cup \{t\}$ by definition of $v$. Thus, the $u$-$v$ subpath does not contain any other vertex in $V(\Ione) \cup V(\Itwo) \cup \{s, t\}$. 
Since $P$ is good, the $u$-$v$ subpath does not use any edge on $\pi_G(s, t)$. Thus the $u$-$v$ subpath is encoded by the edges added in step 9. Since the $s$-$u$ and $u$-$v$ subpaths are both encoded, $G'$ encodes the $s$-$v$ subpath by Lemma~\ref{lem:sub_path_encoding}.

Suppose $v \in \Beforetwo$. In this case, we will directly prove that $G'$ encodes the $s$ to $v$ subpath. Since $P$ is canonical, the subpath from $s$ to $u$ must use edges on the $s$-$u$ subpath on $\pi_G(s, t)$. Thus, the $s$-$u$ subpath does not use any edge in $E(\Itwo)$. Since $v$ is the first vertex after $u$ in $V(\Ione) \cup \Beforetwo$ and $u$ is in $\Beforeone$,  the $u$-$v$ subpath can not use any edge in $E(\Itwo)$ before $e_2$ as that would make it so that $v$ is not the first vertex after $u$  that is in $V(\Ione) \cup \Beforetwo$. In addition, the $u$-$v$ subpath can not use any edge after $e_2$ in $E(\Itwo)$
by Lemma~\ref{lem:forwards_and_backwards_movement}. Therefore, the $u$-$v$ subpath can not use any edge in $E(\Itwo)$, so the $s$-$v$ subpath does not use any edge in $E(\Itwo)$, meaning that it will be present in $G \setminus (E(\Itwo) \cup \{e_1\})$. Therefore, the edge from $s$ to $v$ with weight $d_{G \setminus E(\Itwo)}(s, v, e_1)$ added in step 5 encodes the $s$-$v$ subpath.

Finally, suppose $v = t$. The $u$-$v$ subpath does not use any vertex in $V(\Ione)$ other than $u$, so it does not use any edge in $E(\Ione)$. It means that the $u$-$v$ subpath will be present in $G \setminus (E(\Ione) \cup \{e_2\})$, so the edges added in step 5 encode the $u$-$v$ subpath into $G'$.
\end{proof}

\begin{lemma}
\label{lem:from_\Afterone}
Let $u \in \Afterone$ be a vertex on $P$. Let $v$ be the first vertex after $u$ in $V(\Ione) \cup \Beforetwo \cup \{t\}$ on $P$. Suppose $G'$ encodes the subpath from $s$ to $u$ in $P$, then $G'$ encodes the subpath from $s$ to $v$ as well. 
\end{lemma}
\begin{proof}
We break down the cases depending on where $v$ is. Here, we will prove that $G'$ encodes the $u$ to $v$ subpath, which will imply that $G'$ encodes the $s$ to $v$ subpath by Lemma~\ref{lem:sub_path_encoding}. 

First, $v$ cannot be in $\Beforeone$ since $P$ can not move backward to a vertex in $\Beforeone$ by Lemma~\ref{lem:forwards_and_backwards_movement}.

Suppose $v \in \Afterone$, and $v$ is after $u$ on $\pi_G(s, t)$. The best way to move forward within $\Afterone$ is to follow the edges in $E(\Ione) \setminus \{e_1\}$, which are added in step 2, so the $u$-$v$ subpath is encoded in $G'$.

Suppose $v \in \Afterone$, and $v$ is before $u$ on $\pi_G(s, t)$. Assume that $u$ is not $\rightone$. In this case, the $u$-$v$ subpath cannot use any edge before $e_1$ or after $e_2$ on $\pi_G(s, t)$ by Lemma~\ref{lem:forwards_and_backwards_movement}. It cannot use any edge after $u$ and before $e_2$ because otherwise, since $P$ is canonical, it must use the vertex immediately after $u$ on $\pi_G(s, t)$, contradicting to the assumption that $v$ is the first vertex after $u$ in $V(\Ione) \cup \Beforetwo \cup \{t\}$.  It cannot use any edge after $e_1$ and before $u$ since $v$ is the first vertex after $u$ in $V(\Ione)$ on path $P$. Therefore, the $u$-$v$ subpath does not use any edge on $\pi_G(s, t)$. Thus, the edge from $u$ to $v$  with weight $d_{G \setminus \pi_G(s, t)}(u, v)$
added in step 9 encodes the $u$-$v$ subpath for this case. Now, we assume that $u = \rightone$. In this scenario, the $u$-$v$ subpath can not use any edge in $E(\Ione)$ since $v$ is the first vertex after $u$ in $V(\Ione)$. Therefore, the edge from $\rightone$ to $v$ with weight $d_{G \setminus E(\Ione)}(\rightone, v, e_2)$ added in step 7 encodes the $u$-$v$ subpath for this case.

Suppose $v \in \Beforetwo$. Since $P$ is canonical, it must be the case that $u = \rightone$ and $b = \lefttwo$. Therefore, the edge between $\rightone$ and $\lefttwo$ added in step 3 encodes the $u$-$v$ subpath into $G'$.

Finally, suppose $v = t$. The $u$-$v$ subpath won't use any edge in $E(\Ione)$ because it does not use any vertex in $V(\Ione)$ other than $u$. It means that the $u$-$v$ subpath will be present in $G \setminus (E(\Ione) \cup \{e_2\})$, so the edge from $u$ to $t$ with weight $d_{G \setminus E(\Ione)}(u, t, e_2)$ added in step 5 encodes the $u$-$v$ subpath.
\end{proof}

\begin{lemma}
\label{lem:from_\Beforetwo}
Let $u \in \Beforetwo$ be a vertex on $P$. Let $v$ be the first vertex after $u$ in $V(\Ione) \cup \Beforetwo \cup \{t\}$ on $P$. Suppose $G'$ encodes the subpath from $s$ to $u$ in $P$, then $G'$ encodes the subpath from $s$ to $v$ as well. 
\end{lemma}
\begin{proof}
We break down the cases depending on where $v$ is. Here, (for all but the last case) we will prove that $G'$ encodes the $u$ to $v$ subpath, which will imply that $G'$ encodes the $s$ to $v$ subpath by Lemma~\ref{lem:sub_path_encoding}. 

First, $v$ cannot be in $\Beforeone$ since $P$ can not move backward to a vertex in $\Beforeone$ by Lemma~\ref{lem:forwards_and_backwards_movement}.

Suppose $v \in \Afterone$. Since $v$ is the first vertex after $u$ in $\Afterone \cup \Beforetwo$, the $u$-$v$ subpath does not use any edge in $E(\Ione)$ after $e_1$ or any edge in $E(\Itwo)$ before $e_2$. In addition, the $u$-$v$ subpath can not use any edge on $\pi_G(s, t)$ before $e_1$ or after $e_2$ by Lemma~\ref{lem:forwards_and_backwards_movement}. Overall, the $u$-$v$ subpath lies completely in the graph $G \setminus \pi_G(s, \rightone) \setminus \pi_G(\lefttwo, t) \subseteq G \setminus \pi_G(s, v) \setminus \pi_G(u, t)$. By the optimality of $P$, the length of the $u$-$v$ subpath equals $d_{G \setminus \pi_G(s, v) \setminus \pi_G(u, t)}(u, v) = f(u, v)$.  Thus, the edge from  $u$ to $v$ with weight $f(u, v)$ added in step 6 encodes the $u$-$v$ subpath.

Suppose $v \in \Beforetwo$, and $v$ is after $u$ on $\pi_G(s, t)$. The best way to move forward within $\Beforetwo$ is to follow the edges on $E(\Itwo) \setminus \{e_2\}$, which are added in step 2, so the $u$-$v$ subpath is encoded in $G'$.

Suppose $v \in \Beforetwo$, and $v$ is before $u$ on $\pi_G(s, t)$. Assume that $v$ is not $\lefttwo$. 
In this case, the $u$-$v$ subpath cannot use any edge before $e_1$ or after $e_2$ on $\pi_G(s, t)$ by Lemma~\ref{lem:forwards_and_backwards_movement}. It cannot use any edge after $e_1$ and before $v$ on $\pi_G(s, t)$ because otherwise, since $P$ is canonical, the vertex right before $v$ on $\pi_G(s, t)$ is also used, contradicting to the assumption that $v$ is the first vertex after $u$ in $V(\Ione) \cup \Beforetwo \cup \{t\}$. It cannot use any edge after $v$ and before $e_2$ on $\pi_G(s, t)$ since $v$ is the first vertex after $u$ in $V(\Ione) \cup \Beforetwo \cup \{t\}$ on path $P$. Therefore, the $u$-$v$ subpath does not use any edge on $\pi_G(s, t)$. 
Thus, the edge from $u$ to $v$ with weight $d_{G \setminus \pi_G(s, t)}(u, v)$ added in step 9 encodes the $u$-$v$ subpath for this case. Now, we assume that $v = \lefttwo$. In this scenario, the $u$-$v$ subpath can not use any edge in $E(\Itwo)$ since $v$ is the first vertex after $u$ in $V(\Ione) \cup \Beforetwo \cup \{t\}$. Therefore, the edge from $u$ to $\lefttwo$ with weight $d_{G \setminus E(\Itwo)}(u, \lefttwo, e_1)$
added in step 8 encodes the $u$-$v$ subpath for this case.

Finally, suppose $v = t$. This case is slightly more complicated.\footnote{For graphs with only positive edge weights this case can also be simplified by using $d_{G \setminus E(\Ione)}(u, t, e_2)$ from full SSRP computations instead of SSRP with small target sets.} First, by Lemma~\ref{lem:forwards_and_backwards_movement}, the $u$-$t$ subpath cannot contain any vertex on $\pi_G(s, t)$ before $e_1$. Also, by definition of $v$, it cannot contain any vertex in $\Afterone$ or $\Beforetwo$ either. We further divide this case based on which vertices this subpath contains. 
\begin{enumerate}
    \item The subpath contains some vertex $w$ after $\rightone$ and before $\lefttwo$ on $\pi_G(s, t)$. In this case, the $s$-$u$ subpath cannot contain any vertex $x$ in $\Afterone$ since otherwise the canonical path $P$ should move directly from $x$ to $w$ instead of going to $u$ first. The $u$-$t$ subpath cannot contain any vertex $x$ in $\Afterone$ either as argued previously. Thus, the whole path $P$ does not contain any vertex in $\Afterone$. 
    
    If $P$ also doesn't contain any vertex in $\Beforeone$, then $P$ doesn't contain any vertex or edge in $\Ione$, so the $s$-$t$ path is encoded by $d_{G \setminus E(\Ione)}(s, t, e_2)$ added in step 1. If $P$ contains some vertex in $\Beforeone$, let $y$ be the last of them. The $s$-$y$ subpath is encoded by the edges $d_G(s, \leftone)$ and the edges $E(\Ione) \setminus \{e_1\}$, which were added in step 3 and 2 respectively. The $y$-$t$ subpath thus doesn't contain any edge in $E(\Ione)$, so it is encoded by $d_{G \setminus E(\Ione)}(y, t, e_2)$ added in step 4. Thus, the $s$-$t$ subpath is encoded by Lemma~\ref{lem:sub_path_encoding}.
    \item The subpath doesn't contain any vertex after $\rightone$ and before $\lefttwo$, but contains some vertex in $\Aftertwo$. Let $w$ be the first vertex in $\Aftertwo$ on this subpath. By the assumption and previous discussions,  the $u$-$w$ subpath cannot contain any edge before $e_2$ on $\pi_G(s, t)$; it cannot contain any edge after $e_2$ on $E(\Itwo)$ because $w$ is the first vertex in $\Aftertwo$; it cannot contain any edge after $\Itwo$ by Lemma~\ref{lem:forwards_and_backwards_movement}. Thus, the $u$-$w$ subpath doesn't contain any edge on $\pi_G(s, t)$ and thus is encoded by  $d_{G \setminus \pi_G(s, t)}(u, w)$  added in step 9. The $w$-$t$ subpath is clearly encoded by edges $E(\Itwo) \setminus \{e_2\}$ and the distance $d_G(\righttwo, t)$ added in step 2 and step 3 respectively. Thus, $G'$ encodes the $u$-$t$ subpath by Lemma~\ref{lem:sub_path_encoding}. Since we assume that $G'$ encodes the $s$-$u$ subpath, it also encodes the $s$-$t$ subpath. 
    \item Finally, we assume the subpath doesn't contain any vertex in $\Aftertwo$. These three cases cover all possibilities. In this case, the $u$-$t$ subpath does not use any edge in $E(\Itwo)$, meaning it will be present in $G \setminus (E(\Itwo) \cup \{e_1\})$, so the edge from $u$ to $t$ with weight $d_{G \setminus E(\Itwo)}(u, t, e_2)$ 
added in step 5 encodes the $u$-$v$ subpath.
\end{enumerate}
\end{proof}

\begin{lemma}
\label{lem:from_\Aftertwo}
Let $u \in \Aftertwo$ be a vertex on $P$. Let $v$ be the first vertex after $u$ in $\Aftertwo \cup \{t\}$ on $P$. Suppose $G'$ encodes the subpath from $s$ to $u$ in $P$, then $G'$ encodes the subpath from $s$ to $v$ as well. 
\end{lemma}
\begin{proof}
Suppose $v \in \Aftertwo$. By Lemma~\ref{lem:forwards_and_backwards_movement}, $u$ can only move forward to $v$. The best way to move forward within $\Aftertwo$ is to follow the edges in $E(\Itwo) \setminus \{e_2\}$, which are added in step 2, so the $u$-$v$ subpath is encoded in $G'$.

If $v = t$, then since $P$ is canonical, $u$ must equal to $\righttwo$, and the $u$-$v$ subpath is encoded by the edge from $\righttwo$ to $t$ added in step 3.
\end{proof}

Now, we can prove Theorem \ref{theorem:DifferentIntervalCorrectness}:
\begin{proof}[Proof of Theorem~\ref{theorem:DifferentIntervalCorrectness}]
Let $P$ be a good canonical shortest path from $s$ to $t$ in $G \setminus \{e_1, e_2\}$. We will show longer and longer prefixes of $P$ are encoded in $G'$ and finally that $P$ is encoded in $G'$. 

Initially, we set $u = s$ and clearly the subpath from $s$ to $u$ is encoded in $G'$. Now depending on where $u$ is, we can apply one of Lemma~\ref{lem:from_s}, Lemma~\ref{lem:from_\Beforeone}, Lemma~\ref{lem:from_\Afterone}, Lemma~\ref{lem:from_\Beforetwo} and Lemma~\ref{lem:from_\Aftertwo}. Using those lemmas, we find a vertex $v$ on $P$ after $u$ in a certain subset of vertices $U \subseteq V(G')$. The subset $U$ can be different depending on which lemma we are applying, but $U$ always contains $t$, so we can always find such a vertex $v$. By applying the corresponding lemma, and using the inductive assumption that $G'$ encodes the $s$ to $u$ subpath, we obtain that $G'$ encodes the $s$ to $v$ subpath. Thus, we could set $u$ to $v$ and keep applying the lemmas. Eventually $u$ will be set to $t$.

We have shown that $G'$ encodes $P$, so the shortest $s$-$t$ path in $G'$ can not be longer than the replacement path in $G\setminus \{e_1, e_2\}$. It is impossible for $d_{G'}(s, t)$ to be smaller than $d_{G\setminus \{e_1, e_2\}}(s, t)$  because all of the edges in $G'$ have weights that correspond to the lengths of some paths that are present in $G\setminus \{e_1, e_2\}$. Therefore, $d_{G'}(s, t)$ must be equal to $d_{G\setminus \{e_1, e_2\}}(s, t)$.
\end{proof}

\subsection{Putting It All Together}

Now we have all the necessary components for proving Theorem~\ref{thm:main}, which is recalled here:

\main*

\begin{proof}
In total, the running time for the general precomputation phase and the precomputation phase of each sub-algorithm is $\tO(M n^{\omega+1}/g + M^{1/(4 - \omega)}n^{2 + 1/(4-\omega)} + M^{1/3}n^{2 + \omega/3} + Mn^{2.8729})$. The fourth term dominates the second and third term, so the pre-processing time simplifies to $\tO(M n^{\omega+1}/g + Mn^{2.8729})$. The space complexity is $\tO(n^{2.5})$, where the space of the DSO is the bottleneck. 

If all edge weights are positive, the pre-processing time can be improved to $\tO(M n^{\omega+1}/g + M^{1/(4 - \omega)} \allowbreak n^{2 + 1/(4-\omega)} + M^{1/3}n^{2 + \omega/3} + Mn^{2.5794})$. The third term can be improved to $\tO(M^{0.3544} n^{2.7778})$ using rectangular matrix multiplication by Lemma~\ref{lem:computing_f}. Note that the second term is always dominated by the third term or the fourth term, so the pre-processing time simplifies to $\tO(M n^{\omega+1}/g +  M^{0.3544} n^{2.7778} + Mn^{2.5794})$. The space complexity is $\tO(n^2)$. 

The query time is $\tO(g^2)$ in both cases, since we always run the near-linear time SSSP algorithm \cite{bernstein2022negative} on an auxiliary graph with $O(g)$ vertices. 
\end{proof}

We can easily obtain our algorithm for $2$FRP from Theorem~\ref{thm:main} by setting the parameter appropriately. Recall Corrollary~\ref{cor: constant_query_time}:

\mainBounded*

\begin{proof}
We run the near-linear time SSSP algorithm \cite{bernstein2022negative}  to find $\pi_G(s, t)$ in $\tO(n^2)$ time, and then run the $\tO(Mn^\omega)$ time RP algorithm by Vassilevska Williams~\cite{williams2011faster}  to find $\pi_{G \setminus \{e\}}(s, t)$ for every $e \in \pi_G(s, t)$. 

Then we can easily generate all $(e_1, e_2)$ pairs such that $e_1 \in \pi_G(s, t)$ and $e_2 \in \pi_{G \setminus \{e_1\}}(s, t)$ in $O(n^2)$ time. Using Theorem~\ref{thm:main}, it will take $\tO(g^2n^2 + Mn^{\omega+1}/g + Mn^{2.8729})$ time to handle all the queries, which is $$\tO(M^{2/3}n^{2(\omega+2)/3} + Mn^{2.8729})$$ by setting $g = M^{1/3} n^{(\omega - 1) / 3}$. Using the current upper bound $\omega < 2.3729$, the running time becomes $\tO(M^{2/3} n^{2.9153} + Mn^{2.8729})$. Note that the second term is larger than the first term only when $M^{2/3} n^{2.9153} = \Omega(n^3)$, so we can just run the $\tO(n^3)$ time algorithm from Theorem~\ref{thm:main_weighted} in this case. 
Thus, the running time is always $\tO(\min(M^{2/3} n^{2.9153}, n^3)) = \tO(M^{2/3} n^{2.9153})$. 

Finally, we need to check $g=O(n)$ by the requirement of Theorem~\ref{thm:main}.
Note that this value of $g$ is $O(n)$ when $M =O(n^{4-\omega})$. When $M = \Omega(n^{4-\omega})$, our claimed running time exceeds $\Omega(n^3)$ and thus we can just run the $\tO(n^3)$ time algorithm from Theorem~\ref{thm:main_weighted}. 
\end{proof}

We immediately obtain Corollary~\ref{cor: constant_query_time_f}:
\mainBoundedf*

\begin{proof}

We first compute a shortest path $P_1$ from $s$ to $t$ in $\tilde{O}(n^2)$ time, by using \cite{bernstein2022negative}. For every edge $e_1$ on $P_1$, we compute a shortest $s$-$t$ path $P_2$ from $s$ to $t$ in $G\setminus \{e_1\}$. More generally, for each $i$, and each choice of $(e_1,\ldots,e_{i})$ and computed paths $P_1,\ldots,P_{i}$ where each $P_j$ is a shortest $s$-$t$ path in $G\setminus \{e_1,\ldots,e_{j-1}\}$ and $e_j\in P_j$, we compute a shortest $s$-$t$ path $P_{i+1}$ in $G\setminus \{e_1,\ldots,e_{i}\}$. This computation takes $\tilde{O}(n^f)$ time. Then for each of the $O(n^{f-2})$ choices of $(e_1,\ldots,e_{f-2})$, we compute $2$FRP using Corollary~\ref{cor: constant_query_time} in $G\setminus \{e_1,\ldots,e_{i}\}$ in overall time $\tilde{O}(M^{2/3} n^{f+0.9153})$.
\end{proof}

\section{Computing Backwards Distances}
\label{sec:backwards}
In this section, we describe algorithms for computing the shortest path distances between all pairs of vertices on the path $\pi_G(s, t)$ that can only use edges not on $\pi_G(s, t)$ and edges on $\pi_G(s, t)$ between these two vertices. More formally, for every two vertices $u, v$ on the path $\pi_G(s, t)$ where $u$ appears earlier than $v$ on the path, we need to compute $f(v, u) = d_{G \setminus  \pi_G(s, u) \setminus \pi_G(v, t)}(v, u)$. Though $ \pi_G(s, u)$ and $\pi_G(v, t)$ might not be unique, we assume they are the subpath between $s$ and $u$ on $\pi_G(s, t)$ and the subpath between $v$ and $t$ on $\pi_G(s, t)$ respectively. 

In this section, let $h=O(n)$ be the number of vertices on $\pi_G(s, t)$ and let $p_1, \ldots, p_h$ be vertices on the path $\pi_G(s, t)$, in the order they appear on $\pi_G(s, t)$. In particular, $s=p_1$ and $t = p_h$. 
For two vertices $u, v$ on $\pi_G(s, t)$, we use $u \prec v$ to indicate that $u$ appears before $v$ on $\pi_G(s, t)$ and use $u \succ v$ to indicate that $u$ appears after $v$ on $\pi_G(s, t)$. 

Before we give the details of the two algorithms, we first show how an optimal path for $f(v, u)$ could look like. 
Figure~\ref{fig:f_structure} pictures the structure of a typical optimal canonical path in $d_{G \setminus  \pi_G(s, u) \setminus \pi_G(v, t)}(v, u)$. 

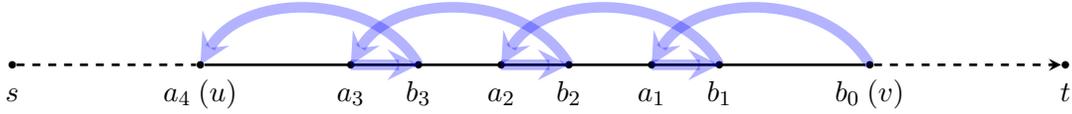
\begin{figure}[ht]
    \centering
    \begin{tikzpicture}
    	\node at (0,0) [circle,fill, inner sep = 1pt, label=below:$s\vphantom{b_0\ (v)}$] (s){};
    	\node at (14,0) [circle,fill, inner sep = 1pt, label=below:$t\vphantom{b_0\ (v)}$] (t){};

	\node at (2.5, 0) [circle, fill, inner sep = 1pt, label = below:$a_4\ (u)\vphantom{b_0\ (v)}$](u){};
	\node at (11.4, 0) [circle, fill, inner sep = 1pt, label = below:$b_0\ (v)\vphantom{b_0\ (v)}$](v){};

	\draw [-, dashed, line width = 1pt] (s) to[] (u);
	\draw [-, line width = 1pt] (u) to[] (v);
	\draw [-stealth, dashed, line width = 1pt] (v) to[] (t);
	
	\node at (8.5, 0) [circle, fill, inner sep = 1pt, label = below:$a_1\vphantom{b_0\ (v)}$](){};
	\node at (9.4, 0) [circle, fill, inner sep = 1pt, label = below:$b_1\vphantom{b_0\ (v)}$](){};
	\node at (6.5, 0) [circle, fill, inner sep = 1pt, label = below:$a_2\vphantom{b_0\ (v)}$](){};
	\node at (7.4, 0) [circle, fill, inner sep = 1pt, label = below:$b_2\vphantom{b_0\ (v)}$](){};
	\node at (4.5, 0) [circle, fill, inner sep = 1pt, label = below:$a_3\vphantom{b_0\ (v)}$](){};
	\node at (5.4, 0) [circle, fill, inner sep = 1pt, label = below:$b_3\vphantom{b_0\ (v)}$](){};

	\begin{scope}[transparency group, opacity=0.3, text opacity=1]
		\draw[-stealth,line width=4pt, blue, bend right = 60] (11.4, 0) to[] (8.5, 0);
		\draw[-stealth,line width=4pt, blue] (8.5, 0) to[] (9.4, 0);
		\draw[-stealth,line width=4pt, blue, bend right = 60] (9.4,0) to[] (6.5, 0);
		\draw[-stealth,line width=4pt, blue] (6.5, 0) to[] (7.4, 0);
		\draw[-stealth,line width=4pt, blue, bend right = 60] (7.4,0) to[] (4.5, 0);
		\draw[-stealth,line width=4pt, blue] (4.5, 0) to[] (5.4, 0);
		\draw[-stealth,line width=4pt, blue, bend right = 60] (5.4,0) to[] (2.5, 0);
	\end{scope}
	
    \end{tikzpicture}
    \caption{An example of a canonical path in $G \setminus \pi_G(s, u) \setminus \pi_G(v, t)$ where $r=3$. All paths shown that do not lie on the original shortest path are shortest paths in $G \setminus \pi_G(s, t)$.}
    \label{fig:f_structure}
\end{figure}

More formally: 
\begin{lemma}
\label{lem:f_structure}
Let $P$ be a canonical shortest path from $v$ to $u$ in $ G \setminus  \pi_G(s, u) \setminus \pi_G(v, t)$. Then we can write $P$ as $v = b_0 \rightarrow a_1 \rightarrow b_1 \cdots \rightarrow a_r \rightarrow b_r \rightarrow u = a_{r+1}$, where the part from $b_k$ to $a_{k+1}$ for each $k \in [0, r]$ are shortest paths in $G \setminus \pi_G(s, t)$, and the part from $a_k$ to $b_k$ are edges from $a_k$ to $b_k$ on $\pi_G(s, t)$ for each $k \in [1, r]$. Also, the order these vertices appear on $\pi_G(s, t)$ is $a_{r+1}, a_r, b_r, a_{r-1}, b_{r-1}, \ldots, a_1, b_1, b_0$, i.e. $a_{r+1} \prec a_r \prec b_r \prec a_{r-1}  \prec b_{r-1}  \prec \ldots \prec a_1 \prec b_1 \prec b_0$.
\end{lemma}
\begin{proof}
We break down $P$ to intervals that alternatively use edges on $\pi_G(s, t)$ and edges not on $\pi_G(s, t)$. Also, the first interval is an interval that does not use edges on $\pi_G(s, t)$ since $v$ cannot move directly to the next vertex on $\pi_G(s, t)$. Similarly, the last interval must be an interval that does not use edges on $\pi_G(s, t)$ as well. The first part of the lemma follows by setting $a_i$ or $b_i$ to boundaries of these intervals accordingly. 

Next, we show $a_{r + 1} \prec a_r$. If $a_r \prec a_{r+1}$, then it's impossible for $a_r$ to move to $b_r$ via edges on $\pi_G(s, t)$ since all edges before $a_r$ are removed. If $a_r = a_{r+1}$, then since $P$ is canonical, the portion between $a_r$ and $a_{r+1}$ should be empty, so it should not move to $b_r$ either. Thus, $a_{r+1} \prec a_r$. Similarly, we can show $b_1 \prec b_0$.

It is clear that $a_i \prec b_i$ for $i \in [r]$, since $a_i$ uses edges on $\pi_G(s, t)$ to move to $b_i$. It remains to show $b_i \prec a_{i-1}$ for $i \in \{2, \ldots, r\}$. For the sake of contradiction, suppose $a_{i-1} \prec b_i$ or $a_{i-1} = b_i$. 
The portion from $a_{i-1}$ to $b_i$ on $\pi_G(s, t)$ exists in $G \setminus \pi_G(s, u) \setminus \pi_G(v, t)$ because the part $a_{i-1} \rightarrow b_{i-1}$ indicates the edge after $a_{i-1}$ exists and the part $a_i \rightarrow b_i$ indicates the edge before $b_i$ exists, so $a_{i-1}$ and $b_i$ are both between $u$ and $v$ on $\pi_G(s, t)$. Then since $P$ is canonical, the portion of $P$ from $a_{i-1}$ to $b_i$ should only use edges on $\pi_G(s, t)$, so $b_i$ should have been $b_{i-1}$, which leads to a contradiction. 
\end{proof}

\subsection{Algorithm for Weighted Graphs}

We first show how to compute $f$ in general weighted graphs in $O(n^3)$ time deterministically: 

\computingFWeighted*

\begin{proof}
First, we run APSP on $G$ and find $\pi_G(s, t)$. Then we run APSP on $G \setminus \pi_G(s, t)$. These take $O(n^3)$ time. 

Recall $p_1, \ldots, p_h$ are vertices on the path $\pi_G(s, t)$, in the order they appear on $\pi_G(s, t)$. 

Fix an canonical shortest path from $p_j$ to $p_i$ for $j > i$. By applying Lemma~\ref{lem:f_structure}, we can write the path as $p_j = b_0 \rightarrow a_1 \rightarrow b_1 \cdots \rightarrow a_r \rightarrow b_r \rightarrow p_i = a_{r+1}$, where the part from $b_k$ to $a_{k+1}$ for each $k \in [0, r]$ are shortest paths in $G \setminus \pi_G(s, t)$, and the part from $a_k$ to $b_k$ are edges from $a_k$ to $b_k$ on $\pi_G(s, t)$ for each $k \in [1, r]$. 

If $r=0$, then the whole path is $b_0 \rightarrow a_1$, and thus the length of the path is $d_{G \setminus \pi_G(s, t)}(p_j, p_i)$. Otherwise, the length of the path is $f(p_j, a_r) + d_G(a_r, b_r) + d_{G \setminus \pi_G(s, t)}(b_r, p_i)$. Therefore, $f(p_j, p_i)$ can be expressed as
\begin{equation*}
\begin{split}
f(p_j, p_i) &=
\min\left\{d_{G \setminus \pi_G(s, t)}(p_j, p_i), 
\min_{p_i \prec a_r \prec b_r \prec p_j} \{ f(p_j, a_r) + d_G(a_r, b_r) + d_{G \setminus \pi_G(s, t)}(b_r, p_i)\}
\right\},
\end{split}
\end{equation*}
where the second term in the minimization can be further written as \begin{equation*}
\begin{split}
&\min_{p_i \prec a_r \prec b_r \prec p_j} \left\{ f(p_j, a_r) - d_G(s, a_r) +d_G(s, b_r) + d_{G \setminus \pi_G(s, t)}(b_r, p_i) \right\}\\
=& \min_{p_i \prec a_r \prec p_j} \left\{ f(p_j, a_r) - d_G(s, a_r) + \min_{a_r \prec b_r \prec p_j} \left\{d_G(s, b_r) + d_{G \setminus \pi_G(s, t)}(b_r, p_i) \right\} \right\}.
\end{split}
\end{equation*}
This formula suggests a simple dynamic programming approach. Suppose we have computed $f(p_j, a_r)$ for all $a_r \succ p_i$, then we can compute $f(p_j, p_i)$ in $O(n)$ time. Initially, we set $f(p_j, p_i)$ to $d_{G \setminus \pi_G(s, t)}(p_j, p_i)$. Then, we enumerate all $a_r$ between $p_i$ and $p_j$ on $\pi_G(s, t)$ from closest to $p_j$ to closest to $p_i$. During the enumeration, we can efficiently maintain $\min_{a_r \prec b_r \prec p_j} \{d_G(s, b_r) + d_{G \setminus \pi_G(s, t)}(b_r, p_i)\}$, since each time we move to the next $a_r$, there is only one additional $b_r$ that contributes to this minimization. Thus, for each $a_r$, we can update $f(p_j, p_i)$ with $\left\{f(p_j, p_i), f(p_j, a_r) - d_G(s, a_r) + \min_{a_r \prec b_r \prec p_j} \left\{d_G(s, b_r) + d_{G \setminus \pi_G(s, t)}(b_r, p_i) \right\}\right\}$ in $O(1)$ time. 

Overall, the running time of this algorithm is $O(n^3)$. 
\end{proof}

\subsection{Algorithm for Bounded-Weight Graphs}

The general approach for computing $f$ in graphs with weights in $\{-M, \ldots, M\}$ is analogous to Zwick's algorithm \cite{Zwick02} for computing all pairs shortest paths in small integer weighted graphs. Similar to Zwick's algorithm, when the distances are large, we use the hitting set idea; when the distances are small, we use Min-Plus product between two matrices with small integer weights as a subroutine.

We will iteratively compute $f^{\le \ell}$, which is defined as $f^{\le \ell}(v, u) = f(v, u)$ if one of the shortest paths in $G \setminus  \pi_G(s, u) \setminus \pi_G(v, t)$ from $v$ to $u$ with weight $f(v, u)$ has at most $\ell$ vertices, and $f^{\le \ell}(v, u)$ could take any value at least $f(v, u)$ otherwise. 
Clearly, $f^{\le \ell}(v, u) \ge -\ell M$, and if $f^{\le \ell}(v, u) > \ell M$, we could instead set $f^{\le \ell}(v, u)$ to $\infty$ and the new value would still satisfy the definition. Thus, without loss of generality, we assume $f^{\le \ell}(v, u)$ has values in $\{-\ell M, \ldots, \ell M\} \cup \{\infty\}$. 

Once we compute $f^{\le \ell}$, we will use its values to compute $f^{\le \frac{3}{2}\ell}$. Similar to Zwick's algorithm, depending on the value of $\ell$, we will use either a hop-short algorithm or a hop-long algorithm. 

We first consider the hop-long algorithm.

\begin{lemma}
\label{lem:long-hop}
Given $f^{\le \ell}$, $d_G$ and $d_{G \setminus \pi_G(s, t)}$, we can compute $f^{\le \frac{3}{2} \ell}$  in $\tO(n^3 / \ell)$ time with high probability.
\end{lemma}
\begin{proof}

Let $h=O(n)$ be the number of vertices on $\pi_G(s, t)$ and let $p_1, \ldots, p_h$ be vertices on the path $\pi_G(s, t)$, in the order they appear on $\pi_G(s, t)$. 

We randomly sample a set of vertices $S \subseteq V$ of size $cn \log n/\ell$ for a sufficiently large constant $c$. 
Fix any path from $p_j$ to $p_i$ for $j > i$ with at most $\frac{3}{2} \ell$ vertices. With high probability, $S$ contains a vertex $x$ in the ``middle-third'' of the path, i.e., the number of vertices from $p_j$ to $x$ is at most $\ell$ and the number of vertices from $x$ to $p_i$ is at most $\ell$.

By Lemma~\ref{lem:f_structure}, we can write the path as $p_j = b_0 \rightarrow a_1 \rightarrow b_1 \cdots \rightarrow a_r \rightarrow b_r \rightarrow p_i = a_{r+1}$, where the parts from $b_k$ to $a_{k+1}$ for each $k \in [0, r]$ are shortest paths in $G \setminus \pi_G(s, t)$, and the part from $a_k$ to $b_k$ are edges from $a_k$ to $b_k$ on $\pi_G(s, t)$ for each $k \in [1, r]$. There are two cases depending on where $x$ is. 

\paragraph{$x$ is on the part $a_k \rightarrow b_k$ for some $k \in [r]$.} In this case, the number of vertices from $p_j$ to $a_k$ is at most $\ell$ and the number of vertices from $b_k$ to $p_i$ is at most $\ell$, so these two subpaths are captured by $f^{\le \ell}$. To resolve this case, we precompute a table $A$ whose first two dimensions are indexed by vertices on $\pi_G(s, t)$ and the third dimension is indexed by vertices in $S$ that are on the path $\pi_G(s, t)$. The entries of $A$ are defined as
\[ A[p_i, p_j, x] = \min_{a_k \text{ is between } p_i \text{ and } x} \left\{f^{\le \ell}(p_j, a_k) + d_G(a_k, x)\right\}.\]
Since $d_G(a_k, x) = d_G(s, x) - d_G(s, a_k)$, we can rewrite the above as 
\[ A[p_i, p_j, x] = d_G(s, x) + \min_{a_k \text{ is between } p_i \text{ and } x} \left\{ f^{\le \ell}(p_j, a_k) - d_G(s, a_k) \right\}.\]
Thus, for each $p_j$, we could pre-process all values of $f^{\le \ell}(p_j, a_k) - d_G(s, a_k)$ so that computing $A[p_i, p_j, x]$ becomes one single range minimum query. Thus, it takes $\tO(n^3 / \ell)$ time to compute the whole table $A$. We can similarly compute the following table $B$ in $\tO(n^3 / \ell)$ time as well:
\[ B[p_i, p_j, x] = \min_{b_k \text{ is between } x \text{ and } p_j} \left\{d_G(x, b_k) + f^{\le \ell}(b_k, p_i)\right\}.\]
With the tables $A$ and $B$, we can resolve this case by setting 
\[f^{\le \frac{3}{2} \ell}(p_j, p_i) = \min_{\substack{x \in S \\ x \text{ on } \pi_G(s, t)  \text{ between } p_i \text{ and } p_j}} \{A[p_i, p_j, x] + B[p_i, p_j, x]\},\] which takes $\tO(n^3 / \ell)$ time in total to compute. 

\paragraph{$x$ is on the part $b_{k-1} \rightarrow a_k$ for some $k \in [r+1]$.} We will proceed similarly to the previous case. Since $x$ is in the middle one third of the path, either $k=1$ and thus $b_{k-1} = p_j$, or the number of vertices from $p_j$ to $a_{k-1}$ is at most $\ell$. Because of this, we can create a table $A$ whose first two dimensions are indexed by vertices on $\pi_G(s, t)$ and the third dimension is indexed by vertices in $S$. The entries of $A$ are defined as
\[
\begin{split}
A[p_i, p_j, x] = \min\left\{ \vphantom{\min_{p_i < a_{k-1} < b_{k-1} < p_j} f^{\le \ell}(p_j, a_{k-1}) + d_G(a_{k-1}, b_{k-1}) + d_{G\setminus \pi_G(s, t)}(b_{k-1}, x)} \right.& d_{G\setminus \pi_G(s, t)} (p_j, x),\\
&\left. \min_{p_i \prec a_{k-1} \prec b_{k-1} \prec p_j} \left\{f^{\le \ell}(p_j, a_{k-1}) + d_G(a_{k-1}, b_{k-1}) + d_{G\setminus \pi_G(s, t)}(b_{k-1}, x)\right\}\right\}.
\end{split}
\]
The $d_{G\setminus \pi_G(s, t)} (p_j, x)$ term accounts for the case where $k=1$, and is easy to compute. To compute the second part the of above formula, we could first compute the following table $T$: 
\[
\begin{split}
T[a_{k-1}, p_j, x] &=   \min_{a_{k-1} \prec b_{k-1} \prec p_j} \left\{f^{\le \ell}(p_j, a_{k-1}) + d_G(a_{k-1}, b_{k-1}) + d_{G\setminus \pi_G(s, t)}(b_{k-1}, x)\right\}\\
&= f^{\le \ell}(p_j, a_{k-1}) - d_G(s, a_{k-1}) + \min_{a_{k-1} \prec b_{k-1} \prec p_j} \left\{ d_G(s, b_{k-1}) + d_{G\setminus \pi_G(s, t)}(b_{k-1}, x) \right\}.
\end{split}
\]
For each $x$, we could pre-process all values of $d_G(s, b_{k-1}) + d_{G\setminus \pi_G(s, t)}(b_{k-1}, x)$, so that computing each entry of table $T$ costs one range minimum query. Thus, it takes $\tO(n^3 /\ell)$ time to compute table $T$. Once we have table $T$, we could use range minimum queries again to compute table $A$, which will take $\tO(n^3 /\ell)$ time as well. 

We could similarly compute the following table in $\tO(n^3 /\ell)$ time
\[
B[p_i, p_j, x] = \min\left\{d_{G\setminus \pi_G(s, t)} (x, p_i),  \min_{p_i \prec a_{k} \prec b_{k} \prec p_j} \left\{d_{G\setminus \pi_G(s, t)}(x, a_k) + d_G(a_{k}, b_{k}) + f^{\le \ell}(b_k, p_i)\right\}\right\}.
\]
Finally, we set $f^{\le \frac{3}{2} \ell}(p_j, p_i) = \min_{x \in S} \left\{A[p_i, p_j, x] + B[p_i, p_j, x] \right\}$.
\end{proof}

Now we turn our attention to the hop-short algorithm. 
We define a diamond product between two $h \times h$ matrices $A$ and $B$ as follows:
\[(A \diamond_\ell B)_{p_j, p_i} = 
\begin{dcases*}
\min\left\{A_{p_j, p_i}, B_{p_j, p_i}, \min_{\substack{i < k < k' < j \\ k' - k < \ell}} \left\{ A_{p_{k'},p_i} + d_G(p_k, p_{k'}) + B_{p_j, p_k}  \right\}\right\} & if $j > i$,\\
d_G(p_j, p_i) & otherwise.
\end{dcases*}
\]
Intuitively, when $j > i$, $(A \diamond_\ell B)_{p_j, p_i}$ is the shortest path from $p_j$ to $p_i$ that either solely uses the path in $A$ or $B$, or combines two paths from $A$ and $B$ by first traveling from $p_j$ to $p_k$ using the distances from matrix $B$, then traveling from $p_k$ to $p_{k'}$ using edges on $\pi_G(s, t)$, and finally traveling from $p_{k'}$ to $p_i$ using distances from matrix $A$. The product enumerates all possible $k <k'$ between $i$ and $j$, with the additional constraint that the number of vertices between $k$ and $k'$ is at most $\ell$.
When $j \le i$, $(A \diamond B)_{p_j, p_i}$ is defined as $d_G(p_j, p_i)$, which equals  $d_{G \setminus  \pi_G(s, u) \setminus \pi_G(v, t)}(p_j, p_i)$ since $\pi_G(p_j, p_i)$ only uses edges between $p_j$ and $p_i$. 

The following lemma is key to our hop-short algorithm. In the following $d_{G \setminus \pi_G(s, t)}^{\le D}$ is defined as $d_{G \setminus  \pi_G(s, t)}^{\le D}(v, u) = d_{G \setminus  \pi_G(s, t)}(v, u)$ if $ \left| d_{G \setminus \pi_G(s, t)}(v,u) \right| \le D$ and $d_{G \setminus  \pi_G(s, t)}^{\le D}(v, u) = \infty$ otherwise. 

\begin{lemma}
\label{lem:short_hop_correct}
It is valid to set $f^{\le \frac{3}{2} \ell}$ to $(f^{\le \ell} \diamond_{\frac{3}{2} \ell} d_{G \setminus \pi_G(s, t)}^{\le \frac{3}{2} \ell M}) \diamond_{\frac{3}{2} \ell} f^{\le \ell}$.
\end{lemma}

Intuitively, this lemma says that (modulo corner cases) a shortest path for $f(u, v)$ with at most $\frac{3}{2} \ell$ vertices can be decomposed into three subpaths, where the first and last subpaths are valid paths for $f(u, v)$ with at most $\ell$ vertices and the middle subpath is a shortest path in $G \setminus \pi_G(s, t)$. 

\begin{proof}

For convenience, we use $d'$ to denote $d_{G \setminus \pi_G(s, t)}^{\le \frac{3}{2} \ell M}$  and we omit the $\frac{3}{2} \ell$ subscript of the diamond product in this proof. 

First, it is easy to see that $\left( (f^{\le \ell} \diamond d') \diamond f^{\le \ell} \right)_{v, u}$ will never be smaller than $d_{G \setminus  \pi_G(s, u) \setminus \pi_G(v, t)}(v, u)$, since $f^{\le \ell}$ and $d'$ are both lower bounded by the distances in the graph $G \setminus  \pi_G(s, u) \setminus \pi_G(v, t)$. Thus, it suffices to show that if there is a shortest path from $v$ to $u$ that uses at most $\frac{3}{2} \ell$ vertices, then $\left( (f^{\le \ell} \diamond d') \diamond f^{\le \ell} \right)_{v, u} \le d_{G \setminus  \pi_G(s, u) \setminus \pi_G(v, t)}(v, u)$. 

Fix any shortest path from $v$ to $u$. 
Using Lemma~\ref{lem:f_structure}, we can write the path as $v = b_0 \rightarrow a_1 \rightarrow b_1 \cdots \rightarrow a_r \rightarrow b_r \rightarrow u = a_{r+1}$, where the subpath from $b_i$ to $a_{i+1}$  uses the shortest path in $G \setminus \pi_G(s, t)$  for each $i \in [0, r]$ and the subpath from $a_i$ to $b_i$ uses edges from $a_i$ to $b_i$ on $\pi_G(s, t)$ for each $i \in [1, r]$. Also, assume this path contains at most $\frac{3}{2} \ell$ vertices. Let $i$ be the smallest positive integer such that the number of vertices between $v$ and $a_i$  on the above path is greater than $\ell$. Depending on what $i$ is, there are several cases:
\begin{itemize}
    \item We separately handle the special case $r=0$, where the whole path is $b_0 \rightarrow a_1$. The distances in $d'$ clearly handles this case. Since the diamond product is always taking a min of the original inputs, $\left( (f^{\le \ell} \diamond d') \diamond f^{\le \ell} \right)_{v, u} \le d'(v, u)$, so the distance from $v$ to $u$ will be computed correctly in this case.
    
    \item Such $i$ does not exist or $i=r+1$. In this case, the number of vertices from $v$ to $a_r$ is at most $\ell$. $f^{\le \ell}$ has the correct length of the subpath from $v$ to $a_r$, and $d'$ has the correct length  of the subpath from $b_r$ to $u$. The subpath from $a_r$ to $b_r$ will be handled when we take the diamond product since $a_r$ corresponds to $p_k$ and $b_r$ corresponds to $p_{k'}$. 
    Thus, the whole path from $v$ to $u$ is accounted for by $d'\diamond f^{\le \ell}$. Since the diamond product is always taking a min of the original inputs, $\left( (f^{\le \ell} \diamond d') \diamond f^{\le \ell} \right)_{v, u} \le (d'\diamond f^{\le \ell})_{v, u}$, so the distance from $v$ to $u$ will be computed correctly in this case.
    \item $i=1$. In this case, 
    The part from $b_1$ to $u$ has at most $\frac{1}{2}\ell \le \ell$ vertices, so the path is accounted for by $f^{\le \ell}\diamond d'$ similar to the previous case. Since the diamond product is always taking a min of the original inputs, the distance from $v$ to $u$ will be computed correctly in this case.
    
    \item $1 < i \le r$. In this case the part from $v$ to $a_{i-1}$ has at most $\ell$ vertices and the part from $b_{i}$ to $u$ has at most $\ell$ vertices. The subpath from $v$ to $a_{i-1}$ then to $a_{i}$ is accounted for by $f^{\le \ell}\diamond d'$ . Also, since the part from $b_{i}$ to $u$ has at most $\ell$ vertices,  $f^{\le \ell}$ also has the correct distance of it. Therefore, the whole path is accounted for by $(f^{\le \ell}\diamond d')\diamond f^{\le \ell}$.
\end{itemize}
\end{proof}

\begin{lemma}
\label{lem:short_hop}
Given $f^{\le \ell}$, $d_G$ and $d_{G \setminus \pi_G(s, t)}^{\le \frac{3}{2} \ell M}$, we can compute $f^{\le \frac{3}{2} \ell}$  in $\tO(\ell M n^{\omega(1, 1 + \log_n \ell, 1)})$ time. 
\end{lemma}
\begin{proof}
By Lemma~\ref{lem:short_hop_correct}, it suffices to compute $(f^{\le \ell} \diamond_{\frac{3}{2} \ell} d_{G \setminus \pi_G(s, t)}^{\le \frac{3}{2} \ell M}) \diamond_{\frac{3}{2} \ell} f^{\le \ell}$. 

For computing $A \diamond_{\frac{3}{2} \ell} B$ between $O(n) \times O(n)$ matrices, we can create a new matrix $A'$ whose first dimension is indexed by pairs $(p_k, p_{k'})$ where $0 < k'-k < \frac{3}{2}\ell$ and the second dimension is indexed by $p_i$ where $i \in [O(n)]$. We set
\[A'_{(p_k, p_{k'}), p_i} = 
\begin{dcases*}
A_{p_{k'}, p_i} + d_G(p_k, p_{k'}) & if $i < k$,\\
\infty & otherwise.
\end{dcases*}
\]
Symmetrically, we create a matrix $B'$ such that 
\[B'_{p_j, (p_k, p_{k'})} = 
\begin{dcases*}
B_{p_j, p_k} & if $k' < j$,\\
\infty & otherwise.
\end{dcases*}
\]
Then clearly, the $(p_j, p_i)$-th entry of the min-plus product $(B' \star A')_{p_j, p_i}$ equals
\[\min_{\substack{i < k < k' < j \\ k' - k < \frac{3}{2}\ell}} \left\{A_{p_{k'},p_i} + d_G(p_k, p_{k'}) + B_{p_j, p_k}\right\}.\]
We can use the algorithm by Alon, Galil and Margalit~\cite{AlonGM97} that computes min-plus product between matrices with small integer or infinite entries to compute $B' \star A'$. 
All entries in $A'$ and $B'$ are either integers bounded by $O(\ell M)$ in absolute values or $\infty$, and  the dimension of $B'$ and $A'$ are $O(n)$ by $O(\ell n)$ and $O(\ell n)$ by $O(n)$, respectively, so we can compute $B' \star A'$ in $\tO(\ell M n^{\omega(1, 1 + \log_n \ell, 1)})$ time. 

All other components of the diamond product can be computed in linear time.
\end{proof}

Now we are ready to prove the running time for computing $f$ in graphs with small integer edge weights in $\{-M, \ldots, M\}$. Recall Lemma~\ref{lem:computing_f}:

\computingF*

\begin{proof}

First, we run SSSP from $s$ to compute $\pi_G(s, t)$.
Then, we use Zwick's algorithm \cite{Zwick02} to compute $d_G$ and $d_{G \setminus \pi_G(s, t)}$, which takes $\tO(M^{1/(4-\omega)} n^{2 + 1/(4-\omega)})$ time or $\tO(M^{0.7519} n^{2.5286})$ time if using rectangular matrix multiplication. These won't be  bottlenecks of our algorithm.

Then we use Lemma~\ref{lem:long-hop} and Lemma~\ref{lem:short_hop}  to compute $f^{\le \ell}$ for increasingly larger values of $\ell$ until $\ell$ reaches $n$. If we use the hop-short algorithm when $\ell \le L$ and use the hop-long algorithm when $\ell > L$ for some value $L$ to be fixed, then the running time is 
\begin{equation}
\label{eq: computing_f_runtime}
 \max\left\{\tO(L M n^{\omega(1, 1 + \log_n L, 1)})  , \tO(n^3/L)\right\}.
\end{equation}
We can use $\omega(1, 1 + \log_n L, 1) \le \omega + \log_n L$ and balance the two terms. The running time is minimized when $L = \frac{n^{1-\omega / 3}}{M^{1/3}}$, so the running time is $\tO(M^{1/3} n^{2+\omega / 3})$. We need to guarantee that  $L \ge 1$. However, when $L < 1$, our running time will be larger than $\tO(n^3)$, which can be obtained by using the algorithm for weighted graphs from Lemma~\ref{lem:computing_f_weighted}, so the $\tO(M^{1/3} n^{2+\omega / 3})$ upper bound will still hold regardless. 

If we use the current best rectangular matrix multiplication algorithm \cite{LU18},  we can set $L = n^{0.2222}/M^{0.3544}$. We assume $L \ge 1$, since otherwise our claimed running time will be larger than the $\tO(n^3)$ bound. 
If we let $M = n^m$, then the running time (\ref{eq: computing_f_runtime}) becomes
\[
\max\left\{\tO(n^{0.2222 + 0.6456m + \omega(1, 1.2222-0.3544m, 1)})  , \tO(n^{2.7778 + 0.3544m})\right\}.
\]
First, we use \cite{LU18} to bound $\omega(1,1.2222, 1) \le 2.5555$. Then, since $L \ge 1$ and $m \ge 0$, $1 \le 1.2222-0.3545m \le 1.2222$, so we can use convexity of $\omega(1, x, 1)$ to bound 
\[
\omega(1, 1.2222-0.3544m, 1) \le \frac{0.2222-0.3544m}{0.2222} \cdot 2.5555 + \frac{0.3544m}{0.2222} \cdot \omega.
\]
We can get our upper bound $ \tO(n^{2.7778 + 0.3544m})$ after some straightforward calculations by using $\omega < 2.3729$ \cite{Vassilevska12,LeGall14,AVW21}.
\end{proof}

\section{More Efficient SSRP with a Small Number of Targets}
\label{sec:sTRP}
In this section, we show a more efficient algorithm for SSRP with a small number of target vertices in a graph with integer edge weights in $\{-M, \ldots, M\}$, by modifying Grandoni and Vassilevska Williams' SSRP algorithm~\cite{GrandoniWilliamsSingleFailureDSO}.

\paragraph*{Overview of Grandoni and Vassilevska Williams' algorithm. } Their algorithm first computes a shortest path tree $\mathcal{T}_s$ rooted at the source vertex $s$. Then they find a set of at most $3H$ subtrees of $\mathcal{T}_s$, each with at most $n/H$ vertices, such that the vertices of these subtrees cover all the vertices of $G$, for some $H = \tO(1)$. Let $\mathcal{T}'$ be one of these subtrees, $t'$ be the root of $\mathcal{T}'$ and $P'$ be the path on $\mathcal{T}_s$ from $s$ to $t'$. For every $t \in V(\mathcal{T}')$, we must have $e \in E(P')$ or $e \in E(\mathcal{T}')$ in order for $d_G(s, t, e) \ne d_G(s, t)$. They deal with the first type of edges via the \textit{subpath} problem, and the second type of edges via the \textit{subtree} problem. 

We first consider the subpath problem. For $t \in V(\mathcal{T}')$ and $e \in E(P')$, the replacement path from $s$ to $t$ in $G \setminus \{e\}$ either uses the vertex $t'$ or not. If the replacement path uses $t'$, then $d_G(s, t, e) = d_G(s, t', e) + d_G(t', t)$, and $d_G(s, t', e)$ for all possibilities of $e \in E(P')$ can be precomputed using the RP algorithm of Grandoni and Vassilevska Williams \cite{GrandoniWilliamsSingleFailureDSO} in $\tO(M n^\omega)$ time. 

It remains to consider the case where the replacement path does not use $t'$. For such a replacement path, once it departs from the original shortest path at $v$ before $e$, it does not use any edge on $P'$ any more: it does not use any edge before $e$ on $P'$ since then it should not depart from $v$; it does not use any edge after $e$ on $P'$ either since then it should follow the edges on $P'$ to reach $t'$ and then to $t$. Thus, in this case, $d_{G}(s, t, e) = \min_{v \in V(P'), v \text{ is before } e} \{d_G(s, v) + d_{G \setminus E(P')}(v, t)\}$. Based on this, Grandoni and Vassilevska Williams showed that, for a fixed $t \in V(\mathcal{T}')$, if we have the distances $d_{G \setminus E(P')}(v, t)$ for every $v \in V(P')$, then we can compute $d_G(s, t, e)$ for all $e \in E(P')$ whose replacement path does not use $t'$ in $\tO(n)$ time. 
We can use Zwick's algorithm \cite{Zwick02} to compute $d_{G \setminus E(P')}(v, t)$ for every $v \in V(P')$ and every $t \in V(\mathcal{T}')$ in $O(M^{0.7519} n^{2.5286})$ time. This step is the bottleneck of their SSRP algorithm. In our modified algorithm, we use the fact that the set of $t$ we are interested in is small to avoid a whole APSP computation. 

For the subtree problem, they apply a compression step to create a graph $G'$ such that
\begin{enumerate}
    \item $G'$ contains $s$ and $\mathcal{T}'$;
    \item $G'$ contains $O(n \log n / H)$ vertices and all edge weights are in $\{-MH, \ldots, MH\}$; 
    \item and with high probability, $d_G(s, t, e)$ for $(t, e) \in V(\mathcal{T}') \times E(\mathcal{T}')$ equals $d_{G'}(s, t, e)$.
\end{enumerate}
Thus, it suffices to solve the SSRP problem recursively in each $G'$ created. 

Note that in the $i$th level of the recursion, there are at most $(3H)^{i+1}$ instances $G'$, each instance contains at most $n(C \log n / H)^{i+1}$ vertices for some constant $C$ and has integer edge weights bounded by $MH^i$ in absolute value. Also, the total number of recursion levels is $\log_{H / C \log n} n$. They finally bound the running time via the following claim. 

\begin{claim}[Implicit in \cite{GrandoniWilliamsSingleFailureDSO}]
\label{cl:recursion_bound}
There exists $H = \tO(1)$ such that for any constant $C$,  
$$\sum_{i=0}^{\log_{H / C \log n} n} (3H)^{i+1} (MH^i)^\alpha \left( n(C \log n / H)^{i+1} \right)^\beta = \tO(M^\alpha n^\beta),$$
as long as $\beta \ge \alpha + 1 \ge 1$. 
\end{claim}

This claim implies that if the running time of the algorithm is $\tO(M^\alpha n^\beta)$ for one subtree $\mathcal{T}'$, then the running time of the whole algorithm is still $\tO(M^\alpha n^\beta)$ as long as $\beta \ge \alpha + 1 \ge 1$.

The other ingredient of our modified algorithm is an efficient algorithm for the STSP problem in which we need to compute the shortest path distances between every pair of $s \in S, t \in T$ in a graph given  two subsets of vertices $S$ and $T$. 

\begin{theorem}[Theorem 1.4 in \cite{GrandoniWilliamsSingleFailureDSO}]
There is a randomized algorithm that solves STSP in a graph with $n$ vertices and edge weights in $\{-M, \ldots, M\}$ in $\tO(Mn^\omega + |S| \cdot |T| \cdot (Mn)^{\frac{1}{4-\omega}})$ time with high probability. 
\end{theorem}

Now we can show our faster algorithm for SSRP with a small number of target vertices. Recall Lemma~\ref{lem:sTRP}:

\sTRP*

\begin{proof}
We modify Grandoni and Vassilevska Williams' SSRP algorithm. As outlined previously they find a set of subtrees of a shortest path tree $\mathcal{T}_s$ rooted at $s$ such that the vertices of these subtrees cover all the vertices of $G$. Even though each target $t \in T$ can exist in multiple subtrees, it suffices to associate it with one of these subtrees. This ensures that the total number of target vertices across subtrees $\mathcal{T}'$ is still $|T|$, in any level of the recursion.

For some subtree $\mathcal{T}'$, say the set of target vertices associated with it is $T_j$. 
As mentioned previously, we can improve the second case of the subpath problem. Namely, let $P'$ be the path on $\mathcal{T}_s$ from $s$ to the root of $\mathcal{T}'$, we previously have to compute $d_{G \setminus E(P')}(v, t)$ for every $v \in V(P')$ and $t \in \mathcal{T}_s$, but now it suffices to compute $d_{G \setminus E(P')}(v, t)$ for every $v \in V(P')$ and $t \in T_j $. Using the algorithm for STSP, this step takes $\tO(Mn^\omega + M^{\frac{1}{4-\omega}} n^{1+\frac{1}{4-\omega}} \cdot |T_j|)$ time. The second term of this running time sums up to $\tO(M^{\frac{1}{4-\omega}} n^{1+\frac{1}{4-\omega}} \cdot |T|)$ over all subtrees. 

Similarly, the sum of the sizes of the target sets is kept at $|T|$ in any recursion level $i$. Thus, the running time for the $i$th level of the recursion is 
$$ \tO \left( (3H)^{i+1} (MH^i) \left( n(C \log n / H)^{i+1} \right)^\omega  +(MH^i)^{\frac{1}{4-\omega}} \left( n(C \log n / H)^{i+1} \right)^{1+\frac{1}{4-\omega}} \cdot |T|\right).$$
Summing over all levels $i$, the running time of the whole algorithm is $\tO(Mn^\omega + M^{\frac{1}{4-\omega}} n^{1+\frac{1}{4-\omega}} \cdot |T|)$ 
by Claim~\ref{cl:recursion_bound}.
\end{proof}

\section{Combinatorial Lower Bounds for \texorpdfstring{$k$}{k}-Fault DSO with Fixed Source and Target}

In this section we will prove Theorem \ref{thm:intro_lower_bound} by giving a reduction from Triangle Detection to $k$-fault DSO with fixed source and sink in unweighted directed graphs, for any small constant $k > 0$. The Triangle Detection problem is as follows: given an undirected graph $G = (V, E)$, output \textit{true} if there exists a group of three vertices $i$, $j$, and $k$ such that $\{(i, j), (j, k), (k, i)\} \subseteq E$, otherwise output \textit{false}.

In the combinatorial setting, Triangle Detection is known to be subcubically equivalent to BMM, meaning that if there exists a combinatorial truly subcubic time (in terms of $n$) algorithm for Triangle Detection, then there is a combinatorial truly subcubic time algorithm for BMM \cite{williams2018subcubic}. Thus our reduction gives a conditional (on the validity of the BMM hypothesis) lower bound for combinatorial data structures that can pre-process a given graph with fixed source $s$ and target $t$ and can answer $k$-fault distance sensitivity queries from $s$ to $t$.

\subsection{Reduction from Triangle Detection}

The reduction given here is a generalization of the reduction given in \cite{williams2018subcubic} from Triangle Detection to unweighted single-fault RP. Suppose we are given an instance of Triangle Detection $G = (V, E)$, where $V$ is identified by $v_0, \dots, v_{n-1}$. Let $L$ be a parameter to be set later. We partition $V$ into $\lceil n/L \rceil$ buckets of size at most $L$, $V_0$ through $V_{\lceil n/L \rceil - 1}$, where $V_i$ is the $i$th bucket containing the vertices $v_{iL}$ through $v_{\min\{n,(i+1)L\}-1}$. We will reduce Triangle Detection to $\lceil n/L \rceil$ instances of $k$-fault DSO with fixed source and target, where we will perform $L$ $k$-fault distance sensitivity queries on each instance. In a particular instance $b$, we will be able to check whether there is a triangle going through a vertex in $v \in V_b$.

Fix some instance $b$. First, create four parts of vertices $Q = \{q_0, \dots, q_{L-1}\}$, $A = \{a_0, \dots, a_{n-1}\}$, $B = \{b_0, \dots, b_{n-1}\}$ and $C = \{c_0, \dots, c_{L-1}\}$ such that there is an edge $(q_i, a_j)$ if $(v_{bL + i}, v_j) \in E$, an edge $(a_i, b_j)$ if $(v_i, v_j) \in E$, and an edge $(b_i, c_j)$ if $(v_i, v_{bL + j}) \in E$. This construction ensures that if there is a triangle in $G$ through vertex $v_{bL+i} \in V_b$, then there is a length-3 path from $q_i$ to $c_i$ in this graph.

Now, we will create $k$ layers of vertices, $P_0$ through $P_{k-1}$. Let $\mathcal{B} = \lceil L^{1/k} \rceil$. Layer $P_i$ will consist of $\mathcal{B}^i$ chunks with $\mathcal{B} + 1$ vertices in each chunk. Let the vertices of a particular chunk be identified by $v_0, \dots, v_{\mathcal{B}}$. For every $0 \leq i < \mathcal{B}$, we will add an edge $(v_i, v_{i+1})$ so that there is a path through the vertices of the chunk. Different chunks in the same layer $P_i$ will not have any edges connecting them. Let $p_{i, j, s}$ be the $s$th vertex in the $j$th chunk of $P_i$, where $s \in \{0, \dots, \mathcal{B} \}$ and $j \in \{0, \dots, \mathcal{B}^i - 1\}$. We will construct our graph so that each chunk $j$ in layer $P_i$ has inter-layer paths to and from chunks between the $(j\mathcal{B})$-th chunk and the $((j+1)\mathcal{B} - 1)$-th chunk in layer $P_{i+1}$:

\begin{itemize}
    \item For every $i \in \{0, \dots, k - 2\}, j \in \{0, \dots, \mathcal{B}^i - 1\}, s \in \{0, \dots, \mathcal{B} - 1\}$, add a path of length $2(\mathcal{B}^{k-i} - s\mathcal{B}^{k - i - 1})$ from $p_{i, j, s}$ to $p_{i+1, j\mathcal{B} + s, 0}$
    \item For every $i \in \{0, \dots, k - 2\}, j \in \{0, \dots, \mathcal{B}^i - 1\}, s \in \{1, \dots, \mathcal{B}\}$, add a path of length $2s\mathcal{B}^{k - i - 1}$ from $p_{i+1, j\mathcal{B} + s - 1, \mathcal{B}}$ to $p_{i, j, s}$. 
\end{itemize}

Finally, we will connect these layers to the rest of the graph by adding edges as follows:

\begin{itemize}
    \item For every $j \in \{0, \dots, \mathcal{B}^{k-1} - 1\}, s \in \{0, \dots, \mathcal{B} - 1\}$, add a path of length $2(\mathcal{B} - s)$ from $p_{k-1, j, s}$ to $q_{j\mathcal{B} + s}$
    \item For every $j \in \{0, \dots, \mathcal{B}^{k-1} - 1\}, s \in \{1, \dots, \mathcal{B}\}$, add a path of length $2s$ from $c_{j\mathcal{B} + s - 1}$ to $p_{k-1, j, s}$. 
\end{itemize}

\begin{figure}[ht]
\centering
\begin{tikzpicture}[scale=0.9, transform shape]
    \node at (-1.5,0) [label=below:$s$] (s){};
    \node at (1.5, 0) [label=below:$t$] (t){}; 

    \node at (-8,0) [label=left:$P_0$] (){};
    \node at (-8,2) [label=left:$P_1$] (){};
    
    \foreach \i in {0,...,0}
    \foreach \j in {0,...,3}
    {
        \pgfmathtruncatemacro{\y}{0};
        \pgfmathtruncatemacro{\x}{5 * \i + \j};
        \pgfmathtruncatemacro{\label}{0 \i \j};
        \node at (\x - 1.5,\y) [circle,fill, inner sep = 1pt] (\label){};
    }

    \foreach \i in {0,...,0}
    \foreach \j in {0,...,2}
    {
        \pgfmathtruncatemacro{\labelone}{0  \i   \j};
        \pgfmathtruncatemacro{\labeltwo}{0  \i  (\j+1)};
        \draw [-{Stealth[length=2mm,width=1.3mm]},line width=0.5] (\labelone) to[] (\labeltwo);
    }

    \foreach \i in {0,...,2}
    \foreach \j in {0,...,3}
    {
        \pgfmathtruncatemacro{\y}{2};
        \pgfmathtruncatemacro{\x}{5 * \i - 5 + \j};
        \pgfmathtruncatemacro{\label}{1 \i \j};
        \node at (\x - 1.5,\y) [circle,fill, inner sep = 1pt] (\label){};
    }

    \foreach \i in {0,...,2}
    \foreach \j in {0,...,2}
    {
        \pgfmathtruncatemacro{\labelone}{1   \i    \j};
        \pgfmathtruncatemacro{\u}{\j + 1};
        \pgfmathtruncatemacro{\labeltwo}{1   \i  \u};
        \draw [-{Stealth[length=2mm,width=1.3mm]},line width=0.5] (\labelone) to[] (\labeltwo);
    }

    \foreach \i in {0,...,0}
    \foreach \j in {0,...,2}
    {
        \pgfmathtruncatemacro{\labelone}{0  \i  \j};
        \pgfmathtruncatemacro{\labeltwo}{1 \j 0};
        \pgfmathtruncatemacro{\weight}{2 * (9 - 3 * \j)};
        \pgfmathtruncatemacro{\xshift}{(1-Mod(\j,2))*5};
        \draw [-{Stealth[length=2mm,width=1.3mm]},line width=0.5, blue] (\labelone) to[] node[right, xshift=\xshift]{$\weight$} (\labeltwo);
    }

    \foreach \i in {0,...,0}
    \foreach \j in {1,...,3}
    {
        \pgfmathtruncatemacro{\labelone}{0  \i  \j};	 		    	
        \pgfmathtruncatemacro{\u}{\j - 1};
        \pgfmathtruncatemacro{\labeltwo}{1  \u  3};
        \pgfmathtruncatemacro{\weight}{2 * \j * 3};
        \pgfmathtruncatemacro{\xshift}{(\j==1)*2+(\j==3)*5};
        \draw [-{Stealth[length=2mm,width=1.3mm]},line width=0.5, red] (\labeltwo) to[] node[right, xshift=\xshift]{$\weight$} (\labelone);
    }

    \node[draw,ellipse,minimum height=1.5 cm,minimum width=8 cm , anchor = center] (Q) at (-4.5,5) {};
    \node[left] at (Q.180) {$Q$};
    \node[draw,ellipse,minimum height=1.5 cm,minimum width=8 cm, anchor = center] (C) at (4.5,5) {};
    \node[right] at (C.0) {$C$};

    \foreach \i in {0,...,8}
    {
        \pgfmathtruncatemacro{\y}{5};
        \pgfmathtruncatemacro{\x}{\i};
        \pgfmathtruncatemacro{\label}{1 \i};
        \node at (0.75 * \x - 0.75 * 4 - 4.5,\y) [circle,fill, inner sep = 1pt] (\label){};
    }

    \node[above] at (10) {$q_0$};
    \node[above left] at (13) {$q_3$};
    \node[above] at (18) {$q_8$};

    \foreach \i in {0,...,2}
    \foreach \j in {0,...,2}
    {
        \pgfmathtruncatemacro{\u}{\i * 3 + \j};
        \pgfmathtruncatemacro{\labelone}{1 \u};
        \pgfmathtruncatemacro{\labeltwo}{1  \i   \j};
        \pgfmathtruncatemacro{\weight}{2 * (3 - \j)};
        \pgfmathtruncatemacro{\xshift}{(\i==2)*3};
        \draw [-{Stealth[length=2mm,width=1.3mm]},line width=0.5, blue] (\labeltwo) to[] node[right, pos=0.9, xshift=\xshift]{$\weight$} (\labelone);
    }

    \foreach \i in {0,...,8}
    {
        \pgfmathtruncatemacro{\y}{5};
        \pgfmathtruncatemacro{\x}{\i};
        \pgfmathtruncatemacro{\label}{2 \i};
        \node at (0.75 * \x - 0.75 * 4 + 4.5,\y) [circle,fill, inner sep = 1pt] (\label){};
    }

    \node[above] at (20) {$c_0$};
    \node[above right] at (23) {$c_3$};
    \node[above] at (28) {$c_8$};
    
    \foreach \i in {0,...,2}
    \foreach \j in {1,...,3}
    {
        \pgfmathtruncatemacro{\u}{\i * 3 + \j - 1};
        \pgfmathtruncatemacro{\labelone}{2 \u};
        \pgfmathtruncatemacro{\labeltwo}{1  \i   \j};
        \pgfmathtruncatemacro{\weight}{2 * \j};
        \pgfmathtruncatemacro{\xshift}{(\i==0)*3};
        \draw [-{Stealth[length=2mm,width=1.3mm]},line width=0.5, red] (\labelone) to[] node[right, pos=0.1, xshift=\xshift]{$\weight$} (\labeltwo);
    }

    \node[draw,ellipse,minimum height=3 cm,minimum width=4 cm , anchor = center] (A) at (-4,9) {};
    \node[left] at (A.150) {$A$};
    \node at (-4.3, 9.2) [circle,fill, inner sep = 1pt, label=above right:$a_i$] (ai){};

    \node[draw,ellipse,minimum height=3 cm,minimum width=4 cm , anchor = center] (B) at (4,9) {};
    \node[right] at (B.30) {$B$};
    \node at (4.4, 8.3) [circle,fill, inner sep = 1pt, label=right:$b_j$] (bj){};

    \draw [-{Stealth[length=2mm,width=1.3mm]}, line width=0.5] (13) to[] node[right]{iff $(v_{bL+3}, v_i) \in G$} (ai);
    \draw [-{Stealth[length=2mm,width=1.3mm]}, line width=0.5] (ai) to[] node[above]{iff $(v_i, v_j) \in G$} (bj);
    \draw [-{Stealth[length=2mm,width=1.3mm]},line width=0.5] (bj) to[] node[left]{iff $(v_j, v_{bL+3}) \in G$} (23);

\end{tikzpicture}
\caption{An example of our constructed graph for $\mathcal{B} = 3$ and $k = 2$. Here, each black arrow is an actual edge in the graph, while each red or blue arrow is a path of certain length. Note that most edges between $Q, A, B, C$ are omitted for a cleaner presentation. }
\end{figure}
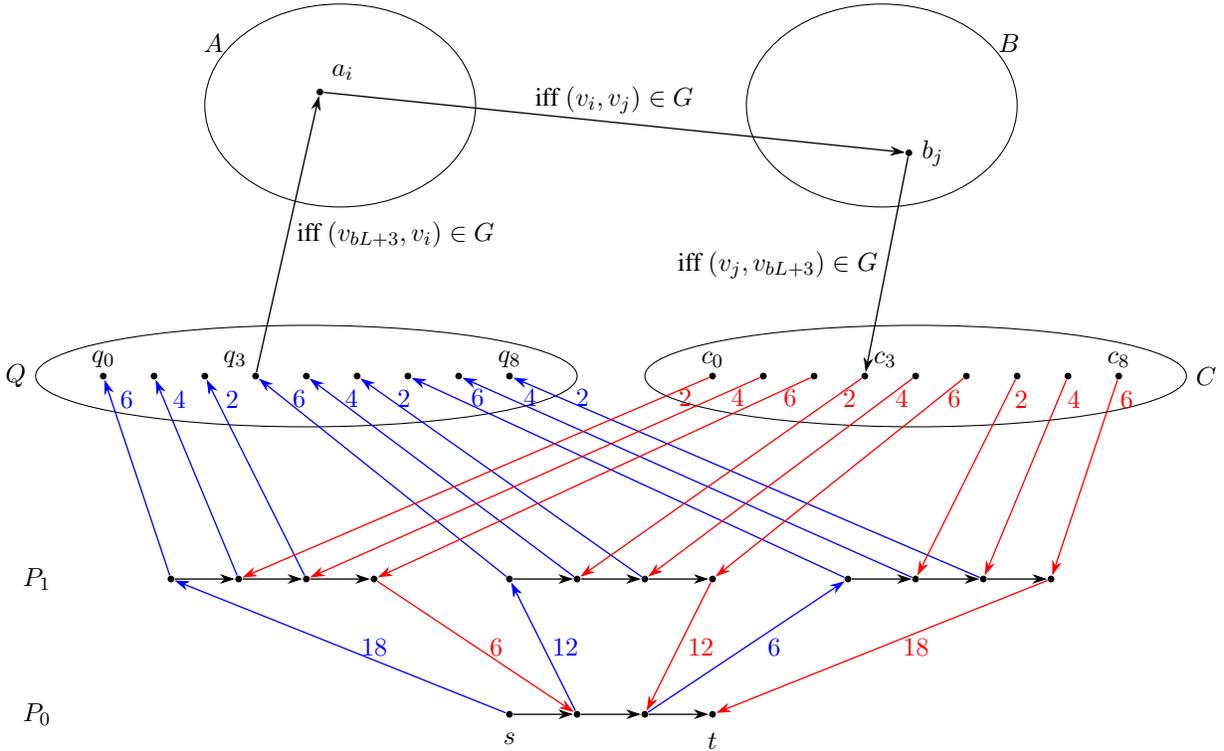

The source and target vertices of our  instance will be $p_{0, 0, 0}$ and $p_{0, 0, \mathcal{B}}$ respectively, i.e. the first and last vertex of the single chunk in $P_0$. For any integer $0 \leq u < L$, let $\phi(u) = \langle a_0, \dots, a_{k-1} \rangle$ be the sequence of non-negative integers in $[0, \mathcal{B} - 1]$ such that $u = \allowbreak \sum_{j = 0}^{k-1} a_j \mathcal{B}^{k - j - 1}$, i.e. the unique representation of integer $u$ as a $k$ digit integer in base-$\mathcal{B}$.

Let $G_b$ be the graph created above and let $H_b$ be a copy of $G_b$ after removing all of the edges between $Q$, $A$, $B$, and $C$. We will show that the arrangement of layers $P_0$ through $P_{k-1}$ gives $H_b$ two key properties.

\begin{lemma}
\label{lem:bigger_q_shorter}
In $H_b$, the length of the shortest path from $s$ to a vertex $q_w \in Q$ is $\sum_{i=0}^{k-1} (\phi(w)_i + 2(\mathcal{B}^{k-i} - \phi(w)_i\mathcal{B}^{k-i-1}))$. Thus, for any $q_u, q_v \in Q$ where $u < v$, the shortest path from $s$ to $q_u$ in $H_b$ is longer than the shortest path from $s$ to $q_v$ in $H_b$.
\end{lemma}
\begin{proof}

First, we will show that once the shortest path has transitioned from layer $P_i$ to layer $P_{i+1}$, it will not return to layer $P_i$. Assume that the shortest path transitions to $p_{i, j, w}$ from $p_{i+1, j\mathcal{B} + w - 1, \mathcal{B}}$. The only way for the path to have reached that chunk on layer $P_{i+1}$ would have been by reaching $p_{i, j, w - 1}$ and using the inter-layer path to the start of that chunk. However, this makes it no longer a shortest path, as it would be shorter to simply use the edge between $p_{i, j, w-1}$ and $p_{i, j, w}$. Therefore, the shortest path will not return to layer $P_i$.

For any vertex in $q_w \in Q$, there is only one edge to $q_w$ in $H_b$, which is from a vertex in $P_{k-1}$. For any vertex $p$ in a layer $P_i$, there is either an inter-layer path directly to $p$ from the previous layer $P_{i-1}$ if it is the beginning of its chunk, or there is no path directly to $p$ from the previous layer and the only path to it from the previous layer is through the start of its chunk. Since each chunk is only accessible through one chunk on the previous layer, the shortest path from $s$ to $q_w$ must travel through the appropriate chunk in each layer to reach $q_w$.

As a result, any shortest path from $s$ to $q_w$ will travel inside a chunk in each layer $P_i$ and then taking an inter-layer path to get to the next layer $P_{i+1}$, until finally the path reaches the vertex on layer $P_{k-1}$ that connected to $q_w$ and takes the last inter-layer path to reach $q_w$.

Let $\phi(w) = \langle a_0, \dots, a_{k-1} \rangle$. By the construction of the graph, the last vertex in layer $P_i$ on the path from $s$ to $q_u$ will be the $a_i$-th vertex in its chunk. Therefore, the length of the path from the first vertex in layer $P_i$ to the first vertex in layer $P_{i+1}$ will be $a_i$, in order to travel along the chunk to the last vertex in the layer, plus $2(\mathcal{B}^{k-i} - a_i \mathcal{B}^{k-i-1})$, in order to transition between the layers. From the first vertex in $P_{k-1}$ to $q_u$, the length of the path will be $a_{k-1}$, to travel along the chunk to the transition point, plus $2(\mathcal{B} - a_{k-1})$, to transition to $Q$. In total, the length of the path from $s$ to $q_w$ is $\sum_{i=0}^{k-1} (a_i + 2(\mathcal{B}^{k-i} - a_i\mathcal{B}^{k-i-1}))$.

By reorganizing $\sum_{i=0}^{k-1} (\mathcal{B}^{k-i} - a_i \mathcal{B}^{k-i-1})$, we get $\sum_{i=0}^{k-1} \mathcal{B}^{k-i} - \sum_{i=0}^{k-1} a_i\mathcal{B}^{k-i-1}$, which we can simplify to $(\sum_{i=0}^{k-1} \mathcal{B}^{k-i}) - w$. Therefore, we can rewrite the length of the path as $(\sum_{i=0}^{k-1} \phi(w)_i) + (2\sum_{i=0}^{k-1} \mathcal{B}^{k-i}) - 2w$. From here, we can easily see that if we instead plug in $u$ and $v$ that, since $v > u$, the path from $s$ to $q_v$ will be shorter than the path from $s$ to $q_u$, since any changes to the sum of the digits will be dominated by the change in the value of the number itself and it is easy to see that $(\sum_{i=0}^{k-1} \phi(w)_i) - 2w > (\sum_{i=0}^{k-1} \phi(w+1)_i)-2(w+1)$ for any $w \in [0, L-2]$. 
\end{proof}

\begin{lemma}
\label{lem:smaller_c_shorter}
In $H_b$, the length of the shortest path from a vertex $c_w \in C$ to $t$ is $\sum_{i=0}^{k-1} (\mathcal{B} - 1  - \phi(w)_i + 2\mathcal{B}^{k-i-1}(\phi(w)_i + 1))$. Thus, for any $c_u, c_v \in C$ where $u < v$, the shortest path from $c_u$ to $t$ in $H_b$ is shorter than the shortest path from $c_v$ to $t$ in $H_b$.
\end{lemma}
\begin{proof}

Since each chunk is only accessible through one chunk on the previous layer, the shortest path from $c_w$ to $t$ must travel through the appropriate chunk in each layer to reach $t$. Using similar logic to the previous proof we know that once the shortest path has reached a chunk on a layer $P_i$, it will not transition back up to layer $P_{i+1}$. As a result, the shortest path from $c_w$ to $t$ will consist of first travelling down an inter-layer path to $P_{k-1}$, then in each layer $P_i$ traveling down a chunk to its end and then taking an inter-layer path to get to the next layer $P_{i-1}$, until finally the path reaches $P_0$ and can follow the chunk to $t$.

For any vertex in $c_w \in C$, there is only one edge in $H_b$ from $c_w$, which ends in a vertex in $P_{k-1}$. For any vertex $p$ in a layer $P_i$, there is either an inter-layer path directly from $p$ to the lower layer $P_{i-1}$ if it is the end of its chunk, or there is no path directly from $p$ to the lower layer and the only path to the lower layer is through the end of its chunk.

Let $\phi(w) = \langle a_0, \dots, a_{k-1} \rangle$. First, to go from $c_w$ to the last vertex in $P_{k-1}$ will take $2(a_{k-1} + 1)$, to transition to $P_{k-1}$, plus $\mathcal{B} - a_{k-1} - 1$, to travel to the end of the chunk. By the construction of this graph, the first vertex in layer $P_i$ on the path from $c_w$ to $t$ will be the $(a_i + 1)$-th vertex in its chunk, and the last will the end of the chunk. Therefore, the length of the path from the last vertex in layer $P_{i+1}$ to the last vertex in layer $P_i$ will be $2(a_i + 1)\mathcal{B}^{k-i-1}$, in order to transition between the layers, plus $\mathcal{B} - a_i - 1$, to travel to the end of the chunk. In total, the length of the path from $c_w$ to $t$ is $\sum_{i=0}^{k-1} (\mathcal{B} - 1  - a_i + 2\mathcal{B}^{k-i-1}(a_i + 1))$.

Breaking down the summation, we can rewrite it as $k\mathcal{B} - k - (\sum_{i=0}^{k-1}( a_i + 2\mathcal{B}^{k-i-1})) + 2(\sum_{i=0}^{k-1} a_i\mathcal{B}^{k-i-1})$, which we can simplify to $k\mathcal{B} - k - (\sum_{i=0}^{k-1}( a_i + 2\mathcal{B}^{k-i-1})) + 2w$. From here, we can easily see that if we instead plug in $u$ and $v$ that, since $v > u$, the path from $c_v$ to $t$ will be longer than the path from $c_u$ to $t$, since any changes to the sum of the digits will be dominated by the change in the value of the number itself.
\end{proof}

Given a fixed $u \in [0, L-1]$, we can construct a particular $k$-fault distance sensitivity query $S_u$ using $\phi(u)$. In each layer $P_i$, remove the edge $(p_{i, j, \phi(u)_i}, p_{i, j, \phi(u)_i + 1})$, where $j = \sum_{r = 0}^{i-1} \phi(u)_r \mathcal{B}^{k-1-r}$ is the chunk in layer $P_i$ that the shortest $s$ to $q_u$ path goes through. Note that all of the removed edges are present in $H_b$. We will now show the following:

\begin{lemma}
\label{lem:reachable_from_s}
In $H_b \setminus S_u$, the only vertices which can be reached from $s$ are $\{ q_v \in Q \mid v \leq u\} \cup \{ p_{i, j, w} \mid i \in [0, k-1], j < \sum_{r=0}^{i-1}\phi(u)_r\mathcal{B}^{k-1-r} \} \cup \{ p_{i, j, w} \mid i \in [0, k-1], j = \sum_{r=0}^{i-1}\phi(u)_r\mathcal{B}^{k-1-r}, w \leq \phi(u)_i \}$. In addition, for any $q_v \in Q$ such that $v \leq u$, $d_{H_b \setminus S_u}(s, q_v) = d_{H_b}(s, q_v)$.
\end{lemma}
\begin{proof}
It is clear (yet might be tedious) to see that $s$ can reach the vertices listed above. We focus on showing why $s$ cannot reach any other vertex. 

First, we will consider which vertices in layers $P_0$ through $P_{k-1}$ can be reached from $s$. Fix some $p_{i, j, w}$ such that $j = \sum_{r=0}^{i-1}\phi(u)_r\mathcal{B}^{k-1-r}$ and $w > \phi(u)_i$. The path from $s$ can not travel directly down the chunk to $p_{i, j, w}$ as it would have to use the failed edge $(p_{i, j, \phi(u)_i}, p_{i, j, \phi(u)_i +1})$. The only other possible path would be to reach $p_{i+1, j, \phi(u)_i}$ and travel to $p_{i+1, j\mathcal{B} + \phi(u)_i, 0}$, then to $p_{i, j\mathcal{B} + \phi(u)_i, \mathcal{B}}$, then back to $p_{i, j, \phi(u)_i + 1}$ to avoid the failed edge. However, this path would require the use of the failed edge $(p_{i+1, j\mathcal{B} + \phi(u)_i, \phi(u)_{i+1}},\allowbreak p_{i+1, j\mathcal{B} + \phi(u)_i, \phi(u)_{i+1} + 1})$ in that chunk. Avoiding that failed edge would require a similar path to layer $P_{i+2}$, and so on and so on, until finally in layer $P_{k-1}$ we find that there is no alternate path to take. Therefore, it is impossible for $s$ to reach $p_{i, j, w}$.

Fix some $p_{i, j, w}$ such that $j > \sum_{r=0}^{i-1}\phi(u)_r\mathcal{B}^{k-1-r}$. We can write $j$ as $\sum_{r=0}^{i-1}a_r\mathcal{B}^{k-1-r}$ for $0 \leq a_r < \mathcal{B}$. Let $r'$ be the smallest index such that $a_{r'} > \phi(u)_{r'}$, and let $j' = \sum_{r=0}^{r'-1}a_r\mathcal{B}^{k-1-r}$. For $s$ to reach $p_{i, j, w}$, it must have gone through $p_{r', j', a_{r'}}$. However, $p_{r', j', a_{r'}}$ falls into the previous case and therefore can not be reached by $s$, therefore $p_{i, j, w}$ can not be reached by $s$.

Trivially, $s$ can not reach any vertices in $A \cup B \cup C$ since none of those vertices have incoming edges in $H \setminus S_u$. Therefore, it suffices to show that $s$ can not reach any vertex $q_v \in Q$ such that $v > u$. 

Fix some value of $v > u$, and let $i$ be the minimum index such that $\phi(v)_i > \phi(u)_i$. For $s$ to have a path to $q_v$, it must first have a path to $p_{i, j, \phi(v)_i}$, where $j$ is the chunk in layer $P_i$ that the $s$ to $q_v$ path goes through. However, this path must use the edge $(p_{i, j, \phi(u)_i}, p_{i, j, \phi(u)_i +1})$, which has failed, so there is no path from $s$ to $p_{i, j, \phi(v)_i}$, and therefore no path from $s$ to $q_v$. Therefore, $s$ can not reach any vertex $q_v$ where $v > u$.

Fix some value of $v \leq u$. Assume for contradiction that the shortest path from $s$ to $q_v$ in $H_b$ used a failed edge in $S_u$, and let $i$ be the smallest index for which the path uses an edge $(p_{i, j, \phi(u)_i}, p_{i+1,j,\phi(u)_i+1}) \in S_u$. By the construction of the layers, this means that the first $i$ elements of $\phi(u)$ and $\phi(v)$ are identical, as otherwise the two paths from $s$ would not have reached the same chunk in layer $P_i$. However, the fact that the path from $s$ to $q_v$ uses that failed edge means that it departs from the chunk on or after $p_{i, j, \phi(u)_i + 1}$, which means that $\phi(v)_i > \phi(u)_i$. This is a contradiction, as it implies that $v > u$. Therefore, the shortest path from $s$ to $q_v$ in $H_b$ does not use any edges in $S_u$, so $d_{H_b \setminus S_u}(s, q_v) = d_{H_b}(s, q_v)$.
\end{proof}

\begin{lemma}
\label{lem:reachable_to_t}
In $H_b \setminus S_u$, the only vertices in $C$ which can reach $t$ are $\{ c_v \in C \mid v \geq u\}$. In addition, for any $c_v \in C$ such that $v \geq u$, $d_{H_b \setminus S_u}(c_v, t) = d_{H_b}(q_v, t)$.
\end{lemma}
\begin{proof}
Fix some value of $v < u$, and let $i$ be the minimum index such that $\phi(v)_i < \phi(u)_i$. For $c_v$ to have a path to $t$, it must first have a path to $p_{i, j, \mathcal{B}}$, where $j$ is the chunk in layer $P_i$ that the $c_v$ to $t$ path goes through. However, similar to the previous proof this path must use the edge $(p_{i, j, \phi(u)_i}, p_{i, j, \phi(u)_i +1})$, which has failed, so there is no path from $c_v$ to $p_{i, j, \mathcal{B}}$, and therefore no path from $c_v$ to $t$. Therefore, $t$ can not be reached by any vertex $c_v$ where $v < u$.

Fix some value of $v \geq u$. Assume for contradiction that the shortest path from $c_v$ to $t$ in $H_b$ used a failed edge in $S_u$, and let $i$ be the smallest index for which the path uses an edge $(p_{i, j, \phi(u)_i}, p_{i+1,j,\phi(u)_i+1}) \in S_u$. By the construction of the layers, this means that the first $i$ elements of $\phi(u)$ and $\phi(v)$ are identical, as the two paths will take the same route to $t$ from that chunk. However, the fact that the path from $c_v$ to $t$ uses that failed edge means that it arrives at the chunk before or on $p_{i, j, \phi(u)_i}$, which means that $\phi(v)_i < \phi(u)_i$. This is a contradiction, as it implies that $v < u$. Therefore, the shortest path from $s$ to $q_v$ in $H_b$ does not use any edges in $S_u$, so $d_{H_b \setminus S_u}(s, q_v) = d_{H_b}(s, q_v)$.
\end{proof}

Via Lemma~\ref{lem:reachable_from_s} we can see that $s$ can not reach $t$ in $H_b \setminus S_u$, which means that the $s$-$t$ replacement path in $G_b \setminus S_u$ must use the edges between $Q$, $A$, $B$, and $C$. In addition, we show that $G_b \setminus S_u$ has the following property:

\begin{lemma}
\label{lem:closest_q_c}
In $G_b \setminus S_u$, the unique closest vertex in $Q$ from $s$ is $q_u$, and the unique closest vertex in $C$ to $t$ is $c_u$.
\end{lemma}
\begin{proof}
Via Lemma~\ref{lem:reachable_from_s} we can see that in $H_b \setminus S_u$, $q_u$ is the closest vertex in $Q$ to $s$. Assume for contradiction that there is a vertex $q_v \in Q, v \neq u$ such that $d_{G_b \setminus S_u}(s, q_v) \le d_{G_b \setminus S_u}(s, q_u)$. This path must use the edges in $E(G_b) \setminus E(H_b)$ as this path can not exist in $H_b$ by Lemma~\ref{lem:reachable_from_s}. To access these edges, the path must contain a vertex in $Q$. Let $q_w \in Q$ be the first vertex on the path in $Q$. All of the edges on the path before $q_w$ are in $H_b \setminus S_u$, so via Lemma~\ref{lem:reachable_from_s} we get that $w \leq u$. The length of this sub-path is already at least $d_{G_b \setminus S_u}(s, q_u)$ by Lemma~\ref{lem:bigger_q_shorter}, and the path must use at least one other edge in $E(G_b) \setminus E(H_b)$. Thus it is impossible that $d_{G_b \setminus S_u}(s, q_v) \le d_{G_b \setminus S_u}(s, q_u)$. 

Using Lemma~\ref{lem:reachable_to_t}, we can similarly show that $c_u$ is the unique closest vertex in $C$ to $t$.
\end{proof}

\begin{theorem}
\label{thm:short_path_is_triangle}
The shortest $s$-$t$ path avoiding $S_u$ has length $3 + \sum_{i=0}^{k-1} (2\mathcal{B}^{k-i} + 2\mathcal{B}^{k-i-1} + \mathcal{B} - 1)$ if and only if $v_{bL+u}$ is in a triangle in $G$.
\end{theorem}
\begin{proof}
Assume that there is a triangle in $G$ which includes $v_{bL+u}$. By construction, that means there is a path of length $3$ from $q_u$ to $c_u$ in $G_b \setminus S_u$. As mentioned previously, the replacement path must go through $Q$, $A$, $B$, and $C$. Via Lemma \ref{lem:closest_q_c}, we know that $q_u$ is the closest vertex in $Q$ from $s$, and $c_u$ is the closest vertex in $C$ to $t$, so the shortest replacement path must consist of a path from $s$ to $q_u$, a path from $q_u$ to $c_u$, and a path $c_u$ to $t$. Using the lengths of the first and last paths as given in Lemmas \ref{lem:bigger_q_shorter} and \ref{lem:smaller_c_shorter}, the length of the replacement path is $3 + \sum_{i=0}^{k-1} (2\mathcal{B}^{k-i} + 2\mathcal{B}^{k-i-1} + \mathcal{B} - 1)$.

Assume that there is a replacement path avoiding $S_u$ of length $3 + \sum_{i=0}^{k-1} (2\mathcal{B}^{k-i} + 2\mathcal{B}^{k-i-1} + \mathcal{B} - 1)$. Let $q_w$ be the first vertex in $Q$ used by the replacement path, and let $c_v$ be the last vertex in $C$ used by the replacement path, so the path can be divided into three sections: a path from $s$ to $q_w$, a path from $q_w$ to $c_v$, and a path from $c_v$ to $t$.

By construction, the shortest possible path from a vertex in $Q$ to a vertex in $C$ has length $3$, so we get that $d_{G_b \setminus S_u}(q_w, c_v) \geq 3$.

By Lemma~\ref{lem:closest_q_c}, $d_{G_b \setminus S_u}(s, q_w) \ge d_{G_b \setminus S_u}(s, q_u)$, and the equality holds if and only if $w = u$. 
Similarly $d_{G_b \setminus S_u}(q_v, t) \ge d_{G_b \setminus S_u}(q_u, t)$, and the equality holds if and only if $v = u$. 

Putting it together we get that $d_{G_b \setminus S_u}(s, t) = d_{G_b \setminus S_u}(s, q_w) + d_{G_b \setminus S_u}(q_w, c_v) + d_{G_b \setminus S_u}(c_v, t) \geq 3 + d_{G_b \setminus S_u}(s, q_u) + d_{G_b \setminus S_u}(c_u, t) = 3 + \sum_{i=0}^{k-1} (2\mathcal{B}^{k-i} + 2\mathcal{B}^{k-i-1} + \mathcal{B} - 1)$. Therefore, it is only possible for the replacement path to be of length $3 + \sum_{i=0}^{k-1} (2\mathcal{B}^{k-i} + 2\mathcal{B}^{k-i-1} + \mathcal{B} - 1)$ if all of the components are minimized, which is only the case if $q_w = q_u$, $c_v = c_u$, and there is a length-3 path in $G_b$ from $q_u$ to $c_u$. A path from $q_u$ to $c_u$ of length 3 is only possible if there is a triangle involving $v_{bL+u}$ in $G$, therefore if the replacement path has length $3 + \sum_{i=0}^{k-1} (2\mathcal{B}^{k-i} + 2\mathcal{B}^{k-i-1} + \mathcal{B} - 1)$ then there must be a triangle involving $v_{bL+u}$ in $G$.
\end{proof}

Via Theorem \ref{thm:short_path_is_triangle}, we can detect if a particular vertex in $V_b$ is a member of a triangle by doing a $k$-fault distance sensitivity query. Therefore, we can detect if $G$ has a triangle by creating $\lceil n/L \rceil$ graphs, one for each subset $V_b$, and in each graph running our algorithm to create the $k$-fault DSO with fixed source and target, and making $L$ queries to this data structure.

\subsection{Implied Lower Bounds for Pre-Processing and Query Times}

Recall Theorem~\ref{thm:intro_lower_bound}:

\lowerbound*
\begin{proof}

Suppose that there is a data structure that can pre-process any
directed unweighted $n$-vertex graph $G$ and fixed vertices $s,t$ where $d_G(s, t) = p$ in $T(n, p)$ time, and can then answer $k$-fault distance sensitivity queries between $s$ and $t$ in $Q(n, p)$ time. Given a Triangle Detection instance on an $n$-vertex graph, the constructed graph  in the previous section has $\Theta(n + L^{(k+1)/k})$ vertices  and an original shortest path of length $\Theta(L^{1/k})$. Let us set $L$ such that $L^{(k+1)/k} = \Theta(n)$, which makes it so that each graph has $\Theta(n)$ vertices and an original shortest path with $\Theta(n^{1/(k+1)})$ edges. Therefore, the runtime for Triangle Detection using our reduction is
\begin{center}
    $\tO(n^{1/(k+1)}T(n, n^{1/(k+1)}) + nQ(n, n^{1/(k+1)})).$
\end{center}

If there was a combinatorial data structure for a directed unweighted $k$-fault DSO with fixed source and sink such that $T(n, n^{1/(k+1)}) = \tO(n^{2 + k/(k+1) - \epsilon})$ and $Q(n, n^{1/(k+1)}) = \tO(n^{2 - \epsilon})$ for $\epsilon > 0$, then we could use this data structure to solve Triangle Detection in $\tO(n^{3 - \epsilon})$ time. This concludes our proof of Theorem \ref{thm:intro_lower_bound}. 
\end{proof}

In the case of $k=2$, Theorem \ref{thm:intro_lower_bound} implies a conditional lower bound for unweighted $2$FRP. Suppose there is an $\tO(n^{8/3-\epsilon})$ time algorithm for $2$FRP for $\epsilon > 0$, then a data structure could run this algorithm in $\tO(n^{8/3-\epsilon})$ time during pre-processing, and answer $2$-fault distance sensitivity queries between $s$ and $t$ in $\tO(1)$ time (way faster than $\tO(n^{2-\epsilon})$). Therefore, assuming that BMM (and thus Triangle Detection) can not be solved in truly subcubic time using a combinatorial algorithm, any combinatorial algorithm for unweighted $2$FRP requires $n^{8/3-o(1)}$ time. 

An interesting yet strange fact about our reduction is that it specifically applies to  instances where the original shortest path only has $\Theta(n^{1/(k+1)})$ edges, instead of the $\Theta(n)$ edges it could have in the worst case.
Therefore, our lower bound even holds for  such restricted graphs.

\bibliographystyle{alpha}
\bibliography{source}

\begin{thebibliography}{BNWN22}

\bibitem[ACC19]{alon2019deterministic}
Noga Alon, Shiri Chechik, and Sarel Cohen.
\newblock Deterministic combinatorial replacement paths and distance
  sensitivity oracles.
\newblock In {\em Proc. 46th International Colloquium on Automata, Languages,
  and Programming (ICALP)}, 2019.

\bibitem[AGM97]{AlonGM97}
Noga Alon, Zvi Galil, and Oded Margalit.
\newblock On the exponent of the all pairs shortest path problem.
\newblock {\em J. Comput. Syst. Sci.}, 54(2):255--262, 1997.

\bibitem[AGMN92]{AlonWitness}
Noga Alon, Zvi Galil, Oded Margalit, and Moni Naor.
\newblock Witnesses for boolean matrix multiplication and for shortest paths.
\newblock In {\em Proc. 33rd Annual Symposium on Foundations of Computer
  Science}, pages 417--426, 1992.

\bibitem[AV21]{AVW21}
Josh Alman and Virginia {Vassilevska Williams}.
\newblock A refined laser method and faster matrix multiplication.
\newblock In {\em Proc. 2021 ACM-SIAM Symposium on Discrete Algorithms (SODA)},
  pages 522--539, 2021.

\bibitem[Ber10]{bernstein2010nearly}
Aaron Bernstein.
\newblock A nearly optimal algorithm for approximating replacement paths and
  {$k$} shortest simple paths in general graphs.
\newblock In {\em Proc. 21st Annual ACM-SIAM Symposium on Discrete Algorithms
  (SODA)}, pages 742--755, 2010.

\bibitem[BG04]{bhosle2004replacement}
Amit~M. Bhosle and Teofilo~F. Gonzalez.
\newblock Replacement paths for pairs of shortest path edges in directed
  graphs.
\newblock In {\em Proc. 16th IASTED International Conference on Parallel and
  Distributed Computing and Systems (PDCS)}, 2004.

\bibitem[Bho05]{bhosle2005improved}
Amit~M. Bhosle.
\newblock Improved algorithms for replacement paths problems in restricted
  graphs.
\newblock {\em Oper. Res. Lett.}, 33(5):459--466, 2005.

\bibitem[BK09]{bernstein2009nearly}
Aaron Bernstein and David Karger.
\newblock A nearly optimal oracle for avoiding failed vertices and edges.
\newblock In {\em Proc. 41st Annual ACM Symposium on Theory of Computing
  (STOC)}, pages 101--110, 2009.

\bibitem[BNWN22]{bernstein2022negative}
Aaron Bernstein, Danupon Nanongkai, and Christian Wulff-Nilsen.
\newblock Negative-weight single-source shortest paths in near-linear time.
\newblock {\em arXiv preprint arXiv:2203.03456}, 2022.

\bibitem[CC20]{chechik2020distance}
Shiri Chechik and Sarel Cohen.
\newblock Distance sensitivity oracles with subcubic preprocessing time and
  fast query time.
\newblock In {\em Proc. 52nd Annual ACM SIGACT Symposium on Theory of Computing
  (STOC)}, pages 1375--1388, 2020.

\bibitem[CM20]{chechikSSRP}
Shiri Chechik and Ofer Magen.
\newblock {Near Optimal Algorithm for the Directed Single Source Replacement
  Paths Problem}.
\newblock In {\em Proc. 47th International Colloquium on Automata, Languages,
  and Programming (ICALP)}, volume 168 of {\em Leibniz International
  Proceedings in Informatics (LIPIcs)}, pages 81:1--81:17, 2020.

\bibitem[DP09]{DuanPettieDualFailureDSO}
Ran Duan and Seth Pettie.
\newblock Dual-failure distance and connectivity oracles.
\newblock In {\em Proc. 20th Annual ACM-SIAM Symposium on Discrete Algorithms
  (SODA)}, pages 506--515, 2009.

\bibitem[DR22]{duan2021maintaining}
Ran Duan and Hanlin Ren.
\newblock Maintaining exact distances under multiple edge failures.
\newblock In {\em Proc. the 54th Annual {ACM} {SIGACT} Symposium on Theory of
  Computing (STOC)}, page 1093–1101, 2022.

\bibitem[DTCR08]{demetrescu2008oracles}
Camil Demetrescu, Mikkel Thorup, Rezaul~Alam Chowdhury, and Vijaya
  Ramachandran.
\newblock Oracles for distances avoiding a failed node or link.
\newblock {\em SIAM J. Comput.}, 37(5):1299--1318, 2008.

\bibitem[EPR10]{emek2010near}
Yuval Emek, David Peleg, and Liam Roditty.
\newblock A near-linear-time algorithm for computing replacement paths in
  planar directed graphs.
\newblock {\em ACM Trans. Algorithms}, 6(4):1--13, 2010.

\bibitem[FT87]{fredman1987fibonacci}
Michael~L. Fredman and Robert~Endre Tarjan.
\newblock Fibonacci heaps and their uses in improved network optimization
  algorithms.
\newblock {\em J. {ACM}}, 34(3):596--615, 1987.

\bibitem[GK17]{GuptaK17}
Manoj Gupta and Shahbaz Khan.
\newblock Multiple source dual fault tolerant {BFS} trees.
\newblock In {\em Proc. 44th International Colloquium on Automata, Languages,
  and Programming (ICALP)}, volume~80 of {\em LIPIcs}, pages 127:1--127:15,
  2017.

\bibitem[GL09]{gotthilf2009improved}
Zvi Gotthilf and Moshe Lewenstein.
\newblock Improved algorithms for the {$k$} simple shortest paths and the
  replacement paths problems.
\newblock {\em Information Processing Letters}, 109(7):352--355, 2009.

\bibitem[GPVX21]{gu_et_al}
Yuzhou Gu, Adam Polak, Virginia {Vassilevska Williams}, and Yinzhan Xu.
\newblock {Faster Monotone Min-Plus Product, Range Mode, and Single Source
  Replacement Paths}.
\newblock In {\em Proc. 48th International Colloquium on Automata, Languages,
  and Programming (ICALP)}, volume 198 of {\em Leibniz International
  Proceedings in Informatics (LIPIcs)}, pages 75:1--75:20, 2021.

\bibitem[GR21]{gurenDSO}
Yong Gu and Hanlin Ren.
\newblock {Constructing a Distance Sensitivity Oracle in $O(n^{2.5794} M)$
  Time}.
\newblock In {\em Proc. 48th International Colloquium on Automata, Languages,
  and Programming (ICALP)}, volume 198 of {\em Leibniz International
  Proceedings in Informatics (LIPIcs)}, pages 76:1--76:20, 2021.

\bibitem[GV12]{GrandoniW12}
Fabrizio Grandoni and Virginia {Vassilevska Williams}.
\newblock Improved distance sensitivity oracles via fast single-source
  replacement paths.
\newblock In {\em Proc. 53rd Annual {IEEE} Symposium on Foundations of Computer
  Science (FOCS)}, pages 748--757, 2012.

\bibitem[GV20]{GrandoniWilliamsSingleFailureDSO}
Fabrizio Grandoni and Virginia {Vassilevska Williams}.
\newblock Faster replacement paths and distance sensitivity oracles.
\newblock {\em ACM Trans. Algorithms}, 16(1):15:1--15:25, December 2020.

\bibitem[HS01]{VickreyPricingShortestPaths}
John Hershberger and Subhash Suri.
\newblock Vickrey prices and shortest paths: What is an edge worth?
\newblock In {\em Proc. 42nd IEEE Symposium on Foundations of Computer Science
  (FOCS)}, pages 252--259, 2001.

\bibitem[HT84]{harel1984fast}
Dov Harel and Robert~Endre Tarjan.
\newblock Fast algorithms for finding nearest common ancestors.
\newblock {\em {SIAM} J. Comput.}, 13(2):338--355, 1984.

\bibitem[Joh77]{johnson1977efficient}
Donald~B Johnson.
\newblock Efficient algorithms for shortest paths in sparse networks.
\newblock {\em J. {ACM}}, 24(1):1--13, 1977.

\bibitem[KMW10]{klein2010shortest}
Philip~N Klein, Shay Mozes, and Oren Weimann.
\newblock Shortest paths in directed planar graphs with negative lengths: A
  linear-space {$O (n \log^2 n)$}-time algorithm.
\newblock {\em ACM Trans. Algorithms}, 6(2):1--18, 2010.

\bibitem[LG12]{le2012faster}
Fran{\c{c}}ois Le~Gall.
\newblock Faster algorithms for rectangular matrix multiplication.
\newblock In {\em Proc. 53rd Annual IEEE Symposium on Foundations of Computer
  Science (FOCS)}, pages 514--523, 2012.

\bibitem[LG14]{LeGall14}
Fran\c{c}ois Le~Gall.
\newblock Powers of tensors and fast matrix multiplication.
\newblock In {\em Proc. 39th International Symposium on Symbolic and Algebraic
  Computation (ISSAC)}, pages 296--303, 2014.

\bibitem[LGU18]{LU18}
Fran{\c{c}}ois Le~Gall and Florent Urrutia.
\newblock Improved rectangular matrix multiplication using powers of the
  {Coppersmith-Winograd} tensor.
\newblock In {\em Proc. 29th Annual ACM-SIAM Symposium on Discrete Algorithms
  (SODA)}, pages 1029--1046, 2018.

\bibitem[LL14]{lee2014replacement}
Cheng-Wei Lee and Hsueh-I Lu.
\newblock Replacement paths via row minima of concise matrices.
\newblock {\em SIAM J. Discrete Math.}, 28(1):206--225, 2014.

\bibitem[MMG89]{malik1989k}
Kavindra Malik, Ashok~K Mittal, and Santosh~K Gupta.
\newblock The {$k$} most vital arcs in the shortest path problem.
\newblock {\em Operations Research Letters}, 8(4):223--227, 1989.

\bibitem[NPW01]{nardelli2001faster}
Enrico Nardelli, Guido Proietti, and Peter Widmayer.
\newblock A faster computation of the most vital edge of a shortest path.
\newblock {\em Information Processing Letters}, 79(2):81--85, 2001.

\bibitem[NR01]{nisan2001algorithmic}
Noam Nisan and Amir Ronen.
\newblock Algorithmic mechanism design.
\newblock {\em Games Econ. Behav.}, 35(1-2):166--196, 2001.

\bibitem[PP16]{ParterP16}
Merav Parter and David Peleg.
\newblock Sparse fault-tolerant {BFS} structures.
\newblock {\em {ACM} Trans. Algorithms}, 13(1):11:1--11:24, 2016.

\bibitem[Ren22]{RenImprovedDSO}
Hanlin Ren.
\newblock Improved distance sensitivity oracles with subcubic preprocessing
  time.
\newblock {\em J. Comput. Syst. Sci.}, 123:159--170, 2022.

\bibitem[RZ12]{roditty2005replacement}
Liam Roditty and Uri Zwick.
\newblock Replacement paths and $k$ simple shortest paths in unweighted
  directed graphs.
\newblock {\em ACM Trans. Algorithms}, 8(4), October 2012.

\bibitem[Sei95]{seidel1995all}
Raimund Seidel.
\newblock On the all-pairs-shortest-path problem in unweighted undirected
  graphs.
\newblock {\em J. Comput. Syst. Sci}, 51(3):400--403, 1995.

\bibitem[SZ99]{shoshan1999all}
Avi Shoshan and Uri Zwick.
\newblock All pairs shortest paths in undirected graphs with integer weights.
\newblock In {\em Proc. 40th Annual IEEE Symposium on Foundations of Computer
  Science (FOCS)}, pages 605--614, 1999.

\bibitem[{Vas}11]{williams2011faster}
Virginia {Vassilevska Williams}.
\newblock Faster replacement paths.
\newblock In {\em Proc. 22nd Annual ACM-SIAM Symposium on Discrete Algorithms
  (SODA)}, pages 1337--1346, 2011.

\bibitem[{Vas}12]{Vassilevska12}
Virginia {Vassilevska Williams}.
\newblock Multiplying matrices faster than {C}oppersmith-{W}inograd.
\newblock In {\em Proc. 44th Annual ACM Symposium on Theory of Computing
  (STOC)}, pages 887--898, 2012.

\bibitem[vdBS19]{brandBatchUpdates}
Jan van~den Brand and Thatchaphol Saranurak.
\newblock Sensitive distance and reachability oracles for large batch updates.
\newblock In {\em Proc. 60th Annual IEEE Symposium on Foundations of Computer
  Science (FOCS)}, pages 424--435, 2019.

\bibitem[VW10]{vw10}
Virginia {Vassilevska Williams} and Ryan Williams.
\newblock Subcubic equivalences between path, matrix and triangle problems.
\newblock In {\em Proc. 51th Annual {IEEE} Symposium on Foundations of Computer
  Science (FOCS)}, pages 645--654, 2010.

\bibitem[VW18]{williams2018subcubic}
Virginia {Vassilevska Williams} and Ryan Williams.
\newblock Subcubic equivalences between path, matrix, and triangle problems.
\newblock {\em J.~ACM}, 65(5):1--38, 2018.

\bibitem[Wil14]{Williams14a}
Ryan Williams.
\newblock Faster all-pairs shortest paths via circuit complexity.
\newblock In {\em Proc. 46th Annual ACM Symposium on Theory of Computing
  (STOC)}, pages 664--673, 2014.

\bibitem[Wil18]{Williams18}
R.~Ryan Williams.
\newblock Faster all-pairs shortest paths via circuit complexity.
\newblock {\em {SIAM} J. Comput.}, 47(5):1965--1985, 2018.

\bibitem[WY10]{weimann2010replacement}
Oren Weimann and Raphael Yuster.
\newblock Replacement paths via fast matrix multiplication.
\newblock In {\em Proc. 51st Annual IEEE Symposium on Foundations of Computer
  Science (FOCS)}, pages 655--662, 2010.

\bibitem[WY13]{WeimannYusterFDSO}
Oren Weimann and Raphael Yuster.
\newblock Replacement paths and distance sensitivity oracles via fast matrix
  multiplication.
\newblock {\em ACM Trans. Algorithms}, 9(2), March 2013.

\bibitem[Zwi02]{Zwick02}
Uri Zwick.
\newblock All pairs shortest paths using bridging sets and rectangular matrix
  multiplication.
\newblock {\em J.~ACM}, 49(3):289--317, 2002.

\end{thebibliography}

\appendix
\section{An Issue in a Previous Work}
\label{append:counter}
Bhosle and Gonzalez \cite{bhosle2004replacement} claimed an $O(n^3)$ time algorithm for $2$FRP in the special case where both failed edges are on the original $s$ to $t$ shortest path. However, we found that their approach doesn't quite work as written and it is unclear if the approach can be made to work. 

They use $u_0, \ldots, u_p$ to denote vertices on the original $s$ to $t$ shortest path in order (in particular $s = u_0$ and $t = u_p$). 
They use $e_i$ for the first failed edge, $f = (u_j, u_{j+1})$ for the second failed edge that is strictly after $e_i$ on $\pi_G(s, t)$. Let $T_s^i$ be a shortest path tree rooted at $s$ in $G \setminus \{e_i\}$, and they consider $T_s^i \setminus \{f\}$. The deletion of $f$ breaks $T_s^i$ to two parts, one containing $s$ and one containing $t$ (otherwise $d_{G \setminus \{e_i, f\}}(s, t) = d_{G \setminus \{e_i\}}(s, t)$ is easy to compute), and they use $V_{s | \{e_i, f\}}$ to denote the set of vertices in the same part as $s$ in $T_s^i \setminus \{f\}$. Finally, they define $\mathcal{R}^v_{e_i, f}(s, t)$ to be the shortest $s$ to $t$ path in $G \setminus \{e_i, f\}$ whose last vertex in $V_{s | \{e_i, f\}}$ is $v$ and use $|\mathcal{R}^v_{e_i, f}(s, t)|$ to denote its length. 

The following equation (Equation (4) in \cite{bhosle2004replacement}) is a key equation in their algorithm (slightly changed to adapt to our notation):
$$|\mathcal{R}^{v}_{e_i, f}(s, t)| = \min_{k=j+1}^p \left\{d_{G \setminus \{e_i\}}(s, v) + d_{G \setminus \pi_G(s, t)}(v, u_k) + d_{G \setminus \{e_i\}}(u_k, t) \right\}.$$
The next step in their analysis is to take the minimum value of $|\mathcal{R}^{v}_{e_i, f}(s, t)|$ over all $v \in V_{s | \{e_i, f\}}$ to obtain the replacement path distance. 

It is unclear to us why the part from $v$ to $u_k$  has to necessarily avoid $\pi_G(s, t)$ and we could not find a proof for this part of the equation in \cite{bhosle2004replacement}. 

We provide a small counter-example in which the $v$ to $u_k$ part of the $2$-fault replacement path has to use part of the original shortest path as shown in Figure~\ref{fig:countera}, Figure~\ref{fig:counterb} and Figure~\ref{fig:counterc}.

In this example, $V_{s | \{e_i, f\}} = \{s, 3, 4\}$, and the only possible value for $u_k$ is $t$. We can try all these combinations and evaluate $d_{G \setminus \{e_i\}}(s, v) + d_{G \setminus \pi_G(s, t)}(v, u_k) + d_{G \setminus \{e_i\}}(u_k, t)$. No matter which of these values we use, $d_{G \setminus \pi_G(s, t)}(v, u_k)$ is always $\infty$ since there is no path from $s, 3$ or $4$ to $t$ avoiding the entire original shortest path. However, the  $2$-fault replacement has length $6$ as shown in Figure~\ref{fig:counterc}.

\begin{subfigures}
\begin{figure}[ht]
    \centering
    \begin{tikzpicture}
    	\node at (0,0) [circle, draw=black, minimum size=20pt] (s){$s$};
    	\node at (2,0) [circle, draw=black, minimum size=20pt] (1){$1$};
    	\node at (4,0) [circle, draw=black, minimum size=20pt] (2){$2$};
    	\node at (6,0) [circle, draw=black, minimum size=20pt] (3){$3$};
    	\node at (8,0) [circle, draw=black, minimum size=20pt] (4){$4$};
    	\node at (10,0) [circle, draw=black, minimum size=20pt] (t){$t$};

    \draw [-stealth, line width = 1pt] (s) to[] node[label=above:$0$](){} (1);
    \draw [-stealth, line width = 1pt] (1) to[] node[label=above:$0$](){} (2);
    \draw [-stealth, line width = 1pt] (2) to[] node[label=above:$0$](){} (3);
    \draw [-stealth, line width = 1pt] (3) to[] node[label=above:$0$](){} (4);
    \draw [-stealth, line width = 1pt] (4) to[] node[label=above:$0$](){} (t);
    \draw [-stealth, line width = 1pt, bend left = 45] (s) to[] node[label=above:$2$](){} (3);
    \draw [-stealth, line width = 1pt, bend left = 45] (2) to[] node[label=above:$2$](){} (t);
    \draw [-stealth, line width = 1pt, bend left = 60] (t) to[] node[label=above:$1$](){} (1);
    \draw [-stealth, line width = 1pt, bend left = 45] (4) to[] node[label=above:$2$](){} (1);
            
	\begin{scope}[transparency group, opacity=0.3, text opacity=1, line width = 4pt, red, -stealth]
	    \draw [] (s) to[]  (1);
        \draw [] (1) to[]  (2);
        \draw [] (2) to[]  (3);
        \draw [] (3) to[]  (4);
        \draw [] (4) to[]  (t);
	\end{scope}
    \end{tikzpicture}
    \caption{The vertices and edges and the (unique) original shortest path of the example. }
    \label{fig:countera}
\end{figure}
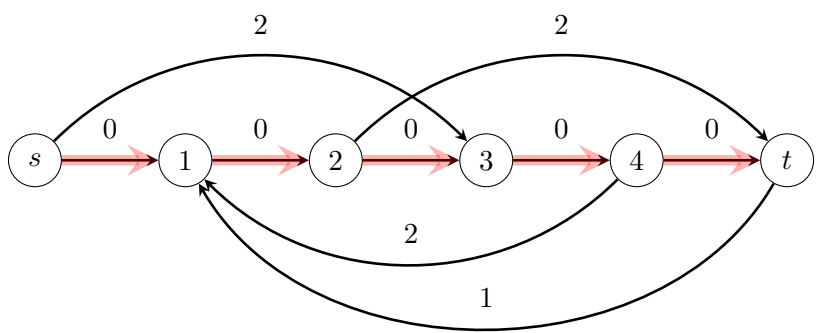

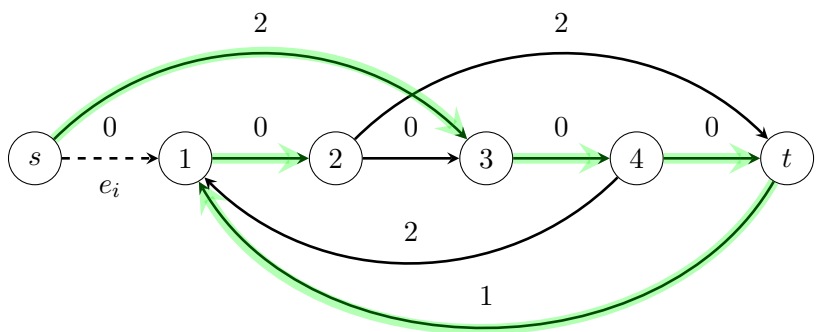
\begin{figure}[ht]
    \centering
    \begin{tikzpicture}
    	\node at (0,0) [circle, draw=black, minimum size=20pt] (s){$s$};
    	\node at (2,0) [circle, draw=black, minimum size=20pt] (1){$1$};
    	\node at (4,0) [circle, draw=black, minimum size=20pt] (2){$2$};
    	\node at (6,0) [circle, draw=black, minimum size=20pt] (3){$3$};
    	\node at (8,0) [circle, draw=black, minimum size=20pt] (4){$4$};
    	\node at (10,0) [circle, draw=black, minimum size=20pt] (t){$t$};

    \draw [-stealth, dashed, line width = 1pt] (s) to[] node[label=above:$0$, label=below:$e_i$](){} (1);
    \draw [-stealth, line width = 1pt] (1) to[] node[label=above:$0$](){} (2);
    \draw [-stealth, line width = 1pt] (2) to[] node[label=above:$0$](){} (3);
    \draw [-stealth, line width = 1pt] (3) to[] node[label=above:$0$](){} (4);
    \draw [-stealth, line width = 1pt] (4) to[] node[label=above:$0$](){} (t);
    \draw [-stealth, line width = 1pt, bend left = 45] (s) to[] node[label=above:$2$](){} (3);
    \draw [-stealth, line width = 1pt, bend left = 45] (2) to[] node[label=above:$2$](){} (t);
    \draw [-stealth, line width = 1pt, bend left = 60] (t) to[] node[label=above:$1$](){} (1);
    \draw [-stealth, line width = 1pt, bend left = 45] (4) to[] node[label=above:$2$](){} (1);
            
	\begin{scope}[transparency group, opacity=0.3, text opacity=1, line width = 4pt, green, -stealth]
	    \draw [bend left = 45] (s) to[]  (3);
        \draw [] (3) to[]  (4);
        \draw [] (4) to[]  (t);
        \draw [bend left = 60] (t) to[]  (1);
        \draw [] (1) to[]  (2);
	\end{scope}
    \end{tikzpicture}
    \caption{The first edge we cut and the (unique) shortest path tree $T_s^i$.}
    \label{fig:counterb}
\end{figure}

\begin{figure}[ht]
    \centering
    \begin{tikzpicture}
    	\node at (0,0) [circle, draw=black, minimum size=20pt] (s){$s$};
    	\node at (2,0) [circle, draw=black, minimum size=20pt] (1){$1$};
    	\node at (4,0) [circle, draw=black, minimum size=20pt] (2){$2$};
    	\node at (6,0) [circle, draw=black, minimum size=20pt] (3){$3$};
    	\node at (8,0) [circle, draw=black, minimum size=20pt] (4){$4$};
    	\node at (10,0) [circle, draw=black, minimum size=20pt] (t){$t$};

    \draw [-stealth, dashed, line width = 1pt] (s) to[] node[label=above:$0$, label=below:$e_i$](){} (1);
    \draw [-stealth, line width = 1pt] (1) to[] node[label=above:$0$](){} (2);
    \draw [-stealth, line width = 1pt] (2) to[] node[label=above:$0$](){} (3);
    \draw [-stealth, line width = 1pt] (3) to[] node[label=above:$0$](){} (4);
    \draw [-stealth, dashed, line width = 1pt] (4) to[] node[label=above:$0$, label=below:$f$](){} (t);
    \draw [-stealth, line width = 1pt, bend left = 45] (s) to[] node[label=above:$2$](){} (3);
    \draw [-stealth, line width = 1pt, bend left = 45] (2) to[] node[label=above:$2$](){} (t);
    \draw [-stealth, line width = 1pt, bend left = 60] (t) to[] node[label=above:$1$](){} (1);
    \draw [-stealth, line width = 1pt, bend left = 45] (4) to[] node[label=above:$2$](){} (1);
            
	\begin{scope}[transparency group, opacity=0.3, text opacity=1, line width = 4pt, blue, -stealth]
	    \draw [bend left = 45] (s) to[]  (3);
        \draw [] (3) to[]  (4);
        \draw [bend left = 45] (4) to[]  (1);
        \draw [] (1) to[]  (2);
        \draw [bend left = 45] (2) to[]  (t);
	\end{scope}
    \end{tikzpicture}
    \caption{The second edge we cut and the (unique) $2$-fault replacement path.}
    \label{fig:counterc}
\end{figure}
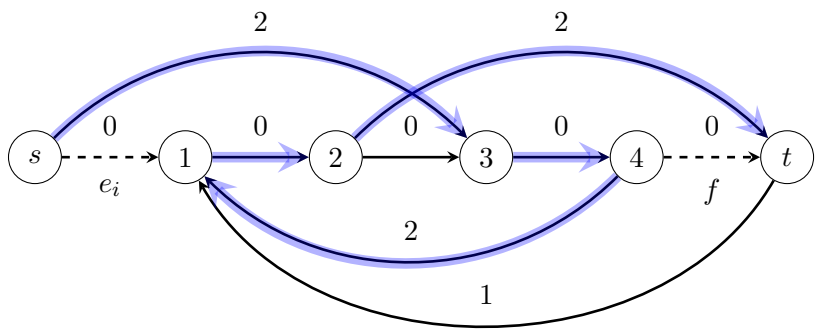
\end{subfigures}

\end{document}